\documentclass[12pt]{article}

\usepackage{graphics,graphicx,fullpage,natbib,multirow}
\usepackage{amsmath,amssymb,verbatim,epsfig}
\usepackage{hyperref}
\usepackage[dvipsnames,usenames]{color}
\usepackage[normalem]{ulem}
\usepackage{soul}

\newtheorem{proposition}{Proposition}

\newtheorem{defin}{\bf Definition}

\newenvironment{proof}{\noindent{\bf Proof}}{$\diamond$}

\def\ga{\mbox{Ga}}

\def\be{\mbox{Be}}
\def\dir{\mbox{Dir}}
\def\ddir{\mbox{DDir}}

\def\mul{\mbox{Mult}}
\def\no{\mbox{N}}

\def\un{\mbox{Unif}}

\def\E{\mbox{E}}
\def\V{\mbox{Var}}
\def\Cov{\mbox{Cov}}
\def\Cor{\mbox{Corr}}

\def\P{\mbox{P}}

\def\rest{\mbox{rest}}

\def\ba{{\bf a}}

\def\bj{{\bf j}}
\def\bk{{\bf k}}

\def\bp{{\bf p}}

\def\bu{{\bf u}}

\def\bz{{\bf z}}

\def\bU{{\bf U}}

\def\bZ{{\bf Z}}

\def\bone{{\bf 1}}
\def\simind{\stackrel{\mbox{\scriptsize{ind}}}{\sim}}

\newcommand{\balpha}{\boldsymbol{\alpha}}

\newcommand{\bfeta}{\boldsymbol{\eta}}

\newcommand{\blambda}{\boldsymbol{\lambda}}

\newcommand{\bomega}{\boldsymbol{\omega}}

\newcommand{\bpartial}{\boldsymbol{\partial}}
\newcommand{\bpi}{\boldsymbol{\pi}}
\newcommand{\btheta}{\boldsymbol{\theta}}

\newcommand{\MC}{\mathcal{M}}

\newcommand{\JB}{\mathbb{J}}

\newcommand{\NB}{\mathbb{N}}
\newcommand{\RB}{\mathbb{R}}

\begin{document}

\baselineskip=24pt

\title{\bf Multivariate temporal dependence via mixtures of rotated copulas}
\author{{\sc Ruyi Pan, Luis E. Nieto-Barajas, Radu V. Craiu} \\[2mm]
{\sl University of Toronto \& ITAM} \\[2mm]
{\small {\tt ruyi.pan@mail.utoronto.ca, luis.nieto@itam.mx \& radu.craiu@utoronto.ca}} \\}
\date{}
\maketitle

\begin{abstract}
Parametric copula families have been known to flexibly capture various dependence patterns, e.g., either positive or negative dependence in either the lower or upper tails of bivariate distributions. In this paper, our objective is to construct a model that is adaptable enough to capture several of these features simultaneously in $m$ dimensions. We propose a mixture of $2^m$ rotations of a parametric copula that can achieve this goal. We illustrate the construction using the Clayton family but the concept is general and can be applied to other families. In order to include dynamic dependence regimes, the approach is extended to a time-dependent sequence of mixture copulas in which the mixture probabilities are allowed to evolve in time via a moving average and seasonal types of relationship. The properties of the proposed model and its performance are examined using simulated and real data sets. 
\end{abstract}

\vspace{0.2in} \noindent {\sl Keywords}: Bayesian inference, dynamic dependence models, moving average process, seasonal model, time varying copulas.

\section{Introduction}
\label{sec:intro}

Copulas have emerged in recent years as viable tools for modelling dependence in non-standard situations in which the usual ``suspects'' such as multivariate Gaussian, Student or Wishart distributions, are not appropriate. Besides being an important tool for methodological development and having considerable potential for applications,  copulas have gained popularity due to several features that are desirable to a statistician. Allowing the separation of modeling effort for the marginal models and the dependence structure continues to rank high, but so is the flexibility it exhibits in capturing dependence patterns using parametric families, especially for bivariate data. In higher dimensions, this flexibility is expressed through the use of   C- or D- vine copulas that make efficient use of bivariate conditional copulas to flexibly model multivariate ones \citep{czado2022vine}. 

In the analysis of extreme value data, it is often desirable to measure the tail dependence in a vector. Some copulas are able to capture tail dependence for instance, the Clayton/Gumbel copula with positive $\theta$ parameter exhibits upper/lower tail dependence \citep{nelsen:06}.  However, while one can identify copula families able to capture a bivariate distribution's various patterns of lower or upper tail dependence, be they positive or negative, there is interest for developing more flexible parametric families that can capture several such patterns simultaneously in higher dimensions. Most attempts have been done in two dimensions. \cite{joe:97} proposed the Joe-Clayton Archimedean copula, which is able to model lower and upper tail dependence. Later, \cite{hu:06} proposed the use of a three-component mixture of Gaussian, Gumbel and survival Gumbel copulas that allows for no,  lower or upper tail dependence. Alternatively, the survival Gumbel copula was replaced by the Clayton copula in  \cite{liu&al:19}.  

A bivariate survival copula is a $180^\circ$ rotation of a bivariate copula. The advantage of such rotation is that, tail properties are reflected with respect to the $v=1-u$ line in the unit square. However, other degrees of rotations, like $90^\circ$ and $270^\circ$, are also possible. For instance, \cite{klein&al:20} considered all four rotations of a Clayton copula in two dimensions and developed model selection criteria for selecting the correct one. On the other hand, \cite{oh&patton:16} proposed a jointly symmetric copula with an equally weighted mixture of the four way rotations of a copula with the same parameter. In a similar fashion, \cite{smith:15} proposed two rotation mixtures of $0^\circ-180^\circ$ and $90^\circ-270^\circ$, both with the same dependence parameter, in order to tackle both serial and cross sectional dependence.

In time series analysis, copulas have been used to capture serial dependence. For example, \cite{hafner&manner:12} proposed a dynamic copula model in which copula parameters follow an autoregressive process, and \cite{loaiza&al:18} factored the joint density of a unit vector $\bu=(u_1,\ldots,u_T)$ as $c(\bu)=\prod_{t=1}^T f(u_t\mid u_{t-1})$ assuming Markov conditional distributions. They further used a bivariate copula to model the transitions, i.e. $f(u_t\mid u_{t-1})=c_2(u_{t-1},u_t)$. Specifically, in the bivariate copula they use a mixture of rotations of degrees $0^\circ-90^\circ$.

In this paper, we first generalise the concept of bivariate tail dependence to any corner of the unit $m$-dimensional hypercube and propose a flexible copula that is able to capture multiple types of tail dependence. Our goal is achieved by mixing all $2^m$ rotated versions of the multivariate copula. Furthermore, the $2^m$-dimensional mixture weights $\bpi_t$ are all different and are allowed to change over time through a moving average type process of order $q$ and a seasonal component of order $p$ that maintains the marginal distribution invariant over time. For each of the mixture copula components, dependence parameters are allowed to vary over time and are, a priori, assumed to be exchangeable.  Our context is not a traditional time series problem in the sense that we do not follow a single individual through time, but rather we monitor the dependence of several individuals in time.

The rest of the paper is organized as follows.  In Section \ref{sec:mot}, we provide the motivation of the paper and the required notation. Section \ref{sec:taild} contains a discussion about   tail dependence. In Section \ref{sec:model}, we define our mixture model for a specific time and define the time dependent mixture weights and association parameters. Section \ref{sec:bayesian} provides the prior and posterior distributions that are required to conduct a Bayesian analysis of the model.  An illustration of the model's performance is reported in Section \ref{sec:numerical}. Section \ref{sec:conclusion} contains conclusions and directions for future work. 

\section{Motivation and Notation}
\label{sec:mot}

The emergence of copulas as important tools for modeling dependence has its origins in Sklar's paper \citep{sklar:59} which demonstrated that the link between any continuous multivariate distribution and its marginals can be achieved via a unique copula $C:[0,1]^m\to[0,1]$. The latter is a multivariate distribution with uniform marginals on the interval $[0,1]$. Specifically, if $F$ is a multivariate cumulative distribution function (CDF) with marginal CDFs $F_1,\ldots,F_m$, then $F(x_1,\ldots,x_m)=C\{F_1(x_1),\ldots,F_m(x_m)\}$. Additionally, the copula function can be obtained as $C(u_1,\ldots,u_m)=F\{F_1^{-1}(u_1),\ldots,F_m^{-1}(u_m)\}$, where $F_j^{-1}$ for $j=1,\ldots,m$ are the marginal inverse CDFs or quantile functions. 

There is a large body of literature devoted to identifying parametric copula families that are able to capture various dependence patterns in the tails \citep{joe:97}. For instance, in the analysis of extreme value theory, an important concept is that of dependence in the upper-right or lower-down quadrants of a joint bivariate distribution. This is quantified by the so-called upper and lower tail dependence coefficients \citep{frahm&al:05,joe:97}. 

Let $(X_1,X_2)$ be a bivariate vector with marginal CDFs, $F_1$ and, respectively, $F_2$, such that the joint CDF is given in terms of the copula $C$ as $F(x_1,x_2)=C(F_1(x_1),F_2(x_2))$. Tail dependence coefficients are defined as limits of the conditional probabilities that both variables are above an upper quantile of order $1-\nu$, or both variables are below a lower quantile of order $\nu$, as $\nu$ approaches zero. We denote 
$$\lambda_{11}=\lim_{\nu\to 0}\P\{X_1>F_1^{-1}(1-\nu)\mid X_2>F_2^{-1}(1-\nu)\}=\lim_{\nu\to 0}\P(U_1>1-\nu\mid U_2>1-\nu)$$
for the upper-right (upper-upper) corner, and
$$\lambda_{00}=\lim_{\nu\to 0}\P\{X_1\leq F_1^{-1}(\nu)\mid X_2\leq F_2^{-1}(\nu)\}=\lim_{\nu\to 0}\P(U_1\leq\nu\mid U_2\leq\nu)$$
for the lower-down (lower-lower) corner. The binary sub-indexes $0$ and $1$ stand for lower and upper, respectively. However, it is possible that both variables have co-movements in the opposite tails, that is, one variable has values in the upper quantile and the other in the lower quantile, or conversely. In this case the opposite tail dependencies are defined as
$$\lambda_{10}=\lim_{\nu\to 0}\P\{X_1>F_1^{-1}(1-\nu)\mid X_2\leq F_2^{-1}(\nu)\}=\lim_{\nu\to 0}\P(U_1>1-\nu\mid U_2\leq\nu)$$
for the upper-lower corner, and
$$\lambda_{01}=\lim_{\nu\to 0}\P\{X_1\leq F_1^{-1}(\nu)\mid X_2>F_2^{-1}(1-\nu)\}=\lim_{\nu\to 0}\P(U_1\leq\nu\mid U_2>1-\nu)$$
for the lower-upper corner. 

These four tail dependence coefficients can be written entirely in terms of the copula. It is straightforward to show that 
\begin{eqnarray}
\label{eq:tdc}
\lambda_{11}=\lim_{\nu\to 0}\frac{2\nu-1+C(1-\nu,1-\nu)}{\nu},\quad\lambda_{00}=\lim_{\nu\to 0}\frac{C(\nu,\nu)}{\nu},\\
\nonumber
\lambda_{10}=\lim_{\nu\to 0}\frac{\nu-C(1-\nu,\nu)}{\nu},\quad\lambda_{01}=\lim_{\nu\to 0}\frac{\nu-C(\nu,1-\nu)}{\nu}.
\end{eqnarray}

It is also well known \citep[e.g.][]{tjostheim&al:22} that the Clayton copula exhibits lower-lower tail dependence, whereas the Gumbel copula has upper-upper tail dependence. One way of defining copulas with the four types of tail dependence \eqref{eq:tdc} is by means of rotation as in \cite{klein&al:20}. It is easy to see from \eqref{eq:tdc} that for most copulas, the four tail dependence coefficients will be different. In the next section, we develop a mixture of copulas that allows identical tail dependence coefficients.

Before we proceed, let us introduce some notation. Let $\ga(\alpha,\beta)$ denote a gamma density with mean $\alpha/\beta$, $\be(\alpha,\beta)$ a beta density with mean $\alpha/(\alpha+\beta)$, $\no(\mu,\tau)$ a normal density with mean $\mu$ and precision $\tau$, $\dir(\balpha)$ a Dirichlet density with parameter vector $\balpha$, and $\mul(c,\bp)$ a multinomial density with total trials $c$ and probability vector $\bp$. The density evaluated at a specific point $x$, will be denoted using the notation for the density, e.g.  $\ga(x\mid\alpha,\beta)$, in the gamma case.

\section{Bivariate rotations}
\label{sec:taild}


Before we introduce the general case, consider for illustration the case $m=2$ in which  the unit square $[0,1]^2$ is divided into four quadrants as in Figure \ref{fig:rotations}. 
To define the 90-degree rotation, we consider the probability in quadrant II, $\P(U_1>u_1,U_2\leq u_2)=\P(U_2\leq u_2)-\P(U_1\leq u_1,U_2\leq u_2)$, which in terms of the copula becomes $u_2-C(u_1,u_2)$. Finally by making the transformation $U_1'=1-U_1$, we maintain the marginal uniformity in $U_1'$ and can obtain a new CDF (copula) of the form $$C_{10}(u_1,u_2)=\P(U_1'\le u_1,U_2\leq u_2)=u_2-C(1-u_1,u_2).$$ 

To define the $180^\circ$ rotation, we consider the following probability in  quadrant III, $\P(U_1>u_1,U_2>u_2)=1-\P(U_1\leq u_1)-\P(U_2\leq u_2)+\P(U_1\leq u_1,U_2\leq u_2)$, which in terms of the copula becomes $1-u_1-u_2+C(u_1,u_2)$. Again, making the transformation $U_1'=1-U_1$ and $U_2'=1-U_2$ we get a new CDF (copula) $$C_{11}(u_1,u_2)=\P(U_1'\leq u_1,U_2'\leq u_2)=u_1+u_2-1+C(1-u_1,1-u_2).$$ 

Lastly, to define the 270-degree rotation we consider the probability, $\P(U_1\leq u_1,U_2>u_2)=\P(U_1\leq u_1)-\P(U_1\leq u_1,U_2\leq u_2)$, which in terms of the copula can be written as $u_1-C(u_1,u_2)$. Making the transformation $U_2'=1-U_2$, we obtain the new CDF (copula) $$C_{01}(u_1,u_2)=\P(U_1\leq u_1,U_2'\leq u_2)=u_1-C(u_1,1-u_2).$$ 
For completeness, we denote the original, un-rotated, copula as $C_{00}(u_1,u_2)$. 

If a particular copula has lower-lower or upper-upper tail dependence, each of the four rotated versions of the copula will have the same tail dependence, but in different corners. To illustrate, consider the Clayton copula defined as $C(u_1,u_2)=\left(u_1^{-\theta}+u_2^{-\theta}-1\right)^{-1/\theta}$ for $\theta\geq-1$. If $\theta=0$ Clayton copula reduces to the independence copula, and for $\theta>0$ the Kendall's tau association coefficient is $\tau=\theta/(2+\theta)$ and there is a lower-lower tail dependence with coefficient $\lambda_{00}=2^{-1/\theta}$. 

The four rotations of the Clayton copula are as follows.
\begin{enumerate}
\item[(i)] $0$-degree rotation:
$$C_{00}(u_1,u_2)=\left(u_1^{-\theta}+u_2^{-\theta}-1\right)^{-1/\theta}$$
\item[(ii)] $90$- and $270$-degrees rotations: 
\begin{align}
\nonumber
C_{10}(u_1,u_2)&=u_2-\left\{(1-u_1)^{-\theta}+u_2^{-\theta}-1\right\}^{-1/\theta} \\
\nonumber
C_{01}(u_1,u_2)&=u_1-\left\{u_1^{-\theta}+(1-u_2)^{-\theta}-1\right\}^{-1/{\theta}} 
\end{align}
\item[(iii)] $180$-degree rotation: 
$$C_{11}(u_1,u_2)=u_1+u_2-1+\left\{(1-u_1)^{-\theta}+(1-u_2)^{-\theta}-1\right\}^{-1/\theta}.$$
\end{enumerate}

Using (i)--(iii) and \eqref{eq:tdc}, it is not difficult to prove that the tail dependence coefficients, when the tail (corner) corresponds to the rotation, are the same $\lambda_{00}(C_{00})=\lambda_{10}(C_{10})=\lambda_{01}(C_{01})=\lambda_{11}(C_{11})$ and are given by
\begin{equation}
\label{eq:tdc2}
\lambda_{jk}(C_{jk})=\lim_{\nu\to 0}\frac{\left(2\nu^{-\theta}-1\right)^{-1/\theta}}{\nu}=\lim_{\nu\to 0}(2-\nu^\theta)^{-1/\theta}=2^{-1/\theta},
\end{equation}
for $j,k\in\{0,1\}$ and for $\theta>0$. 
If the tail coefficient does not correspond to a rotation then it has the value zero.

\section{Multivariate extensions}
\label{sec:mult}

Let $\bU=(U_1,\ldots,U_m)$ be an $m$-multivariate random vector with support on the unit hyper-cube $[0,1]^m$, for $m\geq 2$. Let $C^{(m)}$ be an $m$-dimensional symmetric copula. Most dependence measures involve two random variables, like the Pearson and Spearman correlation coefficients and Kendall's tau. Additionally, as far as we are aware, the tail dependence is only defined for a pair of random variables. Since this last measure is defined in terms of a conditional probability, it can be generalised to more than two random variables in different ways. 

We proposed a generalised tail (corner) dependence multivariate coefficient as follows. In a $m$-dimensional copula there are $2^m$ corners that can be identified by a binary sequence $\bj\in\JB=\{0,1\}^m$ such that a value of $0$/$1$ at position $l$ denotes the lower/upper corner of variable $l$ for $l=1,\ldots,m$. Let us also denote the intervals $B_0(\nu)=[0,\nu]$ and $B_1(\nu)=(\nu,1]$. Then the tail dependence coefficient for variable $l$ at corner $\bj=(j_1,\ldots,j_m)$ is defined as 
\begin{equation}
\label{eq:mtdc}
\lambda_{\bj}^l=\lim_{\nu\to 0}\P\left[U_l\in B_{j_l}(\nu)\mid\cap_{k\neq l}^m \{U_k\in B_{j_k}(\nu)\}\right],
\end{equation}
for $l=1,\ldots,m$. There will be a total of $\#\JB=2^m$ tail dependence multivariate coefficients. It is straightforward to see that each of these coefficients can be expressed in terms of the copula $C^{(m)}$. Moreover, if the copula is symmetric, then $\lambda_\bj^l$ is the same for any $l=1,\ldots,m$. 

For example, if $m=3$ there are $2^3=8$ corners. If $\bj=(0,0,1)$ and $l=2$ then $\lambda_{001}^2=\lim_{\nu\to 0}\P(U_2\leq\nu\mid U_1\leq\nu,U_3>1-\nu)$. By using the definition of conditional probability we ca re-express this tail dependence coefficient as $\lambda_{001}^2=\lim_{\nu\to 0}\left\{C^{(2)}(\nu,\nu)-C^{(3)}(\nu,\nu,\nu)\right\}/\left\{C^{(1)}(\nu)-C^{(2)}(\nu,\nu)\right\}$, with $C^{(1)}(\nu)=\nu$. 

As in the bivariate case, it is possible to rotate an $m$-variate copula to each of the $2^m$ corners. Let $\JB_k$ the subset of $\JB$ with $k$ $1$'s, i.e., if $\bj\in\JB_k$ then $\sum_{l=1}^m j_l=k$, for $k=0,\ldots,m$. Moreover, we can partition the set $\JB$ into disjoint subsets $\JB_k$ such that, $\JB=\cup_{k=0}^m\JB_k$ with $\JB_k\cap\JB_j=\emptyset$ for $k\neq j$. The number of elements in each subset is $\#\JB_k= {{m}\choose{k}}$, and $2^m=\sum_{k=0}^m {{m}\choose{k}}$. Let also define $\bu^{\bj}=(u_1^{1-j_1}(1-u_1)^{j_1},\ldots,u_m^{1-j_m}(1-u_m)^{j_m})$. 

For $k=0$ and $\bj\in\JB_0$, then $$C_\bj(\bu)=C^{(m)}(\bu^{\bj}),$$ where $C^{(m)}$ is the original copula without rotation. 
For $k=1$ and $\bj\in\JB_1$ and if the only index $1$ is situated at the $l^{th}$ position, then $$C_\bj(\bu)=C^{(m-1)}(\bu^{\bj}_{(-l)})-C^{(m)}(\bu^{\bj}).$$ 
For $k=2$ and $\bj\in\JB_2$ and if the two indexes $1$ are situated at positions $(l_1,l_2)$, then $$C_\bj(\bu)=C^{(m-2)}(\bu^{\bj}_{(-(l_1,l_2))})-\sum_{i=1}^2 C^{(m-1)}(\bu^{\bj}_{(-l_i)})+C^{(m)}(\bu^{\bj}).$$ 
For $k=3$ and $\bj\in\JB_3$ and if the three indexes $1$ are situated at positions $(l_1,l_2,l_3)$, then $$C_\bj(\bu)=C^{(m-3)}(\bu^{\bj}_{(-(l_1,l_2,l_3))})-\sum_{i<j}^3 C^{(m-2)}(\bu^{\bj}_{(-(l_i,l_j))})+\sum_{i=1}^3 C^{(m-1)}(\bu^{\bj}_{(-l_i)})-C^{m}(\bu^{\bj}).$$
We carry on until $k=m$ and $\bj\in\JB_m$, so all indexes in $\bj$ are $1$, then 
\begin{align}
\nonumber
C_\bj(\bu)=&1-\sum_{i=1}^m C^{(1)}(1-u_i)+\sum_{i<j}^m C^{(2)}(1-u_i,1-u_j)-\sum_{i<j<k}^m C^{(3)}(1-u_i,1-u_j,1-u_k)\\ 
\nonumber
&+\cdots+(-1)^m C^{(m)}(1-u_1,1-u_2,\ldots,1-u_m).
\end{align}

If a copula exhibits multivariate tail dependence in a particular corner, each of the rotated versions of the copula will also have tail dependence in one corner. To illustrate, we also consider the Clayton family whose multivariate version is defined as
\begin{equation}
\label{eq:clayton}
C^{(m)}(u_1,\ldots,u_m)=\left\{\sum_{j=1}^m(u_i^{-\theta}-1)+1\right\}^{-1/\theta} \end{equation}
for $\theta\geq 0$. 

Specifically, we concentrate in the case $m=3$. The number of possible rotations is $2^3=8$. The four subsets of rotated copulas generated by $\JB_k$ and $k=0,1,2,3$ are: 
\begin{enumerate}
\item[(iv)] For $k=0$ we have one element $$C_{000}(\bu)=C^{(3)}(\bu)=\left(u_1^{-\theta}+u_2^{-\theta}+u_3^{-\theta}-2\right)^{-1/\theta}$$
\item[(v)] For $k=1$ we have three elements 
$$C_{100}(\bu)=\left(u_2^{-\theta}+u_3^{-\theta}-1\right)^{-1/\theta}-\left\{(1-u_1)^{-\theta}+u_2^{-\theta}+u_3^{-\theta}-2\right\}^{-1/\theta}$$
$$C_{010}(\bu)=\left(u_1^{-\theta}+u_3^{-\theta}-1\right)^{-1/\theta}-\left\{u_1^{-\theta}+(1-u_2)^{-\theta}+u_3^{-\theta}-2\right\}^{-1/\theta}$$
$$C_{001}(\bu)=\left(u_1^{-\theta}+u_2^{-\theta}-1\right)^{-1/\theta}-\left\{u_1^{-\theta}+u_2^{-\theta}+(1-u_3)^{-\theta}-2\right\}^{-1/\theta}$$
\item[(vi)] For $k=2$ we have three elements
\begin{align}
\nonumber
C_{110}(\bu)=&\,u_3-\left\{(1-u_2)^{-\theta}+u_3^{-\theta}-1\right\}^{-1/\theta}-\left\{(1-u_1)^{-\theta}+u_3^{-\theta}-1\right\}^{-1/\theta}\\
\nonumber
&+\left\{(1-u_1)^{-\theta}+(1-u_2)^{-\theta}+u_3^{-\theta}-2\right\}^{-1/\theta}\\
\nonumber
C_{101}(\bu)=&\,u_2-\left\{(1-u_1)^{-\theta}+u_2^{-\theta}-1\right\}^{-1/\theta}-\left\{u_2^{-\theta}+(1-u_3)^{-\theta}-1\right\}^{-1/\theta}\\
\nonumber
&+\left\{(1-u_1)^{-\theta}+u_2^{-\theta}+(1-u_3)^{-\theta}-2\right\}^{-1/\theta}\\
\nonumber
C_{011}(\bu)=&\,u_1-\left\{u_1^{-\theta}+(1-u_2)^{-\theta}-1\right\}^{-1/\theta}-\left\{u_1^{-\theta}+(1-u_3)^{-\theta}-1\right\}^{-1/\theta}\\
\nonumber
&+\left\{u_1^{-\theta}+(1-u_2)^{-\theta}+(1-u_3)^{-\theta}-2\right\}^{-1/\theta}
\end{align}
\item[(vii)] For $k=3$ we have one element
\begin{align}
\nonumber
C_{111}(\bu)=&\,1-(1-u_1)-(1-u_2)-(1-u_3)+\left\{(1-u_1)^{-\theta}+(1-u_2)^{-\theta}-1\right\}^{-1/\theta}\\
\nonumber
&+\left\{(1-u_1)^{-\theta}+(1-u_3)^{-\theta}-1\right\}^{-1/\theta}+\left\{(1-u_2)^{-\theta}+(1-u_3)^{-\theta}-1\right\}^{-1/\theta}\\
\nonumber
&-\left\{(1-u_1)^{-\theta}+(1-u_2)^{-\theta}+(1-u_3)^{-\theta}-2\right\}^{-1/\theta}
\end{align}
\end{enumerate}

Similarly to the bivariate case, if the tail coefficient corresponds to a rotation, then we get a coefficient different from zero. Specifically for the Clayton family with $\theta>0$, using (iv)--(vii) and \eqref{eq:mtdc}, the tail coefficients are given by
\begin{equation}
\label{eq:tdc3}
\lambda_{j_1,j_2,j_3}(C_{j_1,j_2,j_3})=\lim_{\nu\to 0}\frac{\left(3\nu^{-\theta}-2\right)^{-1/\theta}}{\left(2\nu^{-\theta}-1\right)^{-1/\theta}}=\lim_{\nu\to 0}\left(\frac{3-2\nu^{\theta}}{2-\nu^{\theta}}\right)^{-1/\theta}=\left(\frac{3}{2}\right)^{-1/\theta},
\end{equation}
for $j_l\in\{0,1\}$, $l=1,2,3$. Note that the superindex in $\lambda$ has been removed because these tail coefficients are the same for any $l$ due to the symmetry in the copula. If the tail (corner) does not correspond to a rotation, the tail coefficient is zero. 

\section{Dynamic mixtures}
\label{sec:model}

Let $\bU_t=(U_{t1},U_{t2},\ldots,U_{tm})$ be an $m$-variate vector with $\un(0,1)$ marginal distributions for each $t=1,2,\ldots,T$. The aim in this section is to model the joint distribution of $\bU_t$ through a flexible copula $C_t$ which is able to capture any kind of tail dependence as it evolves in time. For that, we use all $2^m$ rotations, each one with different parameter $\theta_{t,\bj}$ for $\bj\in\JB$, and define the following mixture copula
\begin{equation}
\label{eq:mixc}
C_t(\bu_t\mid\bpi_t,\btheta_t)=\sum_{\bj\in\JB}\pi_{t,\bj}\,C_{\bj}(\bu_{t}\mid\theta_{t,\bj}),
\end{equation}
with $2^m$-dimensional parameters $\bpi_t=(\pi_{t,00\cdots0},\pi_{t,10\cdots0},\ldots,\pi_{t,11\cdots1})$, \\ $\btheta_t=(\theta_{t,00\cdots0},\theta_{t,10\cdots0},\ldots,\theta_{t,11\cdots1})$, and $C_\bj$ is a rotated copula as defined in Section \ref{sec:mult}. 
Parameters $\pi_{t,\bj}>0$ are mixture weights such that $\sum_{\bj\in\JB}\pi_{t,\bj}=1$, whereas $\theta_{t,\bj}$ are copula parameters whose parameter space depends on the specific copula chosen, for $t=1,\ldots,T$. 

It is not difficult to derive association coefficients, like Kendall's tau, and tail dependence coefficients for a mixture copula in terms of the corresponding coefficients for the individual copulas. In particular, the Kendall's tau between any par of variables, say $(U_{t,i},U_{t,k})$, $i\neq k=1,\ldots,m$, for the mixture copula \eqref{eq:mixc} is
\begin{align*}
\tau_t&=4\E\{C_t^{(2)}(U_{t,i},U_{t,k}\mid\bpi_t,\btheta_t)\}-1=4\sum_{\bj\in\JB}\pi_{t,\bj}\E\{C_\bj^{(2)}(U_{t,i},U_{t,k}\mid\btheta_t)\}-1\\
&=\sum_{\bj\in\JB}\pi_{t,\bj}\left[4\E\{C_\bj^{(2)}(U_{t,i},U_{t,k}\mid\btheta_t)\}-1\right]=\sum_{\bj\in\JB}\pi_{t,\bj}\,\tau_{t,\bj},
\end{align*}
where $\tau_{t,\bj}$ is the individual Kendall's tau, between $(U_{t,i},U_{t,k})$, for each of the mixture copula components $C_\bj$. In particular, for the Clayton family \eqref{eq:clayton}, all rotated mixture components $C_\bj$, as in (i)--(vii), share the same properties of the Clayton family, therefore their Kendall's tau coefficients are $\tau_{t,\bj}=\theta_{t,\bj}/(2+\theta_{t,\bj})$ if $\sum_{l=1}^m j_l$ is an even number, and $\tau_{t,\bj}=-\theta_{t,\bj}/(2+\theta_{t,\bj})$ if $\sum_{l=1}^m j_l$ is an odd number. Therefore, the Kendall's tau coefficient for the mixture copula \eqref{eq:mixc} of Clayton components becomes 
\begin{equation}
\label{eq:ktaum}
\tau_t=\sum_{\bj\in\JB}(-1)^{\sum_{l=1}^m j_l}\,\pi_{t,\bj}\left(\frac{\theta_{t,\bj}}{2+\theta_{t,\bj}}\right).
\end{equation}

Similarly to Kendall's tau, it is not difficult to prove that multivariate tail dependence coefficients \eqref{eq:tdc3} for mixture copula \eqref{eq:mixc} at time $t$,  become the mixture of components-wise tail dependence coefficients, that is 
\begin{equation}
\label{eq:tdcm}
\lambda_{t,\bk}(C_t)=\sum_{\bj\in\JB}\pi_{t,\bj}\lambda_{t,\bk}(C_\bj),
\end{equation}
for any $\bk\in\JB$ and with $\lambda_{t,\bk}(C_\bj)$ the $\bk$-tail dependence coefficient for rotated copula $C_{\bj}$. For the Clayton family case, $\lambda_{t,\bk}(C_\bj)>0$ only for $\bk=\bj$. Therefore, when $m=2$ and using \eqref{eq:tdc2}, we obtain that the tail dependence coefficients simplify to $\lambda_{t,j,k}(C_t)=\pi_{t,j,k}(2)^{-1/\theta_{t,j,k}}$. Now, for $m=3$ we use \eqref{eq:tdc3} and obtain that the tail dependence coefficients are $\lambda_{t,j_1,j_2,j_3}(C_t)=\pi_{t,j_1,j_2,j_3}(3/2)^{-1/\theta_{t,j_1,j_2,j_3}}$.

In summary, our mixture copula proposal \eqref{eq:mixc} is flexible enough to capture a larger class of dependence associations and all $2^m$ tail dependencies, in terms of the copula parameters $\pi_{t,\bj}$ and $\theta_{t,\bj}$ for $\bj\in\JB$ and $t=1,\ldots,T$. 

\section{Bayesian analysis}
\label{sec:bayesian}

\subsection{Prior distributions}

To allow for temporal dependence in the parameter estimation, we propose a prior dynamic process for $\bpi=\{\bpi_t\}_{1\le t \le T}$, where $\bpi_t=\{\pi_{t,\bj},\bj\in\JB\}$. Since $\sum_{\bj\in\JB}\pi_{t,\bj}=1$ for all $t$, the natural marginal prior for $\bpi_t$ would be a Dirichlet distribution with parameter $a_0\bp$, where $\bp=\{p_\bj,\bj\in\JB\}$ such that $a_0>0$, $p_\bj>0$ and $\sum_{\bj\in\JB}{p_\bj}=1$. To relate a set of Dirichlet random variables, we use ideas from \cite{nieto&gutierrez:22}, who defined dependence in univariate random variables whose distributions belong to the exponential family, and define a dynamic prior with temporal dependence as follows. 

Let $\bfeta_t=\{\eta_{t,\bj},\bj\in\JB\}\in \RB^{2^m}$ be a latent vector corresponding to each $\bpi_t$ and let $\bomega=\{\omega_\bj,\bj\in\JB\}$ be a unique latent vector such that 
\begin{equation}
\label{eq:ddir1}
\bomega\sim\dir(a_0\bp)\quad\mbox{and}\quad\bfeta_t\mid\bomega\simind\mul(a_t,\bomega),
\end{equation}
with $a_t\in\NB$, $\eta_{tk}\in\NB$ and $\sum_{\bj\in\JB}\eta_{t,\bj}=a_t$. Then, the prior dependence in $\bpi_t$ is modeled through a subset $\partial_t$ of previous latent variables $\{\bfeta_1,\bfeta_{2},\ldots,\bfeta_{t-1}\}$ 
\begin{equation}
\label{eq:ddir2}
\bpi_t\mid\bfeta\simind\dir\left(a_0\bp+\sum_{k\in\partial_t}\bfeta_{k}\right).
\end{equation}
We denote this construction as $\ddir(a_0,\ba,\bpartial)$ with $a_0>0$, $\ba=(a_1,\ldots,a_T)$ and subsets $\bpartial=\{\partial_t\}$ of lags. Different temporal dependencies can be induced by an appropriate selection of subsets of lags. For instance: \emph{moving average} type of order $q$ can be induced by defining 
\begin{equation}
\label{eq:maq}
\partial_t=\{t,t-1,\ldots,t-q\}; 
\end{equation}
\emph{seasonal dependence} of order $p$ can be induced by defining, for seasonality $s$, 
\begin{equation}
\label{eq:sp}
\partial_t=\{t,t-s,t-2s,\ldots,t-ps\}; 
\end{equation}
or a combination of moving average of order $q$ and seasonal dependence of order $p$. In general, the only requirement is that $t\in\partial_t$. 

Properties of this prior are given in the following proposition. 

\begin{proposition}
Let $\bpi=\{\bpi_t\}\sim\ddir(a_0,\ba,\bpartial)$ a sequence of vectors whose probability law is defined by \eqref{eq:ddir1} and \eqref{eq:ddir2} for $a_0>0$, $a_t\in\NB$ and subsets $\bpartial$. Then, 
\begin{itemize}
\item[(i)]
The marginal distribution for each $\bpi_t$ is $\dir(a_0\bp)$, 
\item[(ii)]
The correlation between $\pi_{t,\bj}$ and $\pi_{r,\bj}$, for $t\neq r$ and $\bj\in\JB$, does not depend on the specific $\bj$ and is given by 
$$\Cor(\pi_{t,\bj},\pi_{r,\bj})=\frac{a_0\left(\sum_{k\in\partial_t\cap\partial_r}a_{k}\right)+\left(\sum_{k\in\partial_t} a_{k}\right)\left(\sum_{k\in\partial_r} a_{k}\right)}{\left(a_0+\sum_{k\in\partial_t} a_{k}\right)\left(a_0+\sum_{k\in\partial_r} a_{k}\right)}$$
\item[(iii)]
If $a_t=0$ for all $t=1,2\ldots$ then the $\bpi_t$'s become independent.
\end{itemize}
\end{proposition}
\begin{proof}

\noindent
For (i) we rely on conjugacy properties of the Dirichlet multinomial Bayesian updating \citep{bernardo&smith:00}. This states that if $\bfeta_t$, $t=1,2,\ldots$ are conditionally independent given $\bomega$ in \eqref{eq:ddir1}, whose prior is $\bomega \sim \dir(a_0\bp)$, then the posterior distribution for $\bomega$ given the $\bfeta_t$'s is $\dir\left(a_0\bp+\sum_{t}\bfeta_t\right)$. Replacing $\bomega$ in the posterior by $\bpi_t$ we obtain that the marginal distribution for $\bpi_t$ is the same as the prior for $\bomega$. 

\noindent
For (ii) we first note that for a specific $\bj$, the distributions for $\omega_\bj$, $\eta_{t,\bj}$ and $\pi_{t,\bj}$ reduce to beta, binomial and beta, respectively.
To obtain the correlation we rely on iterative formulae. The covariance is $\Cov(\pi_{t,\bj},\pi_{r,\bj})=\E\{\Cov(\pi_{t,\bj},\pi_{r,\bj}\mid\bfeta)\}+\Cov\{\E(\pi_{t,\bj}\mid\bfeta),\E(\pi_{r,\bj}\mid\bfeta)\}$, where the first term is zero due to conditional independence. Then $$\Cov(\pi_{t,\bj},\pi_{r,\bj})=\Cov\left\{\frac{a_0p_\bj+\sum_{k\in\partial_t}\eta_{k,\bj}}{a_0+\sum_{k\in\partial_r} a_{k}},\frac{a_0p_\bj+\sum_{k\in\partial_r}\eta_{k,\bj}}{a_0+\sum_{k\in\partial_r} a_{k}}\right\},$$
which, after canceling the additive constants and using the linearity of the covariance, becomes
$$\Cov(\pi_{t,\bj},\pi_{r,\bj})=\frac{1}{\left(a_0+\sum_{k\in\partial_t} a_{k}\right)\left(a_0+\sum_{k\in\partial_r} a_{k}\right)}\Cov\left\{\sum_{k\in\partial_t}\eta_{k,\bj},\sum_{k\in\partial_r}\eta_{k,\bj}\right\}.$$
After using the iterative formula for a second time, we get
\begin{equation}\E\left[\Cov\left\{\left.\sum_{k\in\partial_t}\eta_{k,\bj},\sum_{k\in\partial_r}\eta_{k,\bj}\right|\bomega\right\}\right]+\Cov\left\{\E\left(\left.\sum_{k\in\partial_t}\eta_{k,\bj}\right|\bomega\right),\E\left(\left.\sum_{k\in\partial_r}^q\eta_{k,\bj}\right|\bomega\right)\right\}.\label{cov}
\end{equation}
Within each sum we can isolate the common part as $\sum_{k\in\partial_t}\eta_{k,\bj}=\sum_{k\in\partial_t\cap\partial_r}\eta_{k,\bj}+\sum_{k\in\partial_t-\partial_r}\eta_{k,\bj}$ and $\sum_{k\partial_r}\eta_{k,\bj}=\sum_{k\in\partial_t\cap\partial_r}\eta_{k,\bj}+\sum_{k\in\partial_r-\partial_t}\eta_{k,\bj}$, and using covariance properties and conditional independence, \eqref{cov} becomes
$$\E\left\{\V\left(\left.\sum_{k\in\partial_t\cap\partial_r}\eta_{k,\bj}\right|\omega_\bj\right)\right\}+\Cov\left\{\sum_{k\in\partial_t} a_{k}\omega_\bj,\sum_{k\in\partial_r} a_{k}\omega_\bj\right\}.$$
The first expected value, after obtaining the conditional variance is $\E\{\sum_{k\in\partial_t\cap\partial_r}a_{k}\omega_\bj(1-\omega_\bj)\}=(\sum_{k\in\partial_t\cap\partial_r}a_{k})\E(\omega_\bj-\omega_\bj^2)$ with $\E(\omega_\bj-\omega_\bj^2)=\E(\omega_\bj)-\E^2(\omega_\bj)-\V(\omega_\bj)=a_0\V(\omega_\bj)$. The second term is $(\sum_{k\in\partial_t} a_{k})(\sum_{k\in\partial_r} a_{k})\V(\omega_\bj)$. In conclusion, we obtain  $$\Cov(\omega_{t,\bj},\omega_{r,\bj})=\frac{a_0\left(\sum_{k\in\partial_t\cap\partial_r}a_{k}\right)+\left(\sum_{k\in\partial_t} a_{k}\right)\left(\sum_{k\in\partial_r} a_{k}\right)}{\left(a_0+\sum_{k\in\partial_t} a_{k}\right)\left(a_0+\sum_{k\in\partial_r} a_{k}\right)}\V(\omega_\bj).$$
Since $\omega_\bj$, $\pi_{t,\bj}$ and $\pi_{r,\bj}$ all have the same beta marginal distribution, (ii) is demonstrated. 

\noindent
For (iii) we note that $a_t=0$ for all $t$ implies that $\bfeta_t=0$ with probability one so the dependence disappears and $\bpi_t$ become independent with marginal distribution $\dir(a_0\bp)$. 
\end{proof}

The strength of dependence in the prior for $\bpi$ depends on the model parameters $a_0$, $\ba$ and subsets $\bpartial$. Larger values of any of the first two induce stronger dependence. More shared elements in $\partial_t$ and $\partial_r$ also indicate stronger dependence. However, if the intersection between sets $\partial_t$ and $\partial_r$ is empty, the correlation is still positive. 

Prior distributions are completed by assigning hierarchical gamma distributions for each $\theta_{t,\bj}$, so that information is shared across times $t$ for each $\bj$. That is, 
\begin{equation}
\label{eq:theta}
\theta_{t,\bj}\mid\beta_\bj\simind\ga(d_\bj,\beta_\bj),\;\;\mbox{and}\;\; \beta_\bj\sim\ga(e_\bj,g_\bj)
\end{equation}
for $t\ge 1$ and $\bj\in\JB$. 

\subsection{Posterior distributions}

Let $\bU_{t,i}=(U_{1,t,i},\ldots,U_{m,t,i})$ for $i=1,\ldots,n_t$ a sample of size $n_t$ from model \eqref{eq:mixc} for each $t=1,\ldots,T$. Let $\bZ_{t,i}$ be a latent vector that identifies the mixture component from where observation $i$ is coming from, that is, $\bZ_{ti}=\{Z_{\bj,t,i},\bj\in\JB\}\sim\mul(1,\bpi_t)$. Assuming for the moment that together with $\bU_{t,i}$ we observe $\bZ_{t,i}$, then the extended likelihood has the form
$$f(\bu,\bz\mid\bpi,\btheta)=\prod_{t=1}^T\prod_{i=1}^{n_t}\prod_{\bj\in\JB}\left\{\pi_{t,\bj}f_\bj(u_{1,t,i},\ldots,u_{m,t,i}\mid\theta_{t,\bj})\right\}^{z_{\bj,t,i}},$$
where $f_\bj(\bu_{t,i}\mid\theta_{t,\bj})$ are the density functions associated to each of the copula rotations $C_\bj$ for $\bj\in\JB$. 

In particular, for the Clayton family \eqref{eq:clayton}, the corresponding density becomes
\begin{equation*}
f_C(\bu)=\left\{\sum_{i=1}^{m}u_i^{-\theta}-m-1\right\}^{-1/\theta-m}\prod_{i=1}^{m}\{1+(i-1)\theta\}u_i^{-\theta-1}.
\end{equation*}
For any rotated version we just evaluate at $\tilde{\bu}=(\tilde{u}_1, \tilde{u}_2, \ldots, \tilde{u}_m)$ where $\tilde{u}_i$ can be $u_i$ or $1-u_i$. 

The prior distribution for $(\bpi,\btheta)$ is defined by equations \eqref{eq:ddir1}, \eqref{eq:ddir2} and \eqref{eq:theta}. Again, extending the prior to include the latent variables $\bfeta$ and $\bomega$ we get
$$f(\bpi,\bfeta,\bomega)=\dir(\bomega\mid a_0\bp)\prod_{t=1}^T\dir\left(\bpi_t\left|a_0\bp+\sum_{j\in\partial_t}\bfeta_{j}\right.\right)\mul(\bfeta_t\mid a_t,\bomega)$$
and 
$$f(\btheta)=\prod_{\bj\in\JB}\ga(\beta_\bj\mid e_\bj,g_\bj)\prod_{t=1}^T\ga(\theta_{t,\bj}\mid d_\bj,\beta_\bj),$$
independent of each other. 

Posterior distributions are characterized through their full conditional distributions. These include actual parameters as well as latent variables and are given as follows.
\begin{enumerate}
\item[(a)] The posterior conditional for $\bZ_{ti}$ is
$$\bZ_{ti}\mid\rest\sim\mul(1,\bpi_t^*),$$
where $\bpi^*=\{\pi_{t,\bj}^*\}$ and $$\pi_{t,\bj}^*=\frac{\pi_{t,\bj}f_\bj(\bu_{t,i}\mid\theta_{t,\bj})}{\sum_{\bk\in\JB}\pi_{t,\bk}f_\bk(\bu_{t,i}\mid\theta_{t,\bk})}.$$ 
\item[(b)] The posterior conditional for $\bpi_t$ is
$$\bpi_t\mid\rest\sim\dir\left(a_o\bp+\sum_{k\in\partial_t}\bfeta_{k}+\sum_{i=1}^{n_t}\bz_{t,i}\right).$$
\item[(c)] The posterior conditional for $\bfeta_t$ is
$$f(\bfeta_t\mid\rest)\propto\left\{\prod_{\bj\in\JB}\frac{\left(\omega_\bj\prod_{l\in\varrho_t}\pi_{l,\bj}\right)^{\eta_{t,\bj}}}{\Gamma(\eta_{t,\bj}+1)\prod_{l\in\varrho_t}\Gamma\left(a_0p_\bj+\sum_{k\in\partial_l}\eta_{k,\bj}\right)}\right\}I\left(\sum_{\bj\in\JB}\eta_{t,\bj}=a_t\right),$$ where $\varrho_t=\{l:t\in\partial_l\}$ is the set of inverse subsets. 
\item[(d)] The posterior conditional for $\bomega$ is
$$f(\bomega\mid\rest)=\dir\left(\bomega\left|c_0\bp+\sum_{t=1}^T\bfeta_t\right.\right).$$
\item[(e)] The posterior conditional for $\theta_{t,\bj}$ is
$$f(\theta_{t,\bj}\mid\rest)\propto \theta_{t,\bj}^{d_\bj-1}e^{-\beta_\bj\theta_{t,\bj}}\prod_{i=1}^{n_t}\left\{f_\bj(\bu_{t,i}\mid\theta_{t,\bj})\right\}^{z_{t,\bj,i}}.$$
\item[(f)] The posterior conditional for $\beta_\bj$ is
$$\beta_\bj\mid\rest\sim\ga\left(e_\bj+Td_\bj\,,\,g_\bj+\sum_{t=1}^T\theta_{t,\bj}\right).$$
\end{enumerate}

Posterior inference will rely on the implementation of a Gibbs sampler \citep{smith&roberts:93} based on the previous posterior conditional distributions. Sampling from (a), (b), (d) and (f) is straightforward since they are of standard form. To sample from (c), since $\bfeta_t$ is a vector of dimension $2^m$ with a sum restriction, it is easier if we sample from each of the components $\eta_{t,\bj}$ for $\bj\in\JB_{(-\bone)}=\JB-\bone$, with $\bone=11\cdots1$ the vector of all $1$'s, using $f(\eta_{t,\bj}\mid\rest)\propto$
$$\frac{\left\{\omega_\bj\prod_{l\in\varrho_t}\pi_{l,\bj}/\left(\omega_\bone\prod_{l\in\varrho_t}\pi_{l,\bone}\right)\right\}^{\eta_{t,\bj}}I\left(\eta_{t,\bj}\in\{0,1,\ldots,a_t-\sum_{\bk\in\JB_{(-\bone)}}\eta_{t,\bk}\}\right)}{\Gamma(\eta_{t,\bj}+1)\prod_{l\in\varrho_t}\Gamma\left(a_0p_\bj+\sum_{k\in\partial_l}\eta_{k,\bj}\right)\Gamma(\eta_{t,\bone}+1)\prod_{l\in\varrho_t}\Gamma\left(a_0p_\bone+\sum_{k\in\partial_l}\eta_{k,\bone}\right)},$$
with $\eta_{t,\bone}=a_t-\sum_{\bj\in\JB_{(-\bone)}}\eta_{t,\bj}$. 
Sampling from (e) will require a Metropolis-Hastings step \citep{tierney:94}. We suggest to use an adaptive random walk proposal defined as follows. At iteration $(r+1)$ sample $\theta_{t,\bj}^*\sim\ga(\kappa,\kappa/\theta_{t,\bj}^{(r)})$ and accept it with probability 
$$\alpha(\theta_{t,\bj}^*,\theta_{t,\bj}^{(r)})=\frac{f(\theta_{t,\bj}^*\mid\rest)\ga(\theta_{t\bj}^{(r)}\mid\kappa,\kappa/\theta_{t,\bj}^*)}{f(\theta_{t,\bj}^{(r)}\mid\rest)\ga(\theta_{t,\bj}^*\mid\kappa,\kappa/\theta_{t,\bj}^{(r)})},$$
where $\alpha$ is truncated to the interval $[0,1]$ and $\kappa$ is a tuning parameter that controls the acceptance rate. We adapt $\kappa$ following  the method of \cite{nieto:24}. The adaptation method uses batches of $50$ iterations and for every batch $h$ we compute the acceptance rate $AR^{(h)}$ and increase $\kappa^{(h+1)}=\kappa^{(h)}1.01^{\sqrt{h}}$ if $AR^{(h)}<0.3$ and decrease $\kappa^{(h+1)}=\kappa^{(h)}1.01^{-\sqrt{h}}$ if $AR^{(h)}>0.4$, with $\kappa^{(1)}=1$ as starting value. This adaptation scheme satisfies diminishing adaptation as $h \rightarrow \infty$ and in the applications we restrict the parameters to a compact thus ensuring that the sampler is valid \citep{rr2007}. 

Implementation code for our dynamic Clayton mixture model, for dimensions $m=2$ and $m=3$, can be found at the GitHub repository \url{https://github.com/RuyiPan/TD-MRC}.

\section{Numerical analyses}
\label{sec:numerical}

\subsection{Simulation study}

We conduct a comprehensive simulation study with $m=2$ to evaluate the performance of the proposed model. The true generative model is set using $\theta_t=(\theta_{t,00}, \theta_{t,10}, \theta_{t,01}, \theta_{t,11})=(5,3,4,3)$ for $t=1,\ldots,T$, with $T=20$. The weights of rotated Clayton copulas are chosen to be linearly dependent in time. More specifically, we first set $\bpi_1=(\pi_{1,00}, \pi_{1,10}, \pi_{1,01}, \pi_{1,11})=(0.4,0.25,0.25,0.1)$ as the initial values at $t=1$ and subsequently the  weights are constructed using $\pi_{t,00}=0.95\pi_{t-1,00}, \pi_{t,10}=1.05\pi_{t-1,10}, \pi_{t,11}=0.1, \pi_{t,01}=1-\pi_{t,00}-\pi_{t,10}-\pi_{t,11}$ for $t=2,\ldots,20$. We sampled $n_t=300$ realizations from this model for each time $t=1,\ldots,T$ as the simulated data. 

For the prior distributions, we set hyper-parameters $a_0=1$, and $e_k=g_k=1$. To evaluate the performance of the model in capturing the temporal dependence, we considered a moving average type of order $q$, that is, the subsets of lags are defined by \eqref{eq:maq}. We considered different hyper-parameters for the dynamic process: $a_t=0,1,10,20,30,40$ and $q=0,1,\ldots,7$. The model with $a_t$ and $q$ is denoted as ${\MC}_{a_t,q}$. We ran the MCMC for $7,000$ iterations. To determine the burn-in, we monitor the adaptive $\kappa$ parameter and the acceptance rate for each batch. These are included in Figure \ref{fig:kappa_acc_sim} where we observe that the $\kappa$ becomes stable and the acceptance rate stabilizes between 0.3 and 0.4 after 60 batches. Therefore we set the burn-in to be equal to $3,000=60\times 50$ iterations. This also confirms that the adaptation method proposed at the end of Section \ref{sec:bayesian} performs well. Computing time is around 50 minutes for each run on an Intel Core i9 processor at 2.3 GHz with 32 GB of RAM. 

To assess model performance, we computed two goodness of fit (gof) measures, the logarithm of the pseudo marginal likelihood (LPML) \citep{geisser&eddy:79} and Watanabe–Akaike Information Criterion (WAIC) \citep{watanabe&opper:10}. Table \ref{tab:lpml_waic_sim} shows these values. In general, the two gof measures concur in determining the best model for each value of $q$. Briefly put, for smaller values of $a_t$, better fitting is achieved for larger orders of dependence $q$ in the $\bpi$, whereas for larger values of $a_t$, smaller orders of dependence $q$ produce better fit. Overall, the best model is, ${\MC}_{30,7}$, obtained with $a_t=30$ and $q=7$.

Two more comparisons are also included in Table \ref{tab:lpml_waic_sim}. The case of independence across times for $\pi_{t,\bj}$ is obtained when $a_t=0$, regardless of the value of $q$. Goodness of fit statistics show that the independence fitting is not the worst, but is definitely underperforming in comparison to the other dependence models. Additionally, we also consider the model that assumes independence in the $\theta_{t,\bj}$. This is achieved by considering a very low variance in $\beta_\bj$, as obtained by setting $e_\bj=g_\bj=1000$. This latter model produces inferior gof measures as compared to the hierarchical priors. 

To assess in more detail our model's performance, we compare posterior estimates of $\bpi$ and $\btheta$ with the true values. We use the best fitting model and use posterior means as point estimates, together with quantiles 2.5\% and 97.5\% to produce 95\% credible intervals. Figure \ref{fig:pi_theta_predict_sim} (left panel) shows posterior estimates of $\pi_{t,\bj}$ as time series for $t=1,\ldots,20$ in four panels for $\bj\in\{00,10,01,11\}$. Posterior estimates follow very closely the path of the true values. Similarly, Figure \ref{fig:pi_theta_predict_sim} (right panel) shows posterior estimates of $\theta_{t,\bj}$ as time series in four panels. All the true values lie within the 95\% credible intervals. We note that credible intervals for $\theta_{t,00}$ are narrower at the beginning and become wider toward the end (bottom-left panel) whereas the credible intervals for $\theta_{t,10}$ are wider at the beginning and narrow toward the end (bottom-right panel). The credible intervals for $\theta_{t,11}$ remain wide overall (top-right panel), while those for $\theta_{t,01}$ are narrow (top-left panel) in general. Wider credible intervals are due to smaller weights (less data points) associated to the corresponding mixture components. Specifically, the higher variability for $\theta_{t,11}$ for larger $t$ is a consequence of the smaller weights for the third component $\pi_{t,11}$.

We compare the best fit produced by  ${\cal{M}}_{a_t,q}$ with the dynamic copula model of Hafner \& Manner  \cite{hafner&manner:12}. This latter model assumes a Gaussian copula with time-varying association parameter $\rho_t$. It relies on a Fisher transformation (inverse hyperbolic tangent) of the association parameter as $\lambda_t=(1/2)\log\left((1+\rho_t)/(1-\rho_t)\right)$ and models the dynamics via an autoregressive process of the form $\lambda_t=\alpha+\beta\lambda_{t-1}+\epsilon_t$ with $\epsilon_t\sim\no(0,\tau)$. We perform a Bayesian analysis for this model with vague prior distributions $\alpha\sim\no(0,0.01)$, $\beta\sim\no(0,0.01)$ and $\tau\sim\ga(0.01,0.01)$. We will refer to this model as dynamic Gaussian. 

As a second competitor, we consider a simple Clayton copula, which is the result of assigning fixed weight one to the mixture model and keeping the exchangeable prior on the association copula time varying parameters. 

To compare, we compute the log predictive scores (LPS) defined in \cite{geweke&amisano:10}. We fit models with data up to time $t-1$ and compute the LPS for time $t$, that is, $LPS(t)=\sum_{i=1}^{n_t}\log f(\bu_{t,i}\mid \bu_{t-1})$, where the predictive distribution is approximated via Monte Carlo as $$f(\bu_{t,i}\mid \bu_{t-1})=(1/R)\sum_{r=1}^R f(\bu_{t,i}\mid \bpi_t^{(r)},\btheta_t^{(r)})$$ and $(\bpi_t^{(r)},\btheta_t^{(r)})$ are obtained from the posterior distribution $f(\bpi_t,\btheta_t\mid \bu_1,\ldots,\bu_{t-1})$ for a total of $R$ iterations. 

The LPS measures at $t=20$ as well as the LPML and WAIC for the different models are reported in Table \ref{tab:sim_model_comparison}. The performance of the simple Clayton copula is the worst, as expected. The dynamic Gaussian improves a little the goodness of fit measures, but is far behind the dynamic Clayton mixtures. When removing the $n_{20}$ data points from the fitting, the best performance is achieved when taking $a_t=20$ and $q=7$.

\subsection{Bivariate real data analysis}

We also assess how well our model can capture the relationship between variables in a real life application where data is generated from some unknown distribution, rather than directly from a mixture copula.

We used the Environment and Climate Change Canada (ECCC) data catalogue from the government of Canada. The ECCC managed the National Air Pollution Surveillance Program (NAPS), which began in 1969 and is now comprised of nearly 260 stations in 150 rural and urban communities reporting to the Canada-wide air quality database (for more details about the dataset, please visit \href{https://data-donnees.az.ec.gc.ca/data/air/monitor/national-air-pollution-surveillance-naps-program/?lang=en}{\it{https://data-donnees.az.ec.gc.ca/data/air/monitor}}). 

Specifically, we selected ozone ($O_3$) and particulate matter with diameters 2.5 and smaller ($PM_{2.5}$) as the bivariate data. We used the hourly data from the years 2017 to 2019. Due to a large number of missing values, we took averages across hours and across days to produce monthly data for each station, $t=1,\ldots,T$ with $T=36$ for a total time span of three years. The number of stations varies across months, specifically we have $n_t=177$ for the months in 2017, $n_t=187$ for 2018 and $n_t=189$ for 2019. 

Since our focus is not on the marginal distributions, but on the association between these two variables, we applied the modified rank transformation \citep{deheuvels:79} to produce data in the interval $[0,1]$. Specifically, if observed data is denoted by $(X_{1ti},X_{2ti})$, we compute the empirical cdf's, independently for each variable, say $\widehat{F}_{X_1}$ and $\widehat{F}_{X_2}$ and apply the inverse transformation $U_{1ti}=\widehat{F}_{X_1}^{-1}(X_{1ti})$ and $U_{2ti}=\widehat{F}_{X_2}^{-1}(X_{2ti})$. 

To explore the data, we computed empirical Kendall's tau and Spearman's rho association coefficients per month. These are shown in Figure \ref{fig:kappa_acc_sim}. In both cases the dependence is cyclical around zero, reaching its maximum in June/July and its minimum in October/November. This suggests that a seasonal specification of our model seems to be a good candidate to capture these cycles. For completeness we considered three types of models: moving average type of order $q$ as in \eqref{eq:maq}; seasonal model of order $p$ with annual seasonality $s=12$ as in \eqref{eq:sp}; and a combination of moving average and seasonal. We denote the model ${\cal{M}}_{a_t,q,p}$ where the indexes describe chosen values for $a_t$, $q$ and $p$. 

Similar to the simulation study, we set the parameters $a_0=1$ and $e_k=g_k=1$ to define the prior distributions. In this real data analysis, we also varied the dependence parameters $a_t=0,1,\ldots,5$, $q=0,1,\ldots,4$ and $p=0,1,2$ to assess the performance of the model under different strengths of temporal dependence. We ran the MCMC for $10,000$ iterations and set the burn-in equal to $5,000$. Computing time is around 60 minutes for each run in an Intel Core i9 processor at 2.3 GHz with 32 GB of RAM. 

Table \ref{tab:LPML_WAIC_real} shows the LPML and WAIC values for the different prior specifications. The two gof measures agree that  the  preferred model is ${\MC}_{1,1,2}$, i.e. when $a_t=1$, $q=1$ and $p=2$. In summary, a combination of moving average and seasonality are required to model this particular dataset.  

In Figure \ref{fig:pi_theta_predict_real} we show posterior estimates of the weights $\bpi$ (left panel) and copula coefficients $\btheta$ (right panel), respectively. The cyclical monthly dependence is captured by the weights. Since the first and third components of the mixture induce positive dependence, and second and fourth induce negative dependence, there is an opposite behaviour between the pairs $(\pi_{t,00},\pi_{t,11})$ and $(\pi_{t,10},\pi_{t,01})$. The former reaches its peak in the summer and the latter in the autumn-winter, however the peaks within the second pairs do not occur exactly at the same months, $\pi_{t,10}$ has its peak around September-October (autumn), whereas $\pi_{t,01}$ has its peak around December-January (winter). Among the four components, component 3 is the one with slightly smaller weights,  apart from the summer of the year 2019 where $\pi_{t,11}$ has two peaks in May and September of around $0.75$ and $0.5$, respectively. Therefore, our mixture model is able to capture the seasonal dynamics in the data.  

The strength of the association between the pair $(O_3,PM_{2.5})$ is determined by parameters $\theta_{t,\bj}$. Their posterior estimates are all around slightly below the value of one, with $\theta_{t,10}$ and $\theta_{t,01}$ showing more variability along time. Uncertainty in the estimation of $\theta_{t,11}$ is somehow higher, due to the slightly smaller weights $\pi_{t,11}$ and thus smaller sample size for estimating $\theta_{t,11}$. According to $\theta_{t,00}$, positive dependence is particularly higher in the summer of the years 2017 and 2019 with a lower-lower tail dependence. On the other hand, looking at $\theta_{t,10}$ and $\theta_{t,01}$, negative dependence is high in the winter of the three years. 

We can further assess the tail dependence in the four corners of the unit square by computing the tail dependence coefficients $\blambda_t$, given in \eqref{eq:tdcm}. These are reported in Figure \ref{fig:tail_coefficient_predict_real}. We first note that none of them is larger than $1/2$, the only exception being the upper-lower $\lambda_{t,10}$ in October of 2018, where perhaps $O_3$ was extremely high and $PM_{2.5}$ was extremely low in that month. The lower-lower and upper-upper tail parameters $\lambda_{t,00}$ and $\lambda_{t,11}$ show very similar behaviour, they are most of the time close to zero, both with moderate peaks in July of the three years. On the other hand, the upper-lower and lower-upper tail parameters $\lambda_{t,10}$ and $\lambda_{t,01}$ do not behave exactly in the same way. These have wider peaks than the previous tail coefficients with moderate tail dependencies in the autumn for $\lambda_{t,10}$ and in the winter for $\lambda_{t,01}$.  

Finally, we show joint density estimates in Figure \ref{fig:joint_predict_real} as heat plots, together with scatter plots of the data for each month $t$. We particularly concentrate on the months where the dependence changes from negative to positive. This transition is consistent along the three years of study for the months of June and July. It is interesting to see that August is a transition month, where in 2017 and 2019 the dependence is 4-way, i.e., the four mixture components of our model are present, in fact the estimated weights and coefficients are: $\bpi_{2017-8}=(0.41,0.52,0.02,0.05)$, $\btheta_{2017-8}=(0.69,0.69,0.81,0.59)$; $\bpi_{2018-8}=(0.07,0.86,0.04,0.03)$, $\btheta_{2018-8}=(0.48,0.31,0.84,0.50)$; $\bpi_{2019-8}=(0.11,0.59,0.09,0.21)$, $\btheta_{2019-8}=(0.57,0.31,0.59,0.58)$. What makes the heat plots to show the 4-way dependence is a combination of the weight $\pi_{t,\bj}$ and the intensity $\theta_{t,\bj}$. 

We compare the fit of ${\MC}_{1,1,2}$ with the dynamic Gaussian and simple Clayton copula and carry out two validation studies.  In the first validation study we partition the dataset into two sets, fitting and testing. For each month $t=1,\ldots,T$ we took $n_{1t}=140$ observations for fitting and $n_{2t}=n_t-140$ for testing. We estimate the model parameters with the fitting set and use the testing set to predict $O_3$ conditional on the observed value of $PM_{2.5}$. For this we use the conditional copula $C_t(u_{1,t}\mid u_{2,t},\bpi_t,\btheta_t)=\sum_{\bj}\pi_{t,\bj}C_k(u_{2,t}\mid u_{1,t},\theta_{t,\bj})$ with $C_\bj(u_{2,t}\mid u_{1,t},\theta_{t,\bj})=(\partial/\partial u_{1,t})C_\bj(u_{1,t},u_{2,t}\mid\theta_{t,\bj})$ for $\bj\in\{00,10,01,11\}$, and obtain the posterior predictive mean $\widehat{u}_{2,t}=\E(U_{2,t}\mid u_{1,t},\bpi_t,\btheta_t)$ as point prediction. 

We assess model performance by computing the mean square error between the observed $u_{2,t}$ and predicted $\widehat{u}_{2,t}$ for $O_3$ in the test set, i.e. $MSE=\left(\sum_{t=1}^T n_{2t}MSE_t\right)/\left(\sum_{t=1}^T n_{2t}\right)$, where $MSE_t=(1/n_{2t})\sum_{i=1}^{n_{2t}}(u_{2,t}-\widehat{u}_{2,t})^2$, as well as the LPML and WAIC goodness of fit measures for the fitting sets.  Results from the first validation study are included in Table \ref{tab:validation_real}. All four gof statistics for the dynamic Clayton mixture model are better than those obtained with the dynamic Gaussian and simple Clayton copula models, confirming that our model is preferred for the analysis. 

The second validation study consists in comparing the log predictive scores (LPS) defined in \cite{geweke&amisano:10}. We fit models with data up to time $t-1$ and compute the LPS for time $t$. We repeat this for times $t=s+1,\ldots,T$ and aggregate the scores as $$LPS=\sum_{t=s+1}^T\sum_{i=1}^{n_t}\log f(\bu_{t,i}\mid \bu_{t-1}).$$  In particular we took $s=30$. The values of LPS are reported in Table \ref{tab:real_model_comparison}. Here we observe that the simple Clayton has the worst predictive scores. For times $t=31,32,33,35$ the dynamic Gaussian copula obtains a better predictive score, and for times $t=34,36$ the dynamic Clayton mixtures has better performance. However, when aggregating the predictive scores for the six predicted times, our proposed model achieves the best performance. 

\subsection{Multivariate real data analysis}

Using the same Canadian pollution repository ECCC, we  selected $PM_{2.5}$, nitrogen dioxide ($NO_2$) and sulfur dioxide ($SO_2$) to perform a multivariate analysis with $m=3$ variables. Again, hourly data were averaged to produce monthly data for years 2017 to 2019. In total we have $T=36$ months and $n_t=91$ stations for each month in 2017, $n_t=104$ for 2018, and $n_t=100$ for 2019. We applied the modified rank transformation to produce data in the interval $[0,1]$.

We first performed an exploratory analysis and computed pairwise empirical Kendall's tau coefficients for the three variables. These are reported in Figure \ref{fig:real_tau}. Dependencies are not seasonal as in the bivariate case and the three coefficients fluctuate between negative and positive values. Largest positive dependence occurs between $PM_{2.5}$ and $NO_2$, and lowest negative dependence between $NO_2$ and $CO_2$. 

Prior specifications, dependence search and MCMC settings were the same as in the bivariate real data analysis. Table \ref{tab:LPML_WAIC_real2} reports the LPML and WAIC gof statistics. Both measures agree on the best model which is $\MC_{4,4,2}$, that is, a dynamic total mass $a_t=4$, a moving average order of $q=4$ and a seasonal order of $p=2$. Although the seasonality was not clear in the exploratory analysis, the best model is using information from previous two years. 

Posterior inference for parameters $\bpi$ and $\btheta$ are shown in Figure \ref{fig:pi_dim3_predict_real} obtained with the preferred model. We observe that component $000$ (no rotation) dominates the mixture with the highest weight $\pi_{000}$ of around $75\%$, followed by component $001$ with the remaining $25\%$. The rest of the components have estimated weights close to zero, with component $111$ having sometimes a small positive weight in the winter 2017-2018 and in spring of 2019. Copula coefficients were also estimated. Point estimates are all different than zero for all mixture components, however credible intervals are a lot narrower for those components with high weight. Largest $\btheta$ values correspond to mixture $110$, but the uncertainty is too high (wide credible intervals) to be trusted. 

Joint density estimates are presented as pairwise heat maps in Figure \ref{fig:joint_predict_real_PM25_NO2_SO2}. We only show the last trimester of 2017 and each row corresponds to a different pair, as follows: $(PM_{2.5},NO_2)$ in the top row, $(PM_{2.5},SO_2)$ in the middle row, and $(NO_2,SO_2)$ in the bottom. In all bivariate density estimates we see the presence of tail dependence in the $00$ (lower-lower) corner, whereas in the second and third row we see a $01$ (lower-upper) tail dependence. This corresponds to $SO_2$ which is plotted in the y-axis and has a negative dependence with the other two variables $PM_{2.5}$ and $CO_2$. This combination of tail dependencies makes the heatmaps show a spider effect towards the bottom-left and upper-left corners in September and November of 2017. Moreover, in October of 2017, for the pairs $(PM_{2.5},SO_2)$ and $(NO_2,SO_2)$, a combination of three trends can be identified with three tail dependencies in the $00$, $01$ and $11$ corners. 

Certainly, this non standard dependence among these three pollution variables could not be captured by a single copula model. 

\section{Concluding remarks and future work}
\label{sec:conclusion}

We  extend a copula's versatility in capturing dependence patterns using  mixtures of copulas with a dynamic component in the weights. The idea is illustrated using  Clayton copulas, but it can be applied to any other families. The motivation is given by problems where different extreme regions of the paiwise bivariate distributions exhibit patterns that cannot be captured by a single copula. 

In situations in which the dependence varies in time, we propose a dynamic mixture of copulas model in which the mixture weights and the parameters of the copulas involved in the mixture are modelled either through a moving average or a seasonal dynamic. The resulting increase in modelling flexibility is illustrated by all our numerical experiments, be they synthetic or real. 

Dependence in our dynamic Dirichlet prior on the weights is controlled by the triplet $(a_t,q,p)$. For the analyses considered here we have kept $a_t$ to be constant across time, to make our prior easy to define. However, this parameter can certainly be chosen to be different across time, providing additional flexibility in the model. We have left this substantial generalization for future work. 

Additionally, special care has to be put for the analysis of multivariate data for large dimension $m$, since the number of components in the mixture increases exponentially at a rate $2^m$. In practice, not all extremes are likely to be significant. In order to impose sparsity, we plan to exploit a sparse prior designed to reduce the number of components needed to model the data and still maintain the added flexibility demonstrated in this work. An added question of interest is the identification of lower dimensional manifolds where a mixture of lower-dimensional copulas can be used to capture the dependence in the data.

\section*{Acknowledgements}

This work was supported by \textit{Asociaci\'on Mexicana de Cultura, A.C.} while the second author was visiting the Department of Statistical Sciences at the University of Toronto and by an NSERC of Canada discovery grant of the third author. We are also grateful to Jun Young Park and Patrick Brown for their guidance on potential applications.

\bibliographystyle{abbrv}

\newpage

\begin{table}
\centering
\begin{tabular}{cccccccccc}
  \hline\hline
 \small{$e_k,g_k$} & $a_t$ & $q=0$ & $q=1$ & $q=2$ & $q=3$ & $q=4$ & $q=5$ & $q=6$ & $q=7$ \\  \hline
\multicolumn{10}{c}{LPML}\\ \hline
 1 & 0 & 1685 & \\ 
 1 & 1 & 1686 & 1683 & 1684 & 1688 & 1687 & 1692 & 1691 & 1696 \\ 
 1 & 10 & 1693 & 1698 & 1702 & 1702 & 1706 & 1705 & 1707 & 1706 \\ 
 1 & 20 & 1696 & 1702 & 1705 & 1705 & 1707 & 1706 & 1707 & 1706 \\ 
 1 & 30 & 1699 & 1704 & 1706 & \textbf{1708} & 1705 & 1708 & 1706 & \textbf{1708} \\ 
 1 & 40 & 1702 & 1705 & 1704 & 1708 & 1706 & 1705 & 1705 & 1705 \\ 
 1000 & 30 & 1682 & 1688 & 1691 & 1693 & 1693 & 1691 & 1693 & 1692 \\
   \hline
\multicolumn{10}{c}{WAIC}\\ \hline
 1 & 0 & -3372 &  \\ 
 1 & 1 & -3374 & -3368 & -3370 & -3378 & -3375 & -3384 & -3384 & -3393 \\ 
 1 & 10 & -3388 & -3397 & -3403 & -3405 & -3413 & -3410 & -3413 & -3411 \\ 
 1 & 20 & -3393 & -3405 & -3409 & -3411 & -3415 & -3413 & -3414 & -3412 \\ 
 1 & 30 & -3398 & -3409 & -3411 & \textbf{-3416} & -3410 & -3415 & -3413 & \textbf{-3416} \\ 
 1 & 40 & -3403 & -3410 & -3409 & -3416 & -3412 & -3410 & -3409 & -3411 \\ 
 1000 & 30 & -3364 & -3376 & -3382 & -3386 & -3386 & -3382 & -3385 & -3385 \\
 \hline \hline
\end{tabular}
\caption{Simulated data. LPML and WAIC gof values of different ${\cal{M}}_{a_t,q}$ models. Best values are bolded.}
\label{tab:lpml_waic_sim}
\end{table}

\begin{table}
    \centering
    \begin{tabular}{lrrr} \hline \hline
      Model / Measure & LPS & LPML & WAIC \\
        \hline
       ${\cal{M}}_{2,2}$ & 73.2 & 1603 & -3209 \\
        ${\cal{M}}_{20,7}$ & \textbf{74.9} & \textbf{1625} & \textbf{-3250} \\
        ${\cal{M}}_{30,7}$ & 73.1 & 1624 & -3247 \\
        D.Gaussian & 25.4 &  146 & -292 \\
        S.Clayton & -12.3 & 9  & -18 \\
        \hline \hline
    \end{tabular}
    \caption{Simulated data. LPS statistic using times $1,\ldots,19$ for fitting and $t=20$ for prediction. The other gof measures, LPML and WAIC, were calculated with the fitting data.}
    \label{tab:sim_model_comparison}
\end{table}

\begin{table}
\centering
\begin{tabular}{c|ccccc|ccccc}
\hline \hline 
 & \multicolumn{5}{c|}{LPML} & \multicolumn{5}{c}{WAIC} \\ \hline
\small{$a_t$} & \scriptsize{$MA(0)$} & \scriptsize{$MA(1)$} & \scriptsize{$MA(2)$} & \scriptsize{$MA(3)$} & \scriptsize{$MA(4)$} & \scriptsize{$MA(0)$} & \scriptsize{$MA(1)$} & \scriptsize{$MA(2)$} & \scriptsize{$MA(3)$} & \scriptsize{$MA(4)$} \\ \hline 
 & \multicolumn{10}{c}{$S(0)$}\\ \hline
 0 & 435 & & & & & -871 & & & & \\
 1 & 433 & \textbf{438} & 436 & 434 & 434 & -867 & \textbf{-875} & -873 & -868 & -867 \\ 
 2 & 433 & 433 & 433 & 426 & 422 & -866 & -867 & -866 & -852 & -843 \\ 
 3 & 432 & 432 & 428 & 421 & 416 & -864 & -863 & -857 & -843 & -832 \\ 
 4 & 432 & 427 & 421 & 417 & 410 & -863 & -855 & -841 & -834 & -820 \\ 
 5 & 429 & 424 & 417 & 409 & 403 & -858 & -849 & -834 & -817 & -805 \\ \hline 
 & \multicolumn{10}{c}{$S(1)$}\\ \hline
 0 & 435 & & & & & -871 & & & & \\  
 1 & 440 & \textbf{441} & 439 & 435 & 433 & -880 & \textbf{-881} & -878 & -869 & -866 \\ 
 2 & 439 & 435 & 433 & 427 & 422 & -877 & -871 & -865 & -854 & -843 \\ 
 3 & 434 & 432 & 427 & 420 & 414 & -868 & -865 & -854 & -839 & -829 \\ 
 4 & 432 & 428 & 422 & 414 & 408 & -863 & -857 & -844 & -829 & -815 \\ 
 5 & 427 & 424 & 416 & 408 & 401 & -854 & -848 & -832 & -815 & -802 \\ \hline 
 & \multicolumn{10}{c}{$S(2)$}\\ \hline
 0 & 435 & & & & & -871 & & & & \\  
 1 & 440 & \textbf{442} & 438 & 433 & 430 & -879 & \textbf{-883} & -876 & -867 & -859 \\
 2 & 437 & 433 & 433 & 427 & 420 & -874 & -865 & -865 & -854 & -840 \\ 
 3 & 434 & 432 & 427 & 419 & 413 & -867 & -864 & -853 & -838 & -826 \\ 
 4 & 430 & 425 & 424 & 412 & 407 & -861 & -850 & -847 & -824 & -813 \\ 
 5 & 426 & 422 & 418 & 408 & 397 & -852 & -844 & -836 & -817 & -794 \\ \hline \hline
\end{tabular}
\caption{Bivariate real data. LPML and WAIC statistics for different prior choices for ${\cal{M}}_{a_t,q,p}$.}
\label{tab:LPML_WAIC_real}
\end{table}

\begin{table}
\centering
\begin{tabular}{lrrr} \hline \hline
Model / Measure & MSE & LPML & WAIC \\ \hline
${\cal{M}}_{1,1,2}$ & \textbf{0.0751} & \textbf{397} & \textbf{-794} \\
D.Gaussian & 0.0763 & 383 & -767 \\
S.Clayton & 0.0826 & 39 & -78 \\
\hline \hline
\end{tabular}
\caption{Bivariate real data. Goodness of fit measures in first validation study.}
\label{tab:validation_real}
\end{table}

\begin{table}
\centering
\begin{tabular}{lrrrrrrrrrr} \hline\hline \\[-4mm]
Model & \small{$t=31$} & \small{$t=32$} & \small{$t=33$} & \small{$t=34$} & \small{$t=35$} & \small{$t=36$} & \text{Total}\\[1mm] \hline 
 ${\cal{M}}_{1,1,0}$ & -10.5 & -3.1 &  -9.1 & 17.8 & 21.6 & 11.4 & 28.2 \\ 
 ${\cal{M}}_{1,2,0}$ & -3.2 &-4.7 & -12.1 & \textbf{21.8} & 25.3 & 12.2 &\textbf{39.3} \\ 
 ${\cal{M}}_{1,1,1}$ & -12.4 &-3.8 & -10.9 & 16.1 & 24.0 & 11.9 & 24.9 \\ 
 ${\cal{M}}_{1,2,1}$ & -5.9 &-2.7 & -11.6 & 18.5 & 26.8 & 11.6 & 36.7 \\ 
 ${\cal{M}}_{1,2,2}$ & -8.2 &-3.1 & -10.0 & 15.4 & 26.8 &\textbf{12.3} & 33.2\\
 D.Gaussian & \textbf{5.2} &  \textbf{-1.8} &  \textbf{-1.3}  & -7.1 & \textbf{37.5} & 6.1 & 38.8 \\ 
 S.Clayton & -79.5 & -170.5 & -153.6 & -258.2 & -246.1 & -199.8 & -1107.6  \\
 \hline \hline
\end{tabular}
\caption{Bivariate real data. LPS statistic computed by fitting models from time $1$ to $t-1$ and predicting time $t$.} 
\label{tab:real_model_comparison}
\end{table}

\begin{table}
\centering
\begin{tabular}{c|ccccc|ccccc}
\hline \hline 
 & \multicolumn{5}{c|}{LPML} & \multicolumn{5}{c}{WAIC} \\ \hline
\small{$a_t$} & \scriptsize{$MA(0)$} & \scriptsize{$MA(1)$} & \scriptsize{$MA(2)$} & \scriptsize{$MA(3)$} & \scriptsize{$MA(4)$} & \scriptsize{$MA(0)$} & \scriptsize{$MA(1)$} & \scriptsize{$MA(2)$} & \scriptsize{$MA(3)$} & \scriptsize{$MA(4)$} \\ \hline 
 & \multicolumn{10}{c}{$S(0)$}\\ \hline
 0 & 271 &  &  &  & & -542 & & &  & \\ 
1 & 286 & 289 & 293 & 294 & 294 & -573 & -579 & -589 & -588 & -589 \\ 
2 & 281 & 278 & 278 & 278 & 276 & -562 & -557 & -557 & -556 & -553 \\ 
3 & 287 & 291 & 289 & 292 & 291 & -574 & -582 & -579 & -586 & -583 \\ 
4 & 290 & \bf{299} & 293 & 297 & 297 & -581 & \bf{-601} & -587 & -594 & -595 \\ 
5 & 282 & 282 & 281 & 282 & 282 & -565 & -565 & -563 & -565 & -565 \\ 
\hline 
 & \multicolumn{10}{c}{$S(1)$}\\ \hline
0 & 271 &  &  &  &  & -542 &  &  &  &  \\ 
1 & 288 & 288 & 291 & 295 & 287 & -576 & -577 & -583 & -590 & -575 \\ 
2 & 279 & 279 & 278 & 278 & 276 & -559 & -558 & -557 & -557 & -553 \\ 
3 & 289 & 289 & 292 & 292 & {294} & -579 & -579 & -584 & -586 & -589 \\ 
4 & 292 & 294 & \bf{298} & 297 & 293 & -586 & -591 & \bf{-597} & -594 & -587 \\ 
5 & 283 & 282 & 283 & 283 & 283 & -568 & -565 & -567 & -566 & -567 \\
\hline 
 & \multicolumn{10}{c}{$S(2)$}\\ \hline
 0 & 271 &  &  &  &  & -542 &  &  &  &  \\ 
1 & 287 & 286 & 289 & 295 & 279 & -575 & -572 & -578 & -590 & -562 \\ 
2 & 280 & 279 & 281 & 277 & 277 & -561 & -558 & -562 & -555 & -556 \\ 
3 & 290 & 291 & 291 & 292 & 292 & -579 & -582 & -582 & -586 & -585 \\ 
4 & 296 & 296 & 295 & 295 & \bf{301} & -594 & -593 & -591 & -590 & \bf{-603} \\ 
5 & 285 & 285 & 286 & 282 & 283 & -571 & -571 & -572 & -565 & -566 \\   \hline \hline
\end{tabular}
\caption{Multivariate real data. LPML and WAIC statistics for different prior choices for ${\cal{M}}_{a_t,q,p}$.}
\label{tab:LPML_WAIC_real2}
\end{table}


\begin{figure}[h]
\setlength{\unitlength}{1cm}
\hspace{-1cm}
\begin{center}
\begin{picture}(6,6)
\put(0.5,0.5){\vector(1,0){5.5}}
\put(0.5,0.5){\vector(0,1){5.5}}
\put(0.5,3.0){\line(1,0){5.0}}
\put(3.0,0.5){\line(0,1){5.0}}
\put(0.5,5.5){\line(1,0){5.0}}
\put(5.5,0.5){\line(0,1){5.0}}
\put(0.1,0.1){$0$}
\put(5.3,0.1){$1$}
\put(0.1,5.3){$1$}
\put(1.5,1.5){$00$}
\put(4.0,1.5){$10$}
\put(4.0,4.0){$11$}
\put(1.5,4.0){$01$}
\end{picture}
\end{center}
\caption{Unit square divided in four quadrants.}
\label{fig:rotations} 
\end{figure}
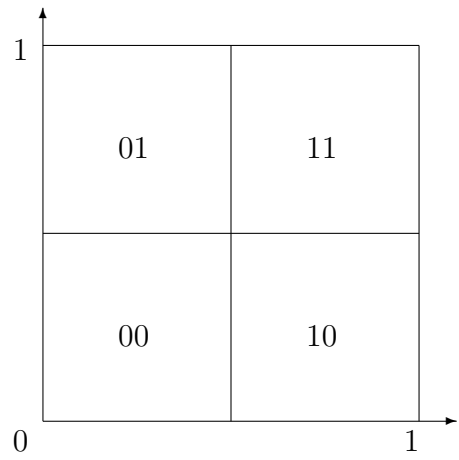

\begin{figure}
\centerline{\includegraphics[scale=0.5]{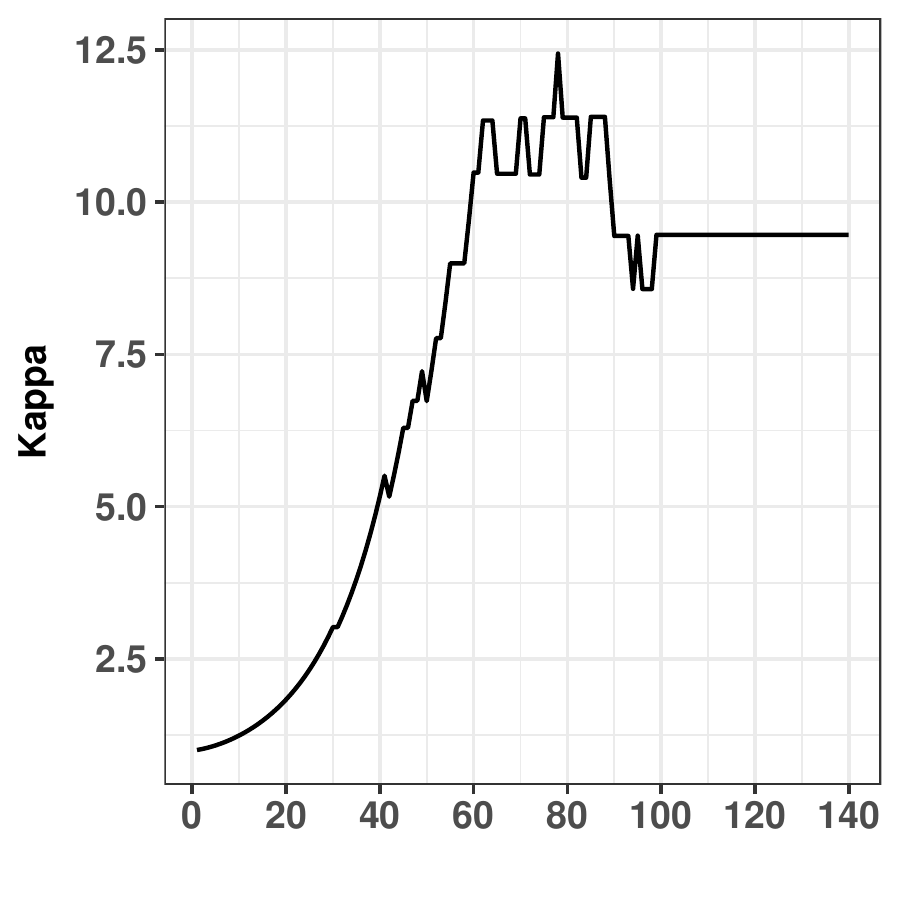}
\includegraphics[scale=0.5]{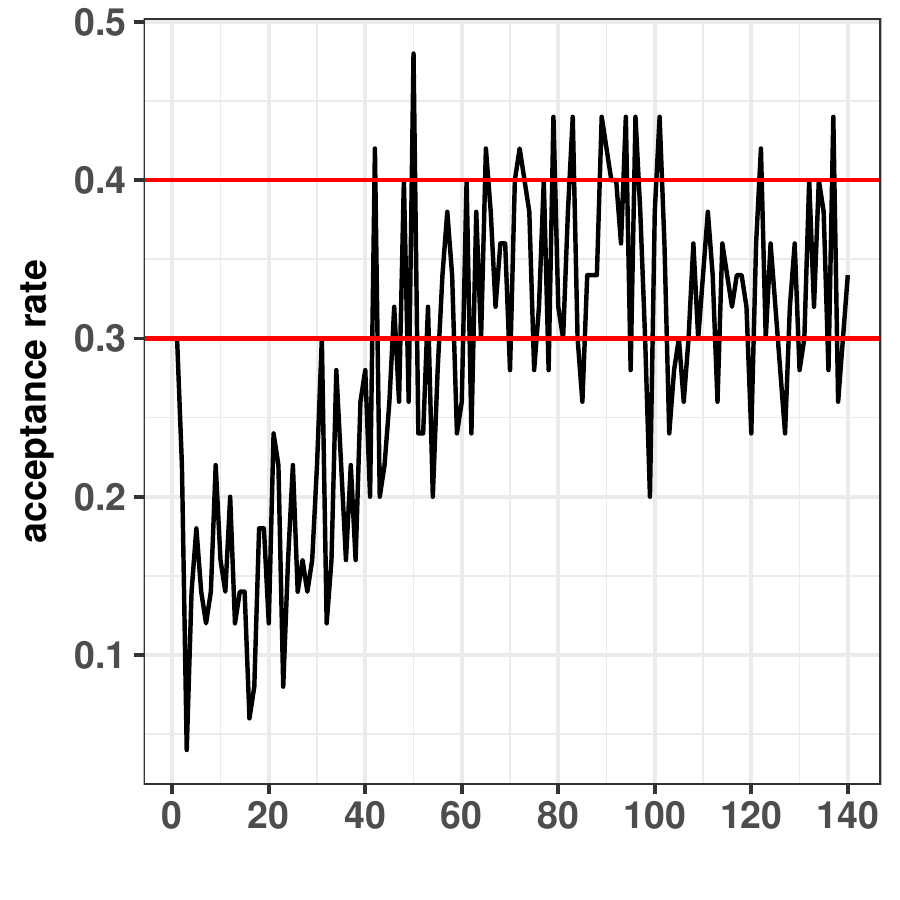}}
\caption{The recorded $\kappa^{(h)}$ and acceptance rate for each batch $h$. The batch size is $50$.}
\label{fig:kappa_acc_sim}
\end{figure}

\begin{figure}
 \centering
    \includegraphics[scale=0.32]{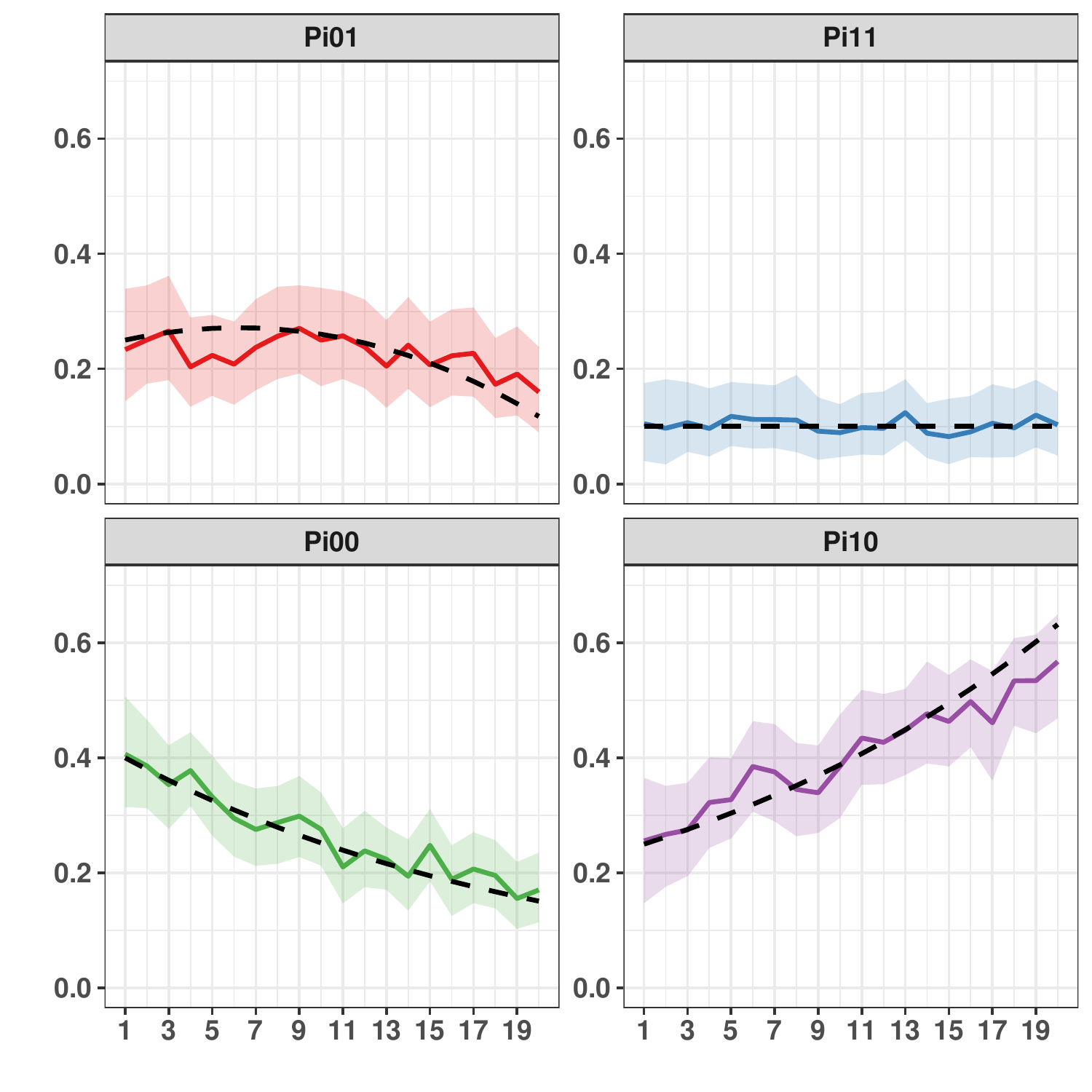}
    \includegraphics[scale=0.32]{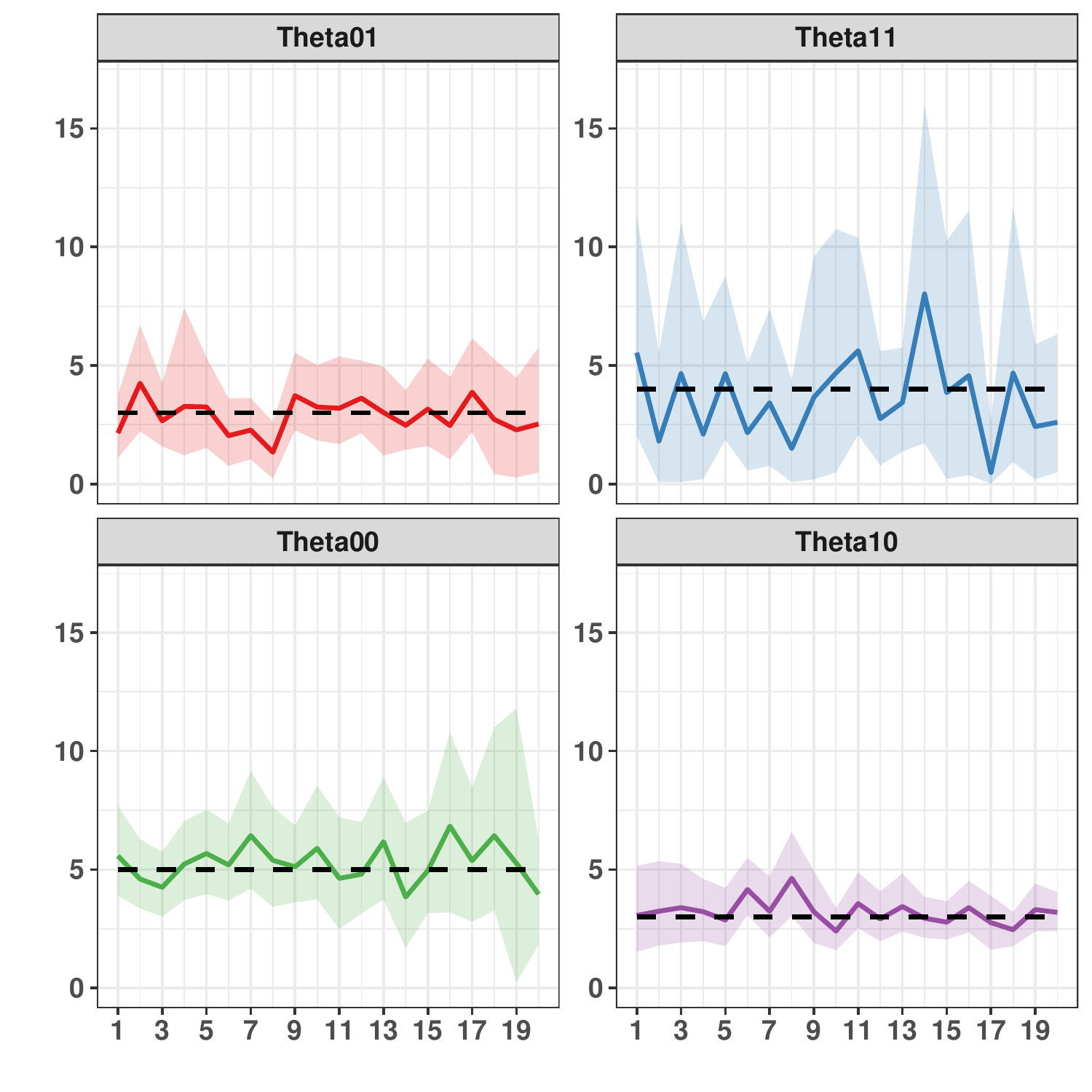}
\caption{Simulated data. Posterior estimates of $\bpi$ (left) and $\btheta$ (right):  posterior mean (solid line) with 95\% credible intervals (shadows), together with the true value (dotted black line).}
\label{fig:pi_theta_predict_sim}
\end{figure}

\begin{figure}
\centerline{\includegraphics[scale=0.5]{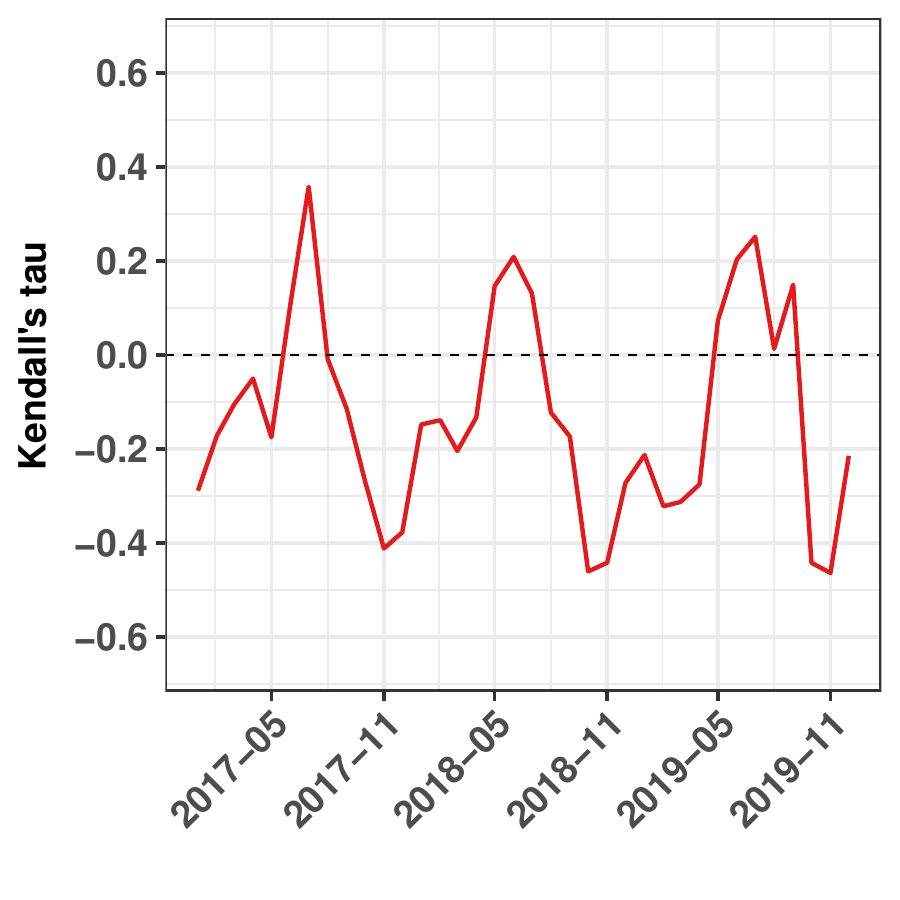}
\includegraphics[scale=0.5]{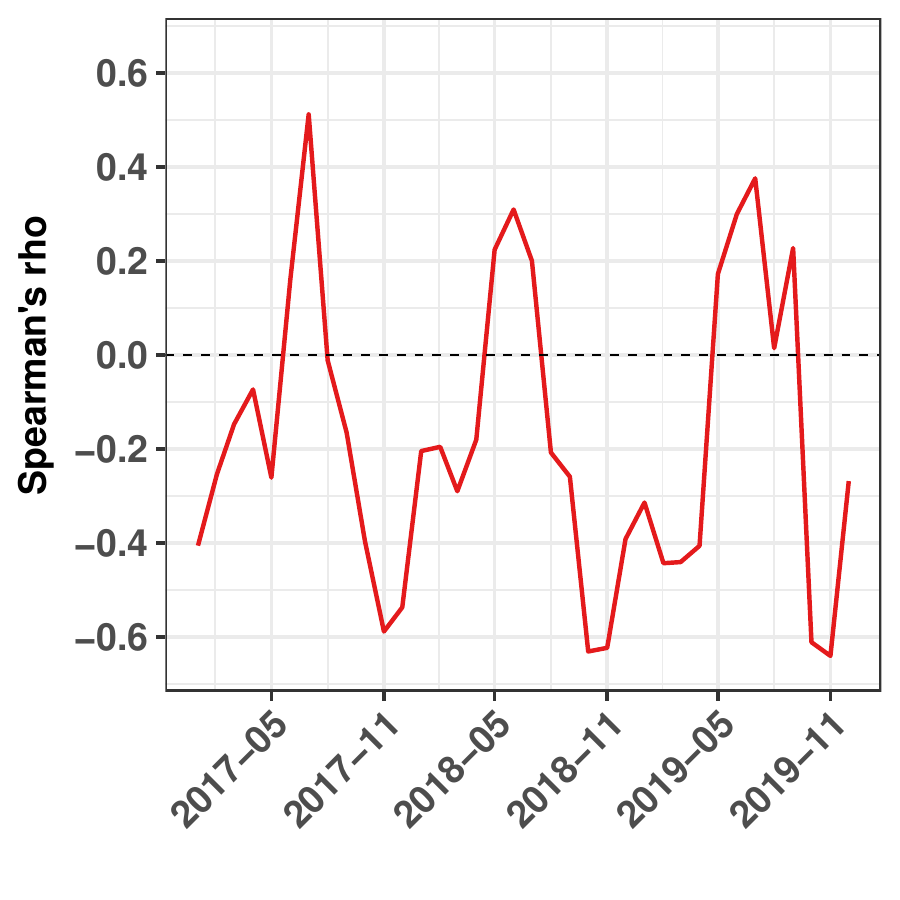}}
\caption{Bivariate real dataset. Empirical Kendall's tau (left) and Spearman's rho (right).}
\label{fig:kendall_spearman}
\end{figure}

\begin{figure}
    \centering
    \includegraphics[scale=0.3]{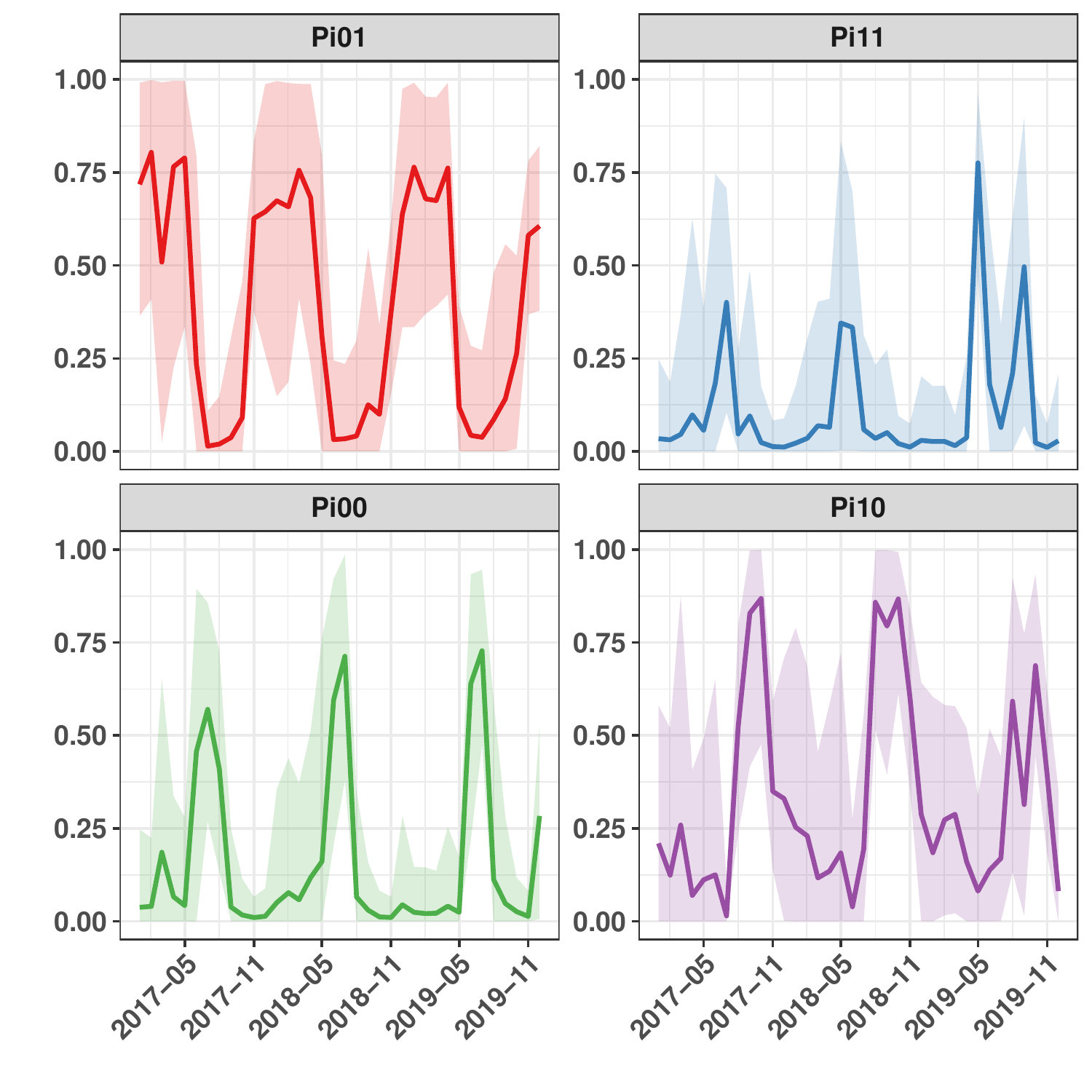}
    \includegraphics[scale=0.3]{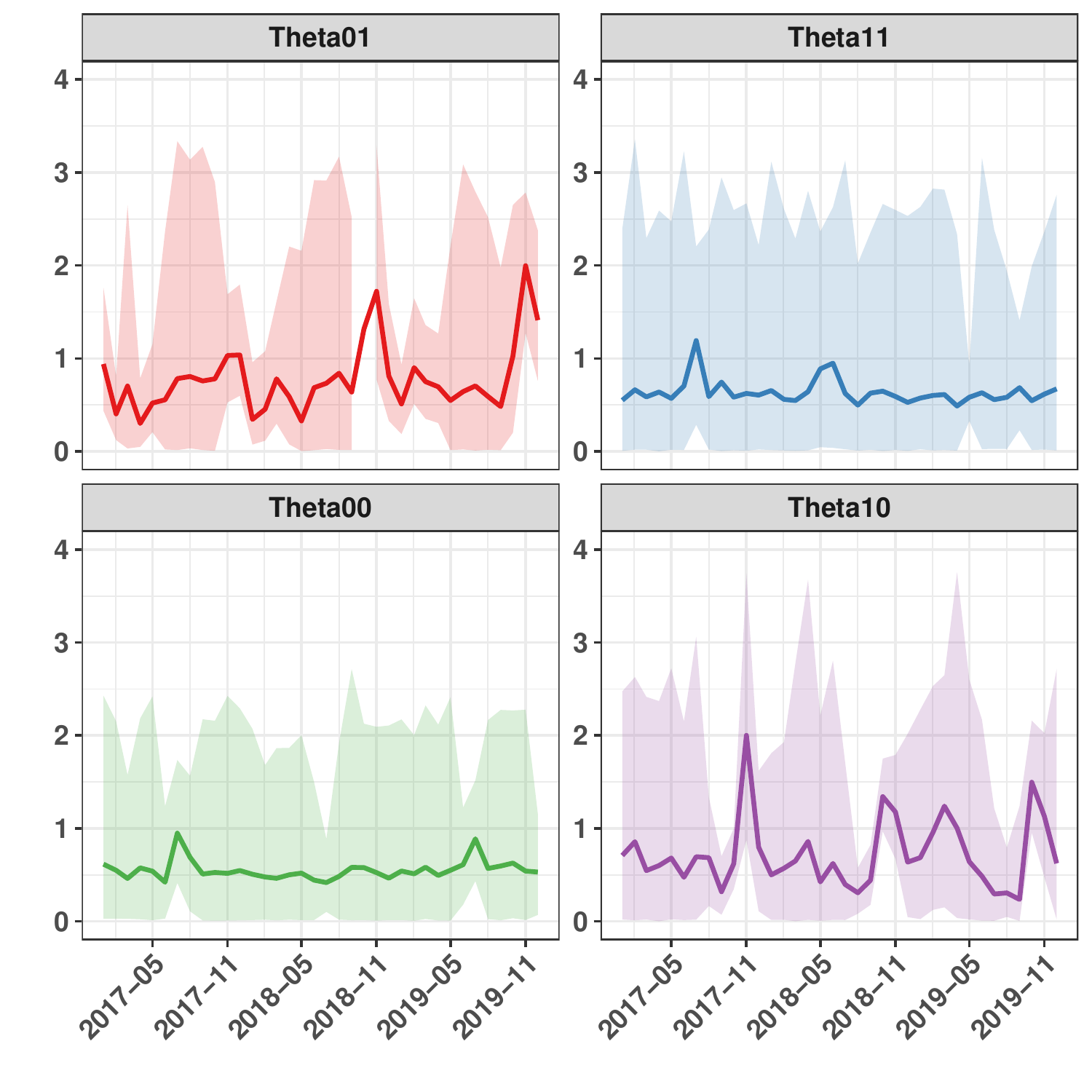}
    \caption{Bivariate real dataset. Posterior estimates of $\bpi$ (left) and $\btheta$ (right): posterior mean (solid line) with 95\% credible intervals (shadows).}
\label{fig:pi_theta_predict_real}
\end{figure}

\begin{figure}
    \centering
    \includegraphics[scale=0.5]{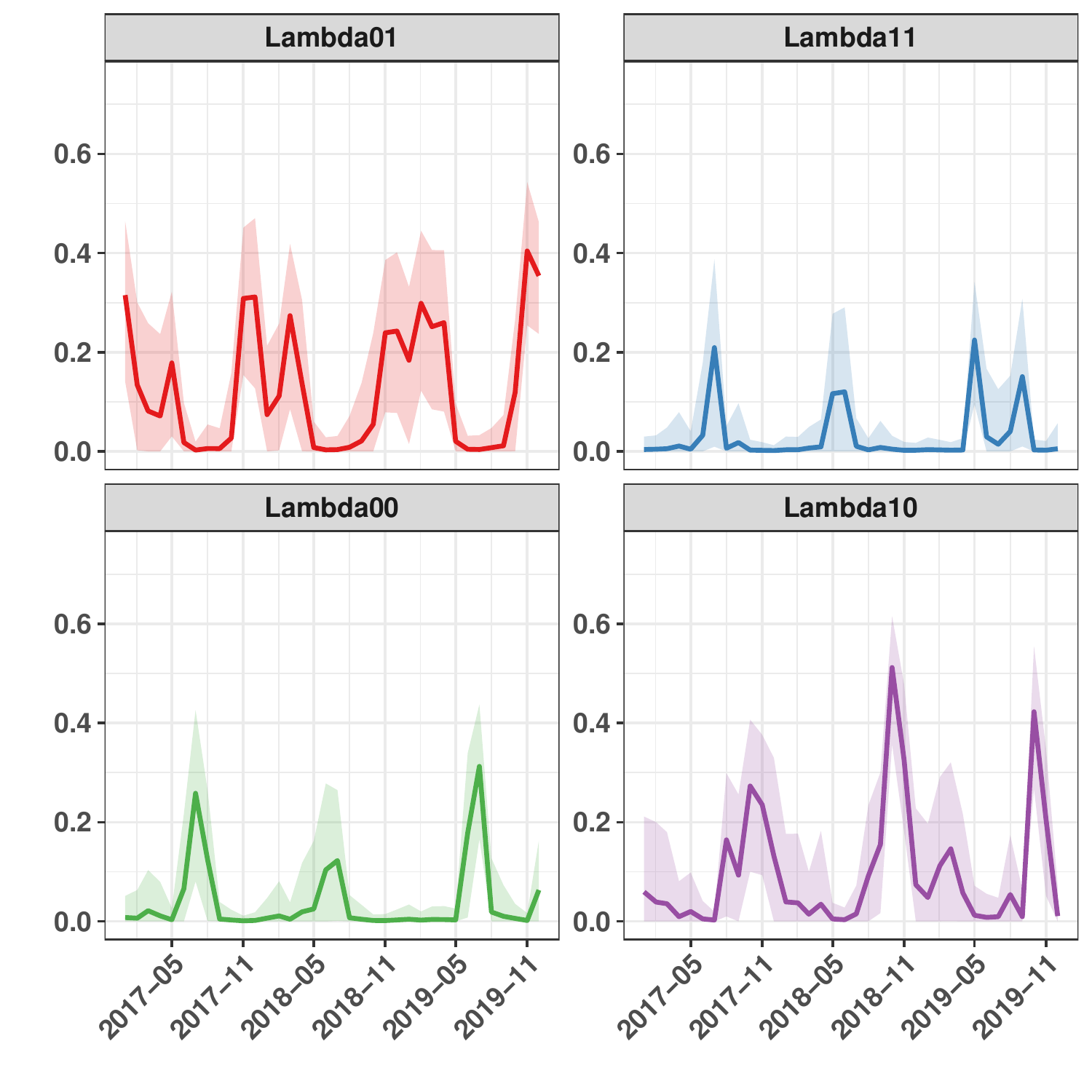}
    \caption{Bivariate real dataset. Posterior estimates of $\blambda$:  posterior mean (solid line) with 95\% credible intervals (shadows).}
\label{fig:tail_coefficient_predict_real}
\end{figure}

\begin{figure}
\centerline{
\includegraphics[scale=0.144]{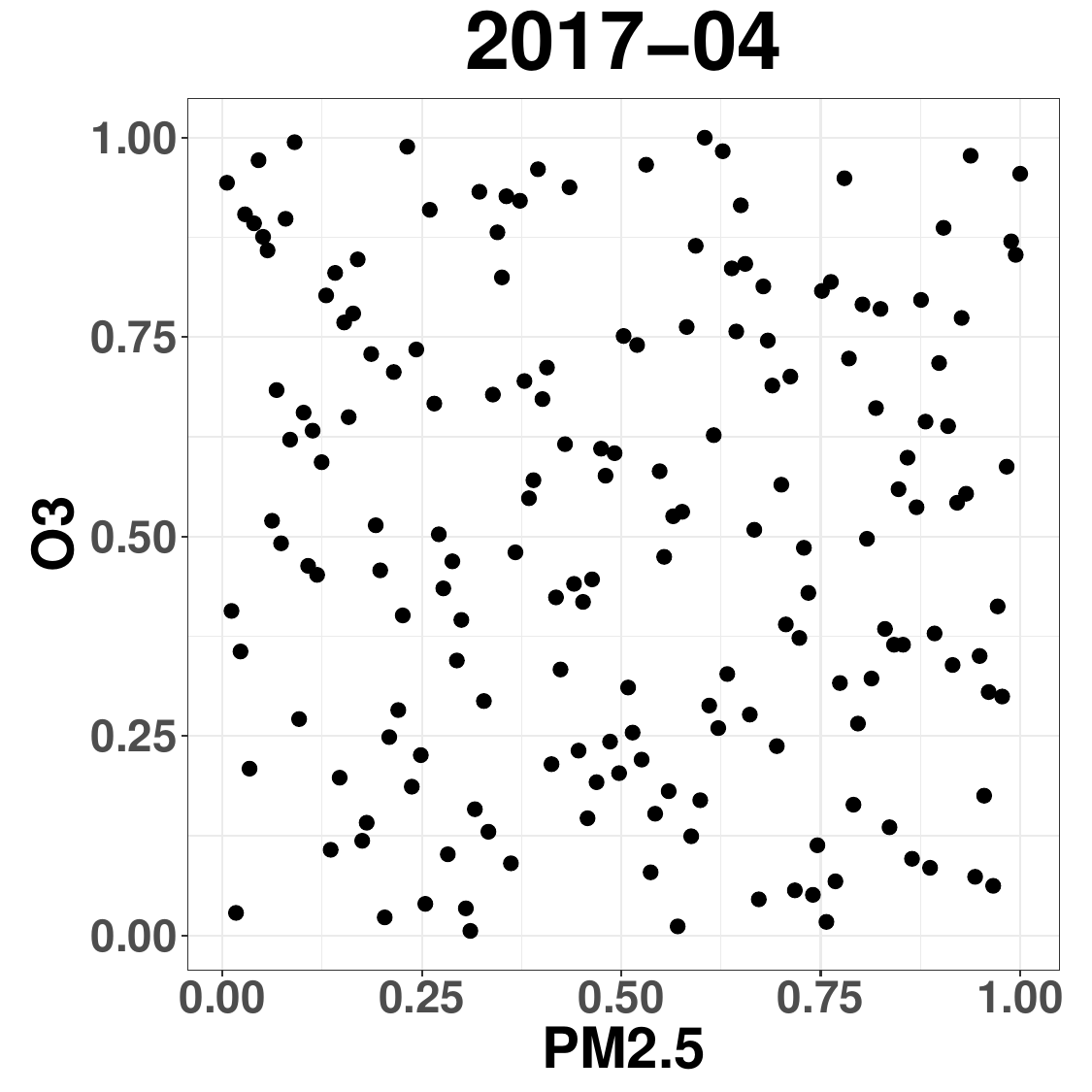}
\includegraphics[scale=0.144]{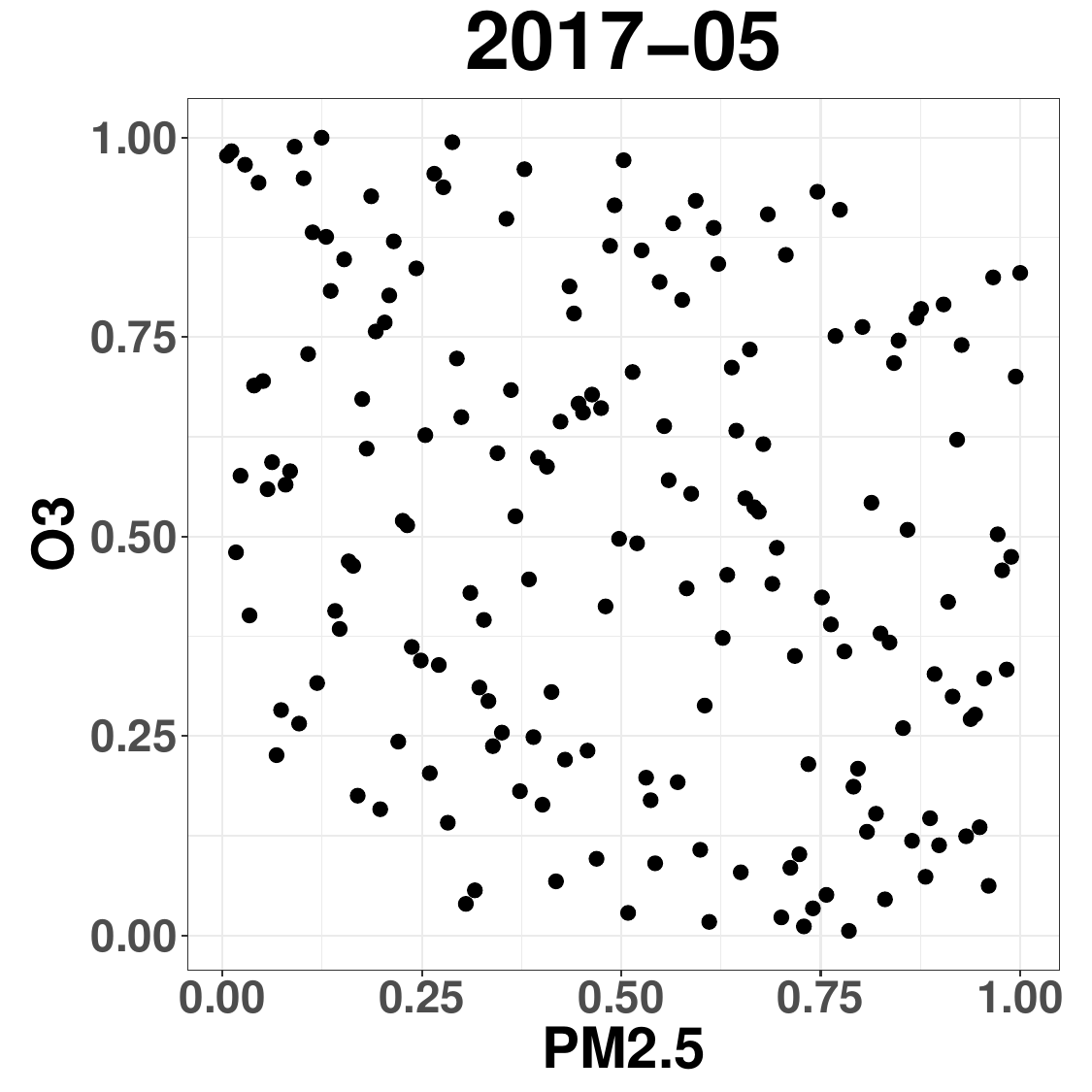}
\includegraphics[scale=0.144]{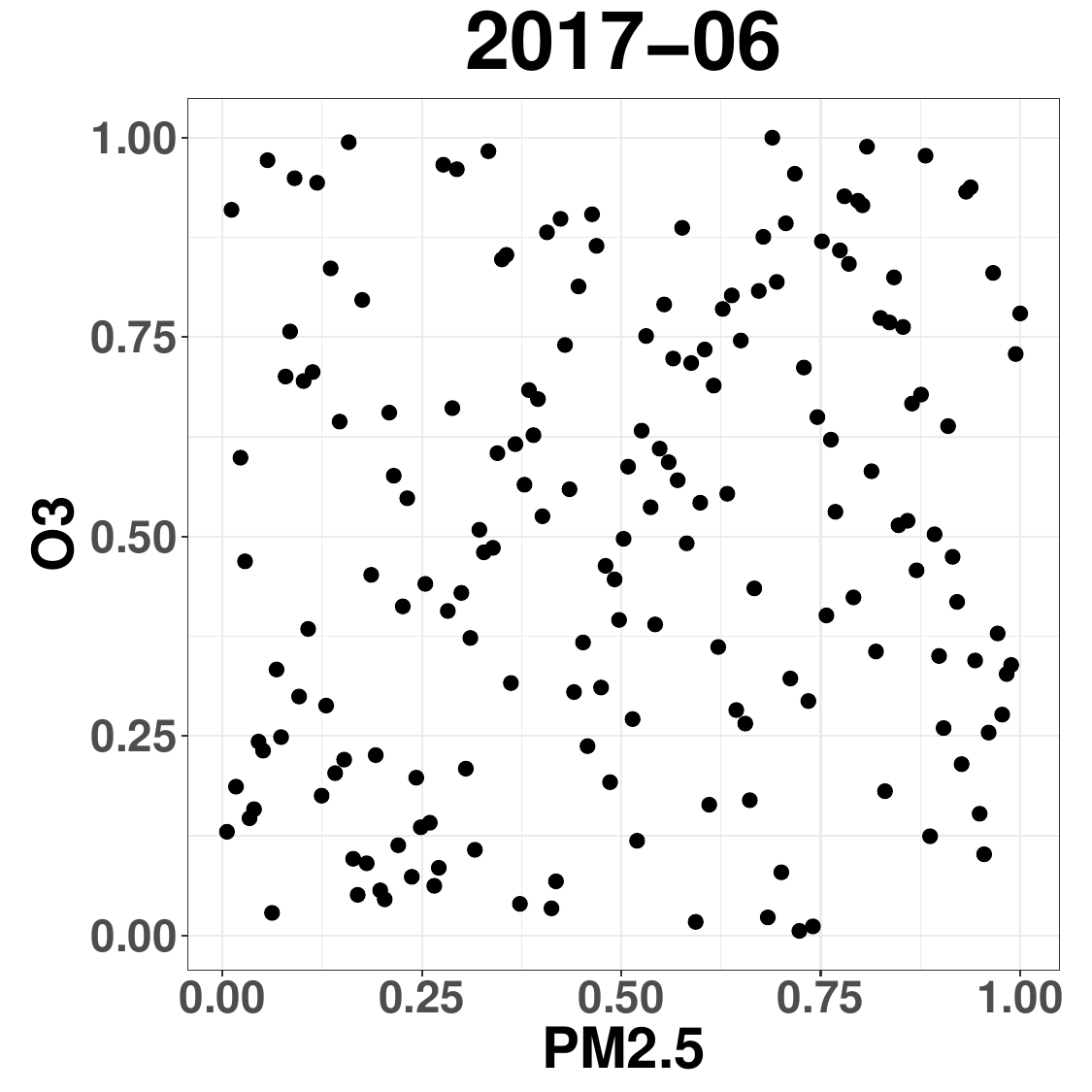}
\includegraphics[scale=0.144]{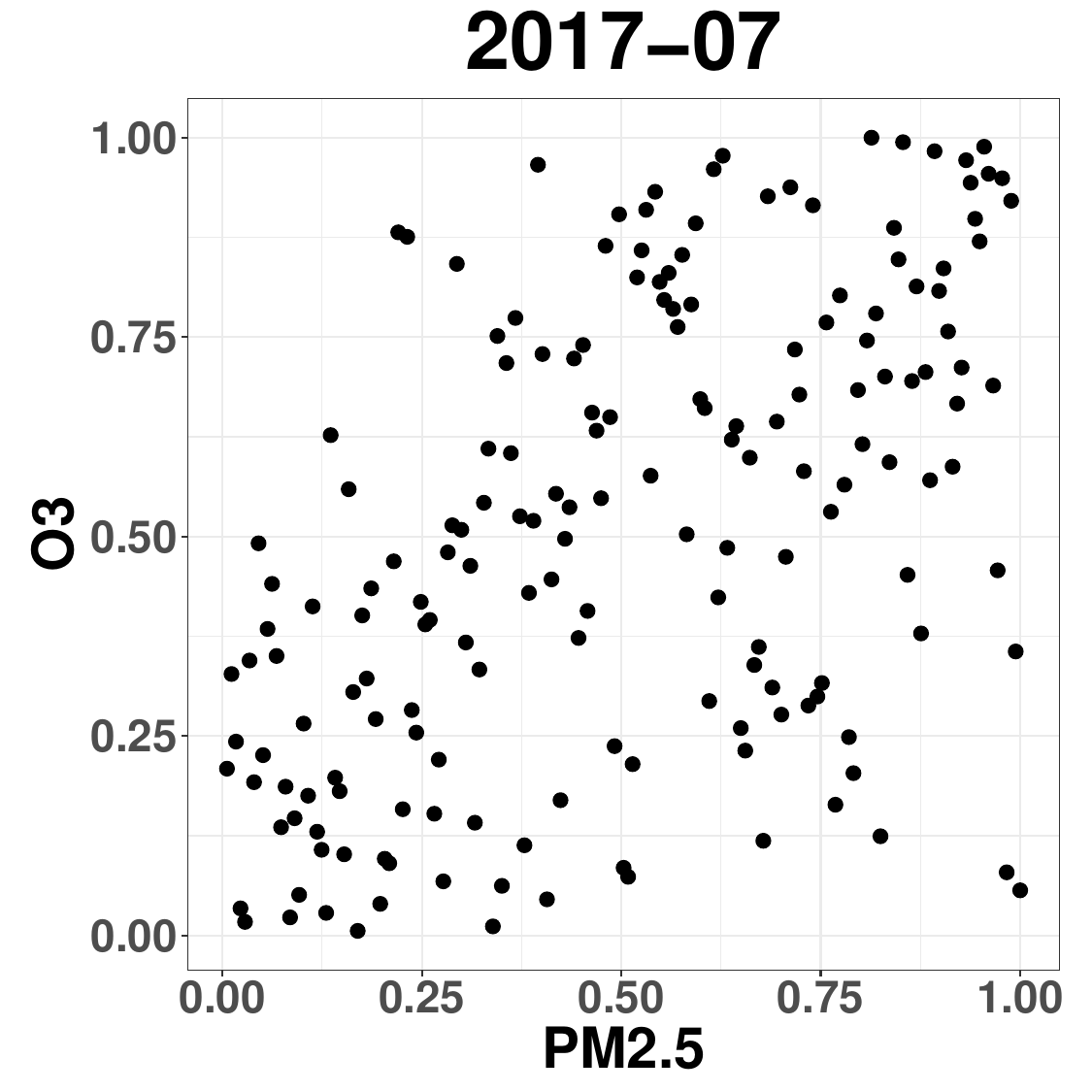}
\includegraphics[scale=0.144]{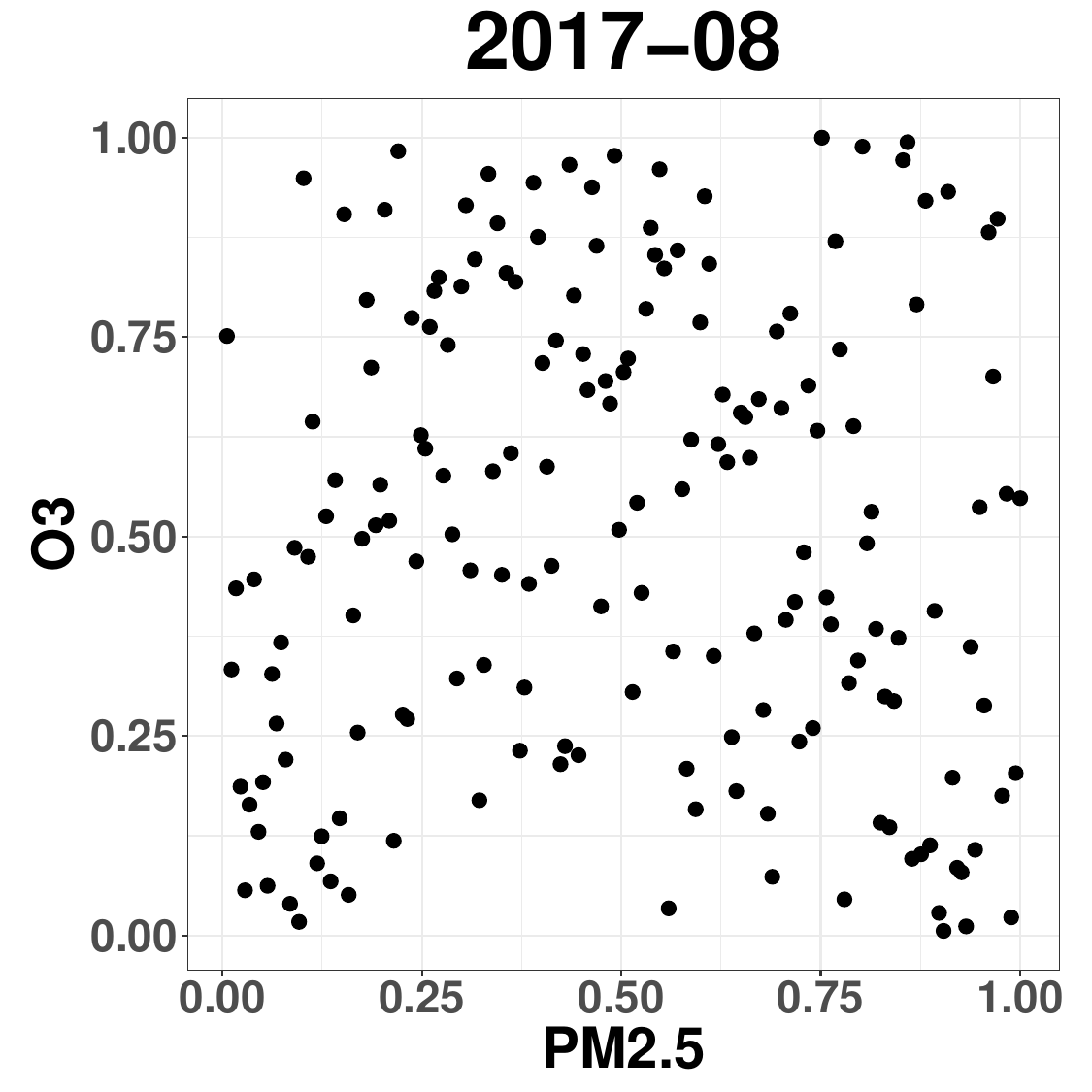}
\includegraphics[scale=0.144]{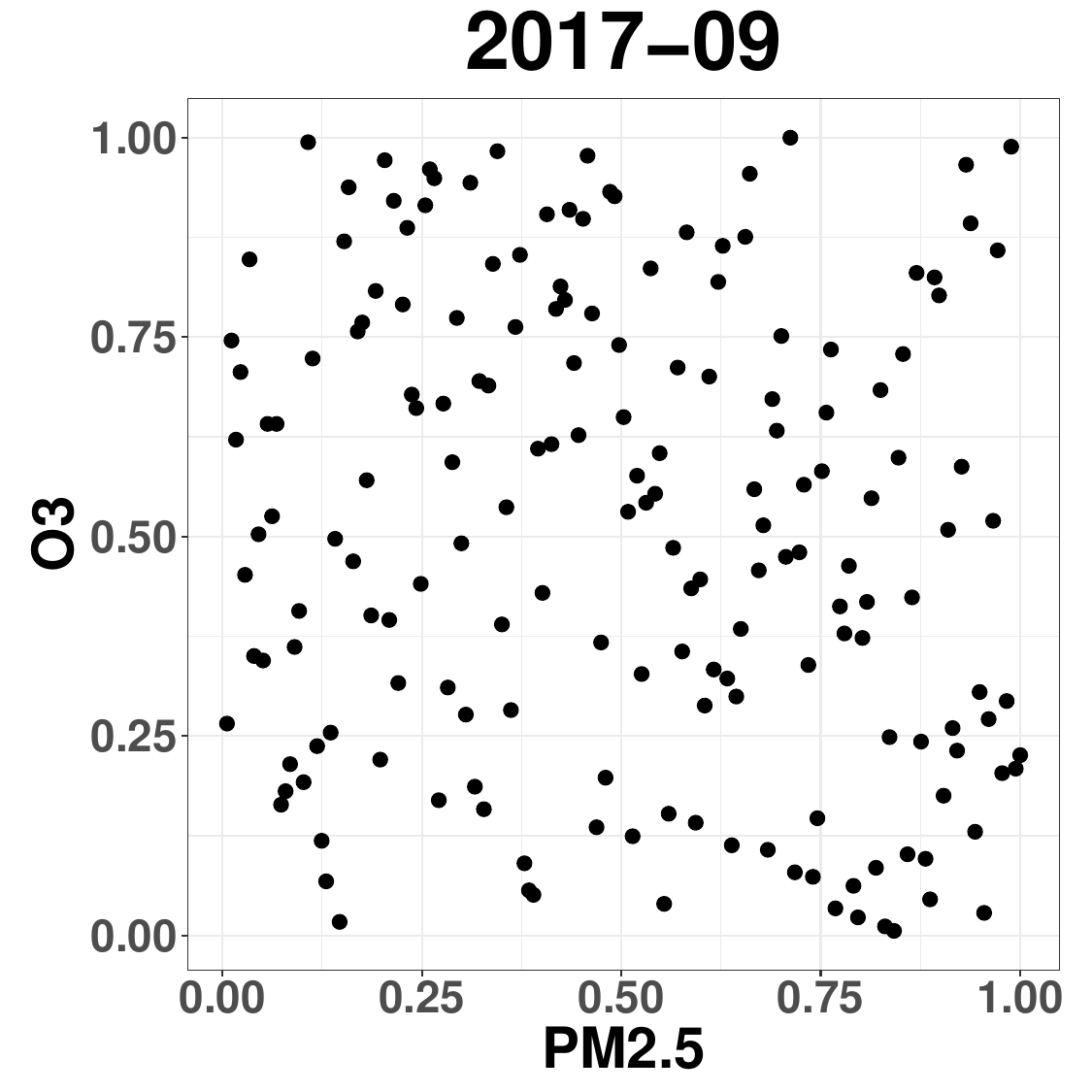}
}
\centerline{
\includegraphics[scale=0.144]{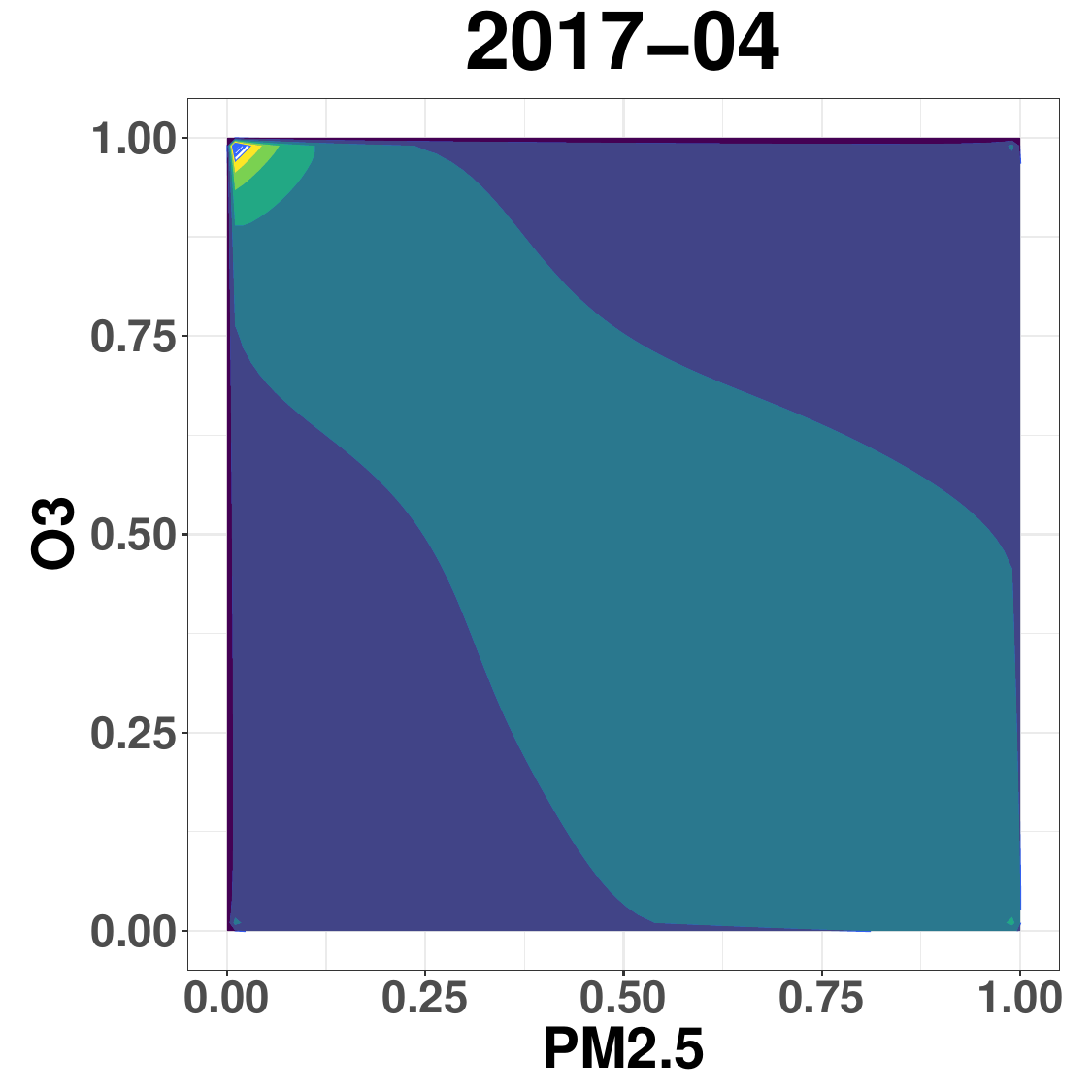}
\includegraphics[scale=0.144]{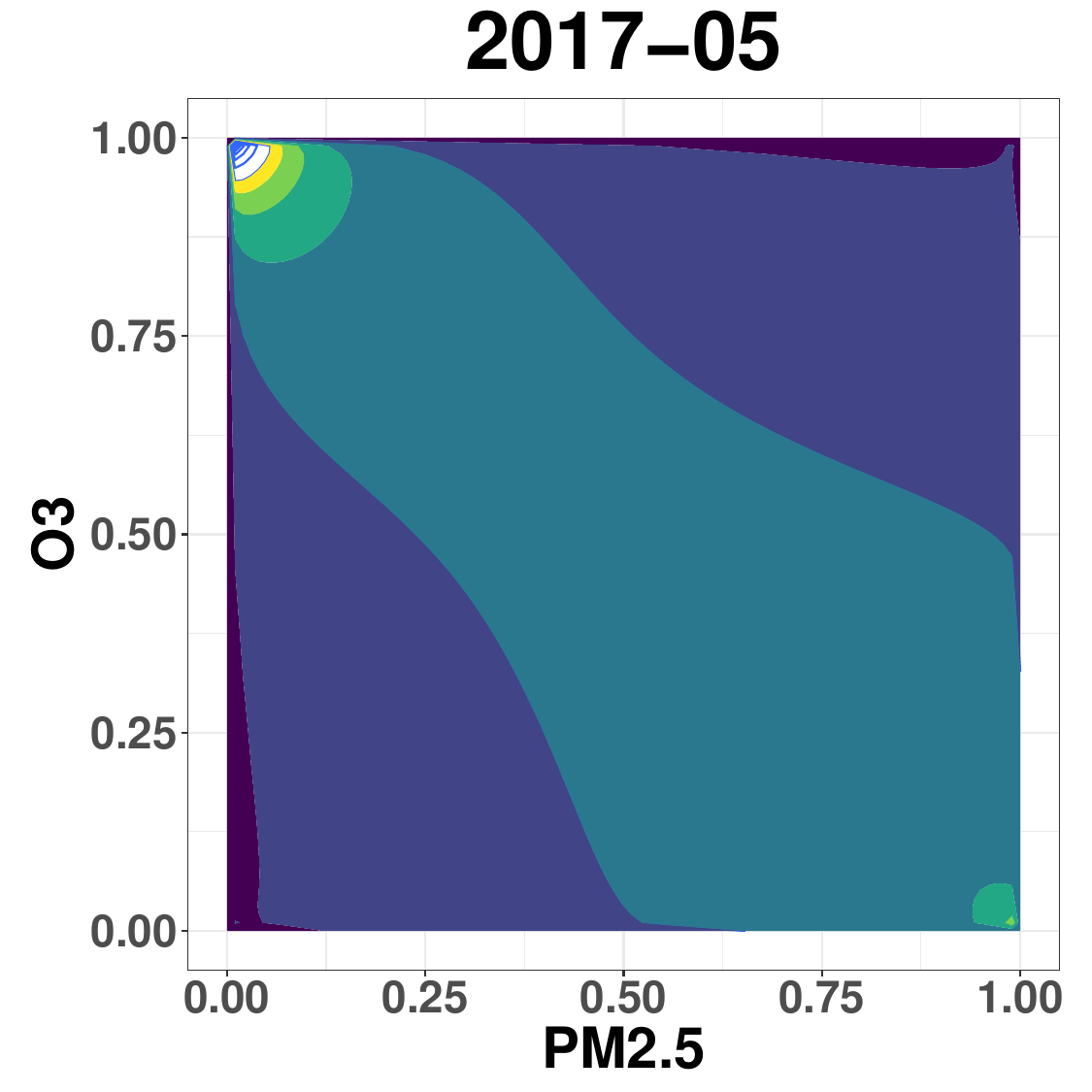}
\includegraphics[scale=0.144]{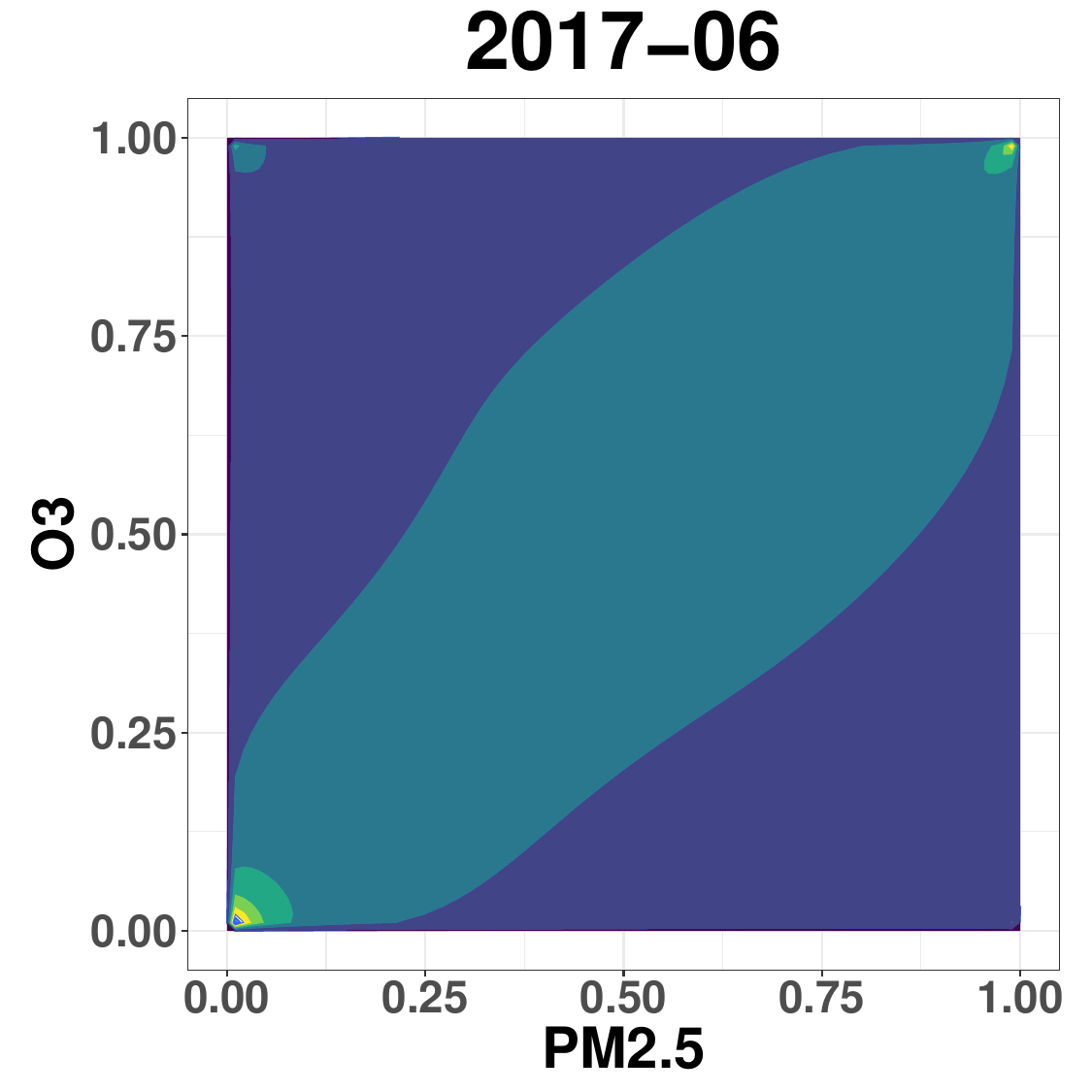}
\includegraphics[scale=0.144]{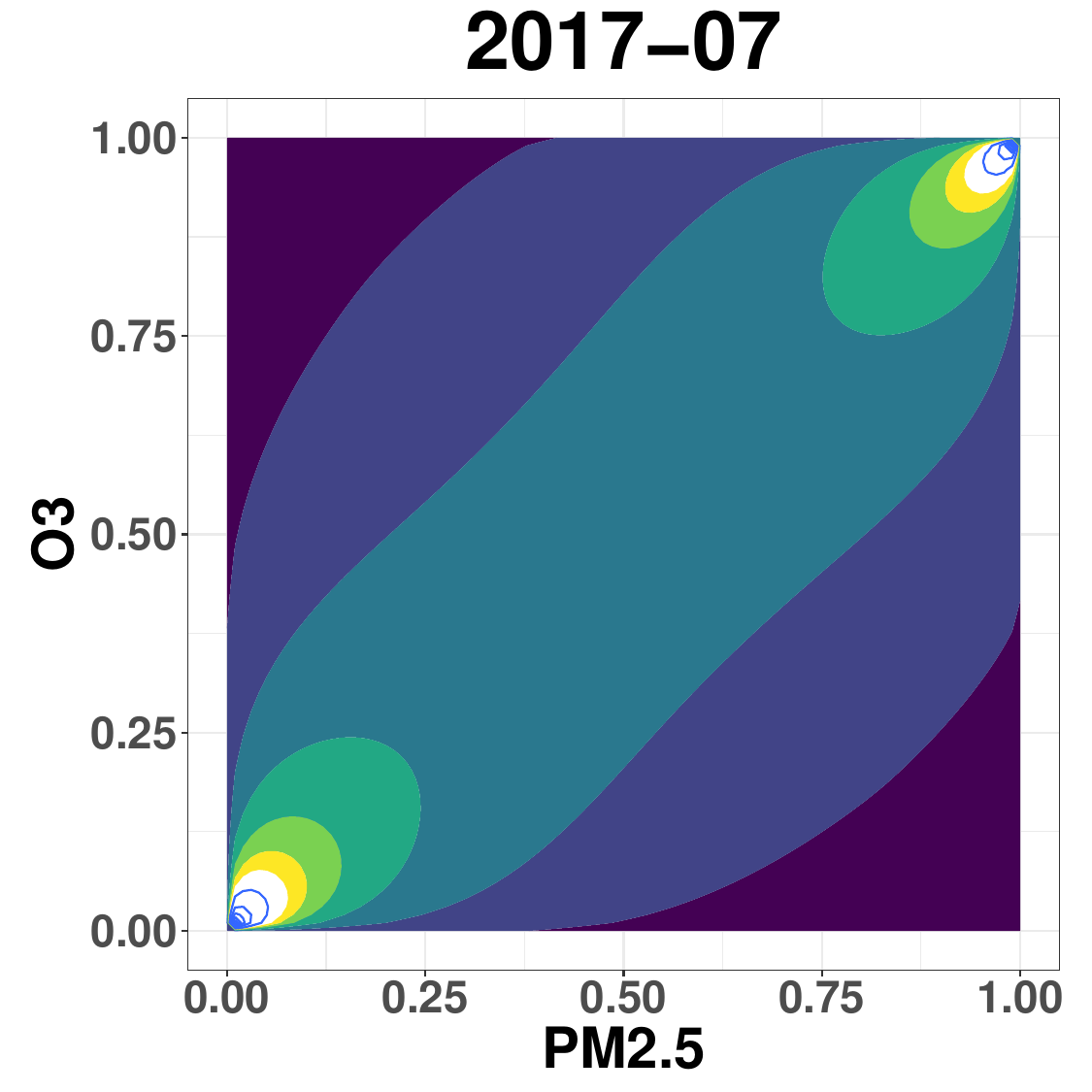}
\includegraphics[scale=0.144]{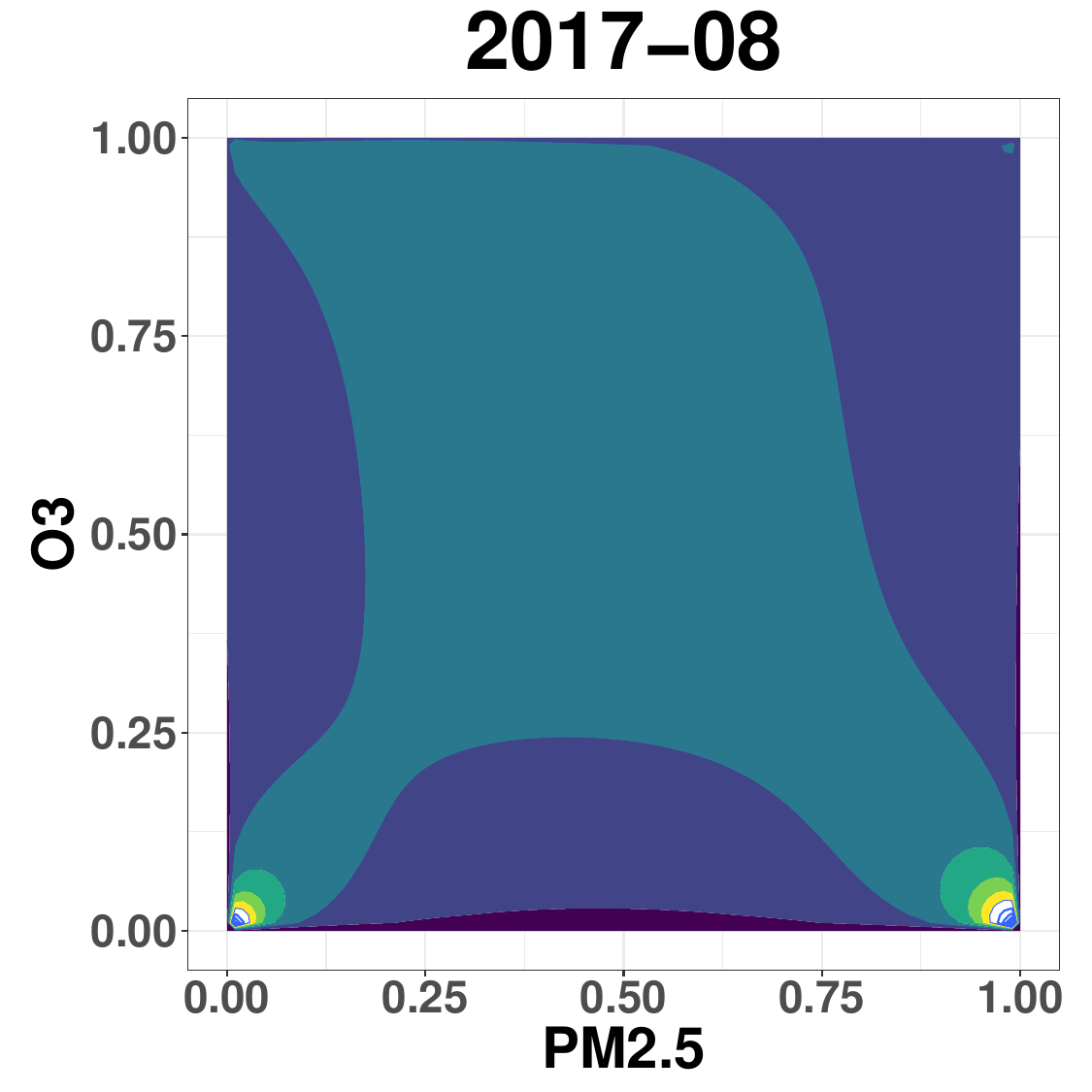}
\includegraphics[scale=0.144]{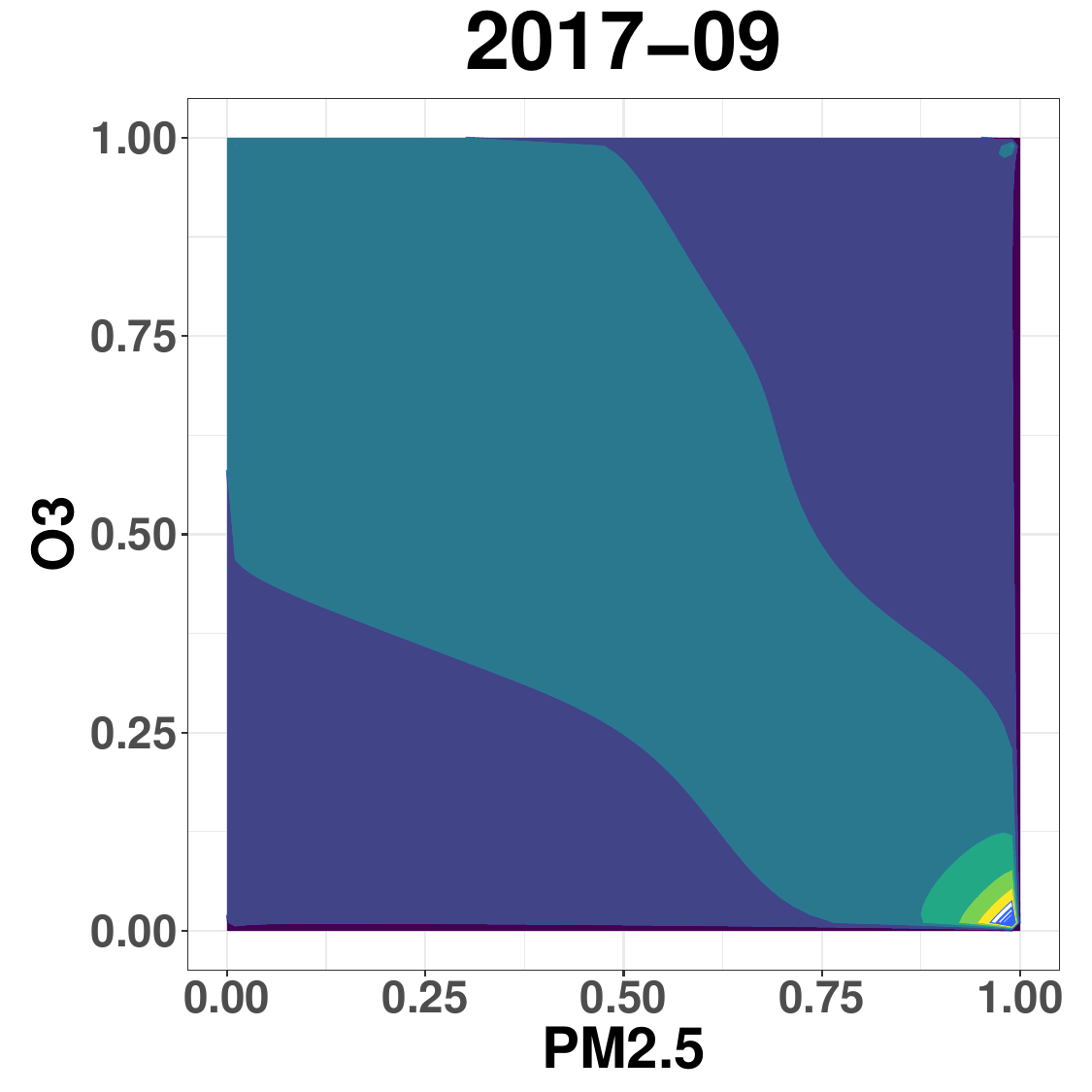}
}
\centerline{
\includegraphics[scale=0.144]{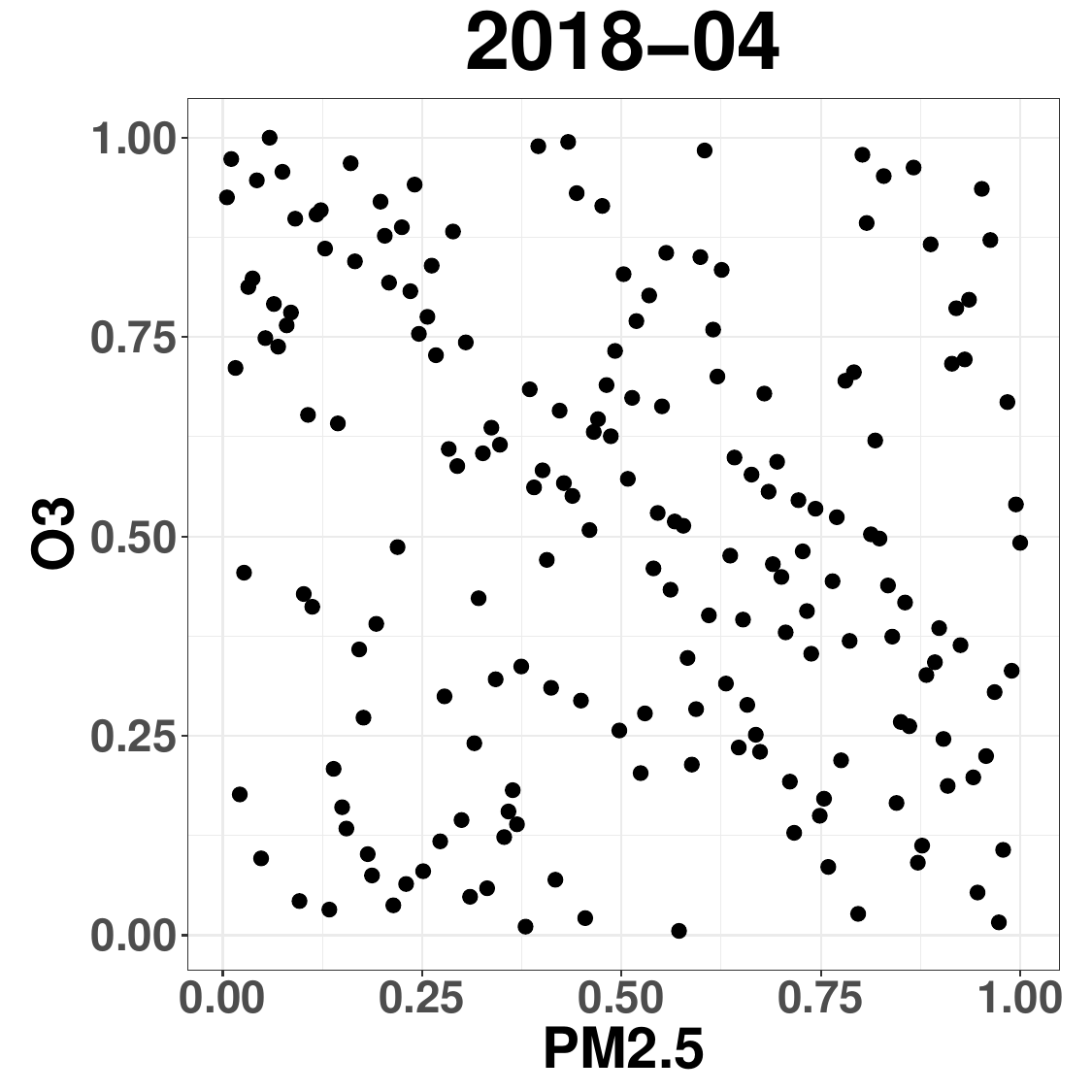}
\includegraphics[scale=0.144]{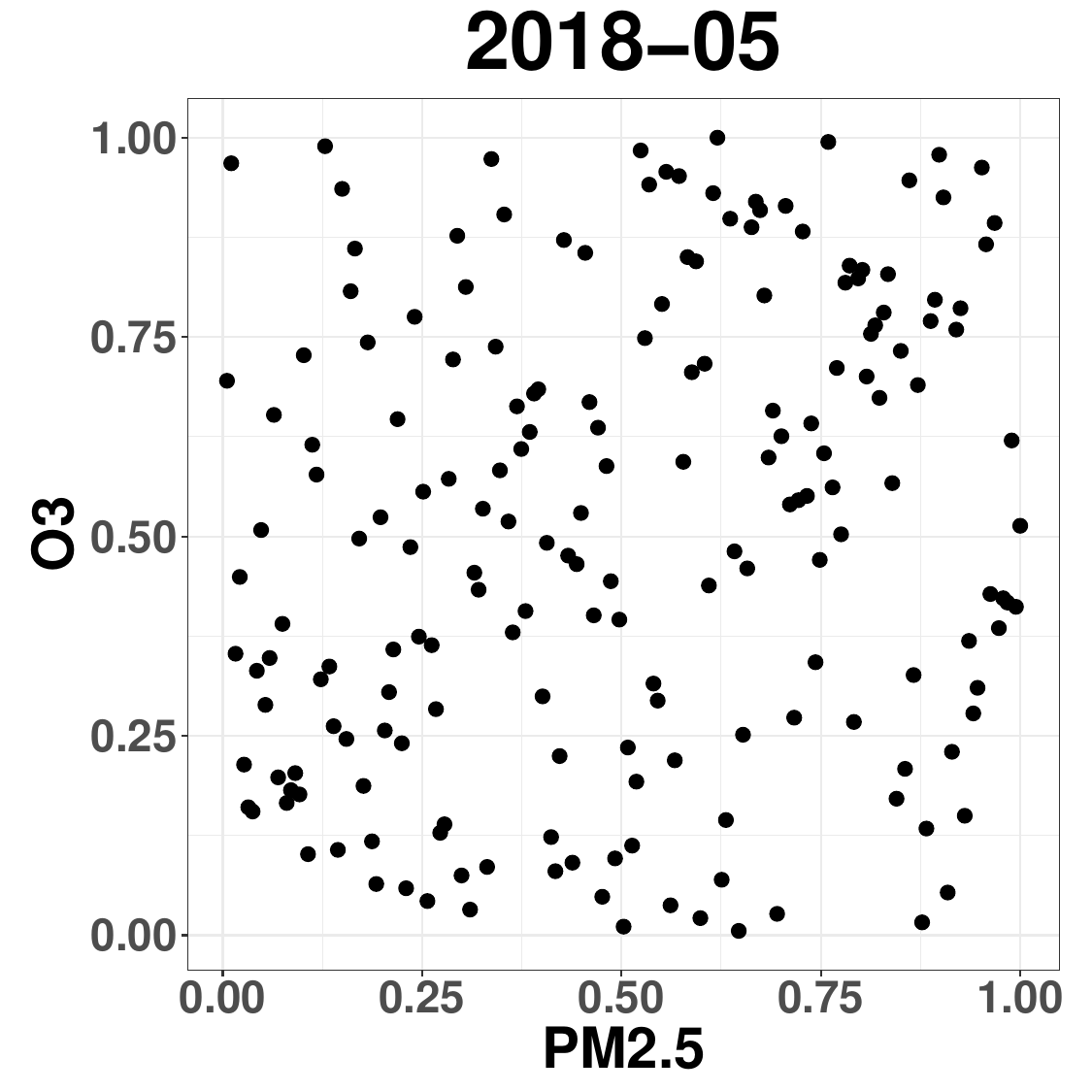}
\includegraphics[scale=0.144]{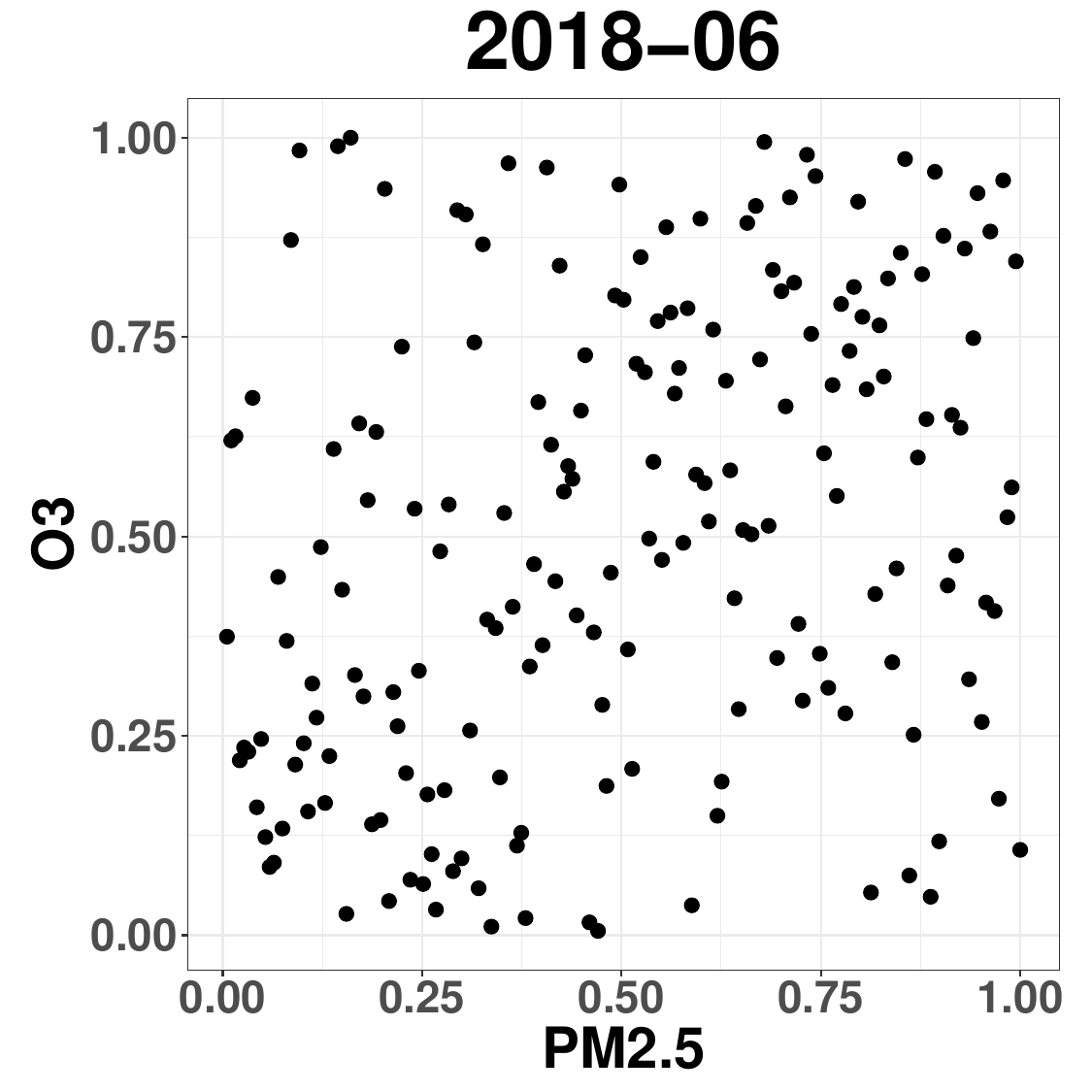}
\includegraphics[scale=0.144]{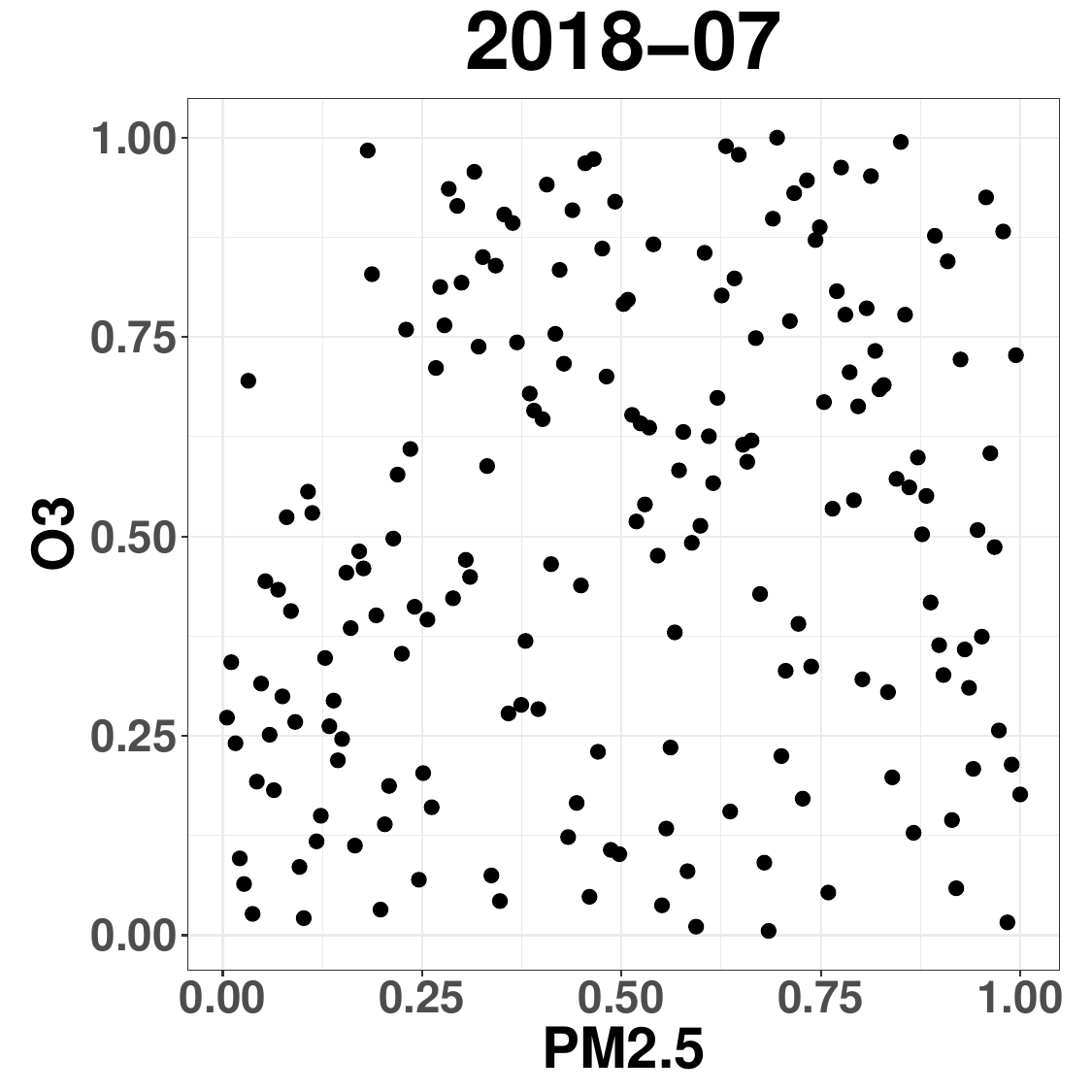}
\includegraphics[scale=0.144]{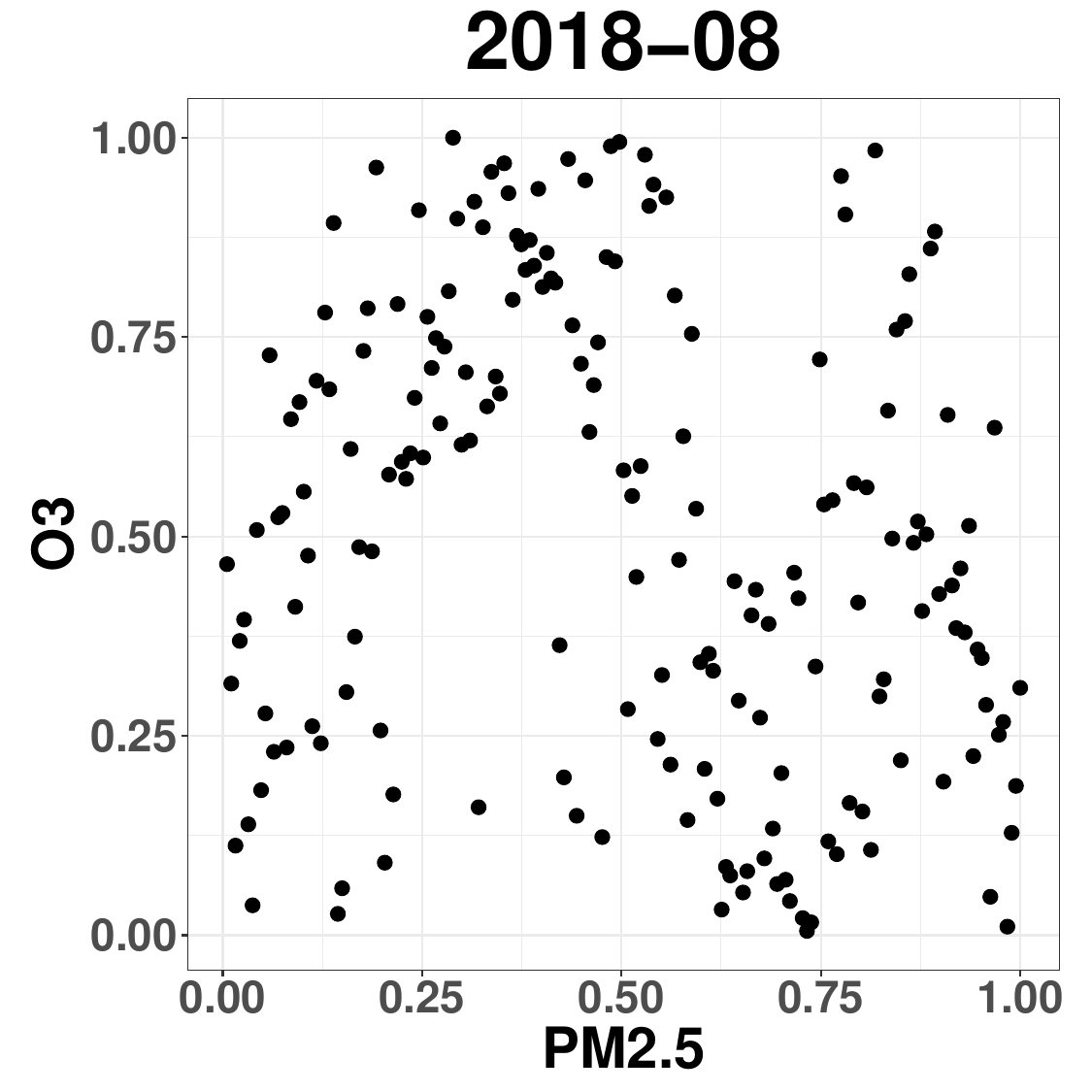}
\includegraphics[scale=0.144]{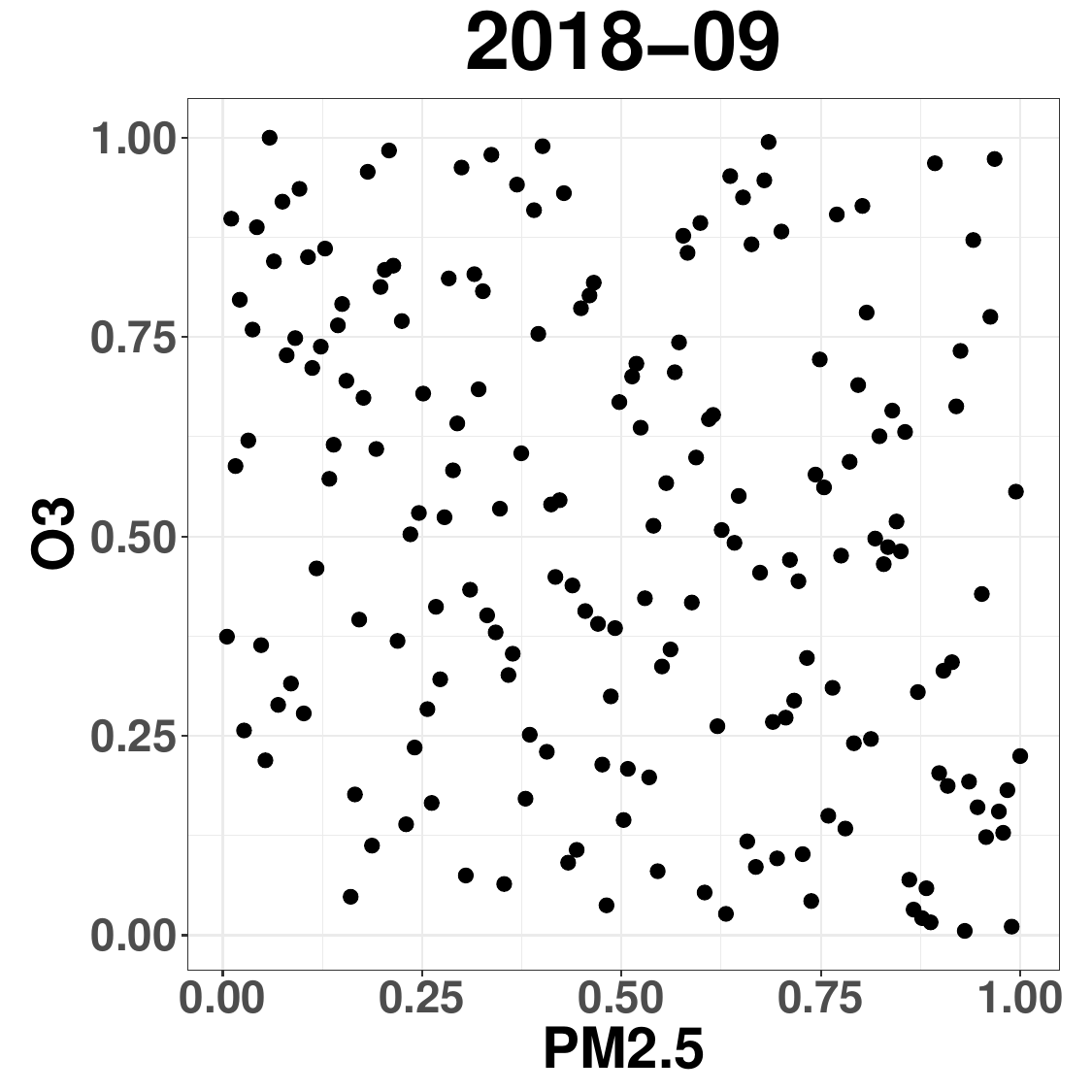}
}
\centerline{
\includegraphics[scale=0.144]{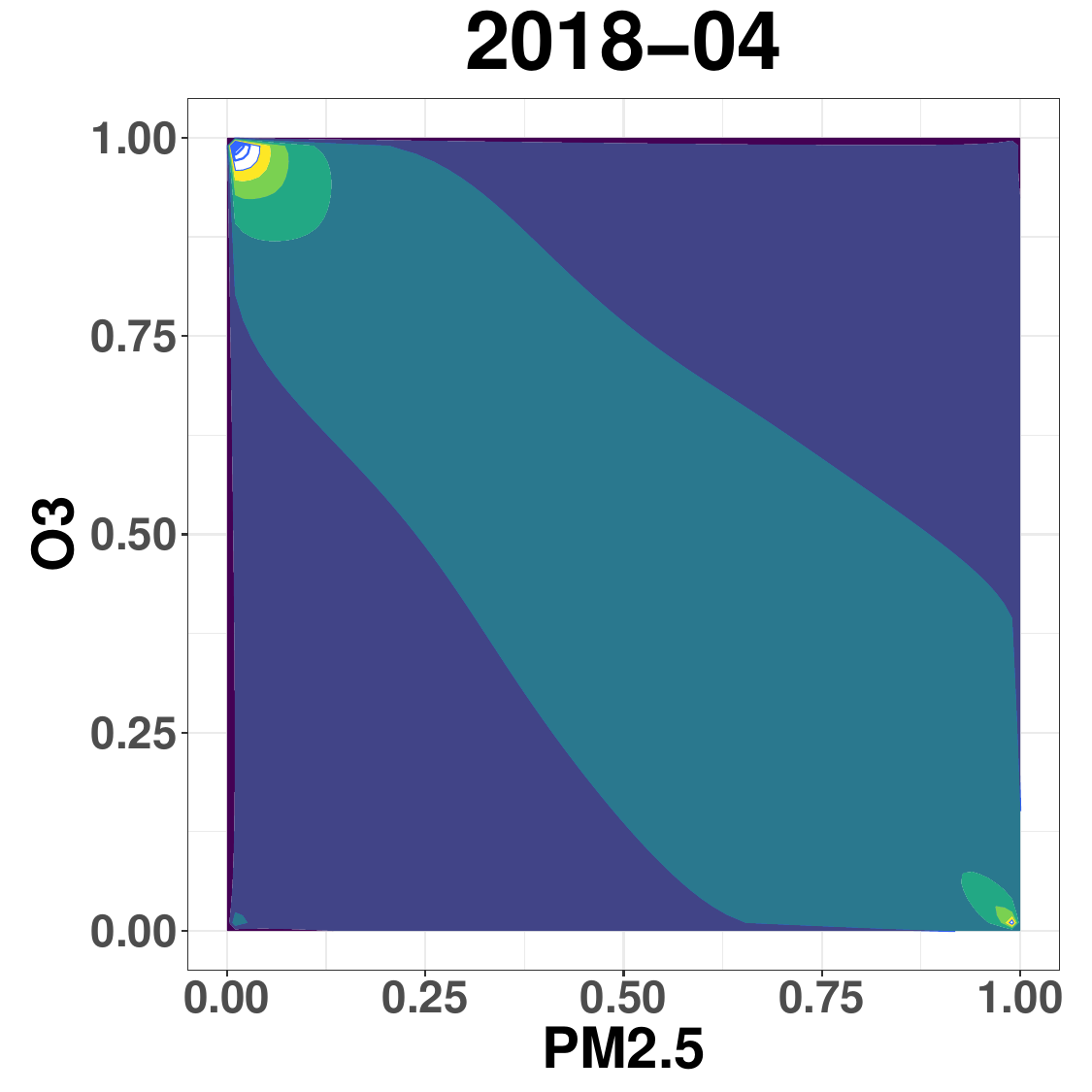}
\includegraphics[scale=0.144]{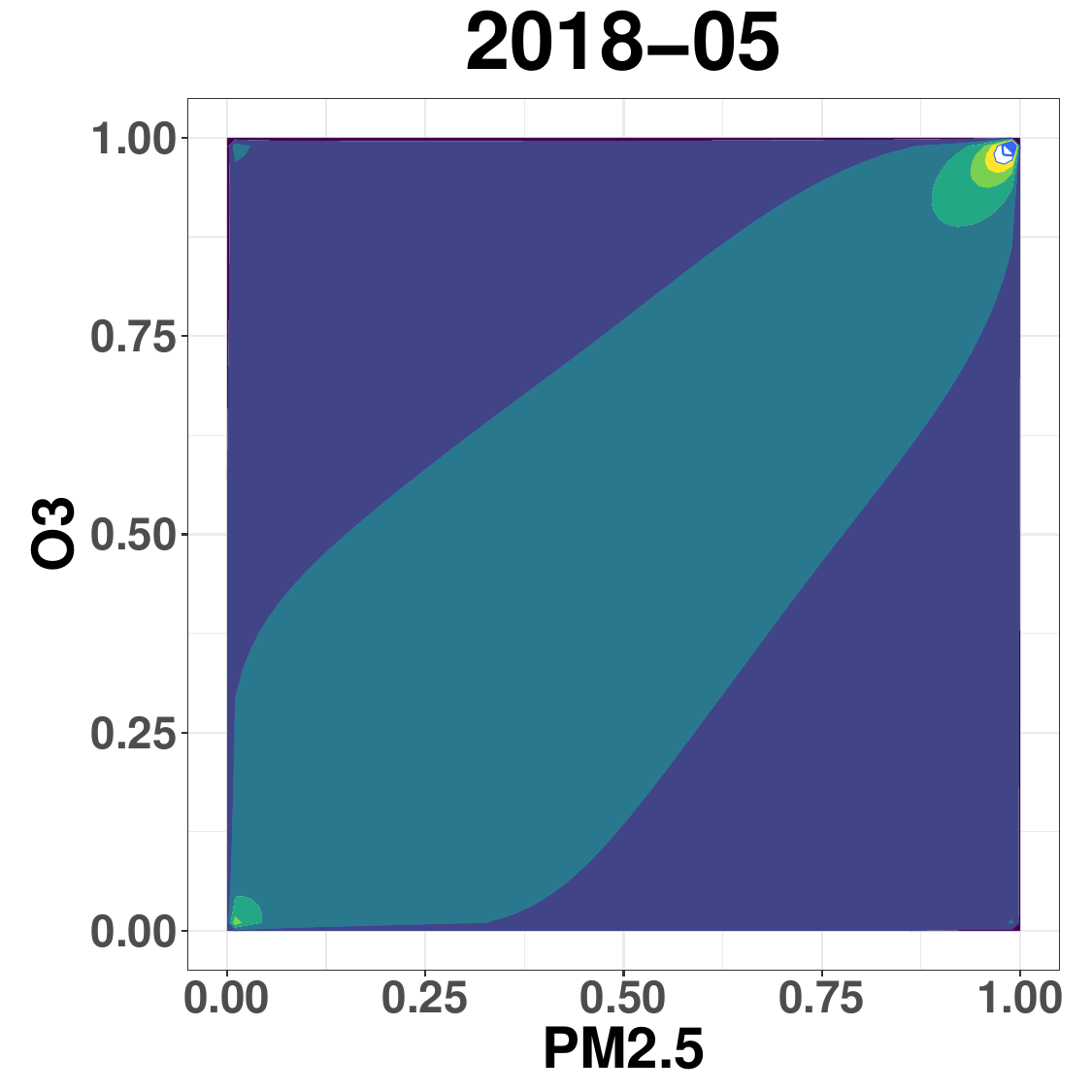}
\includegraphics[scale=0.144]{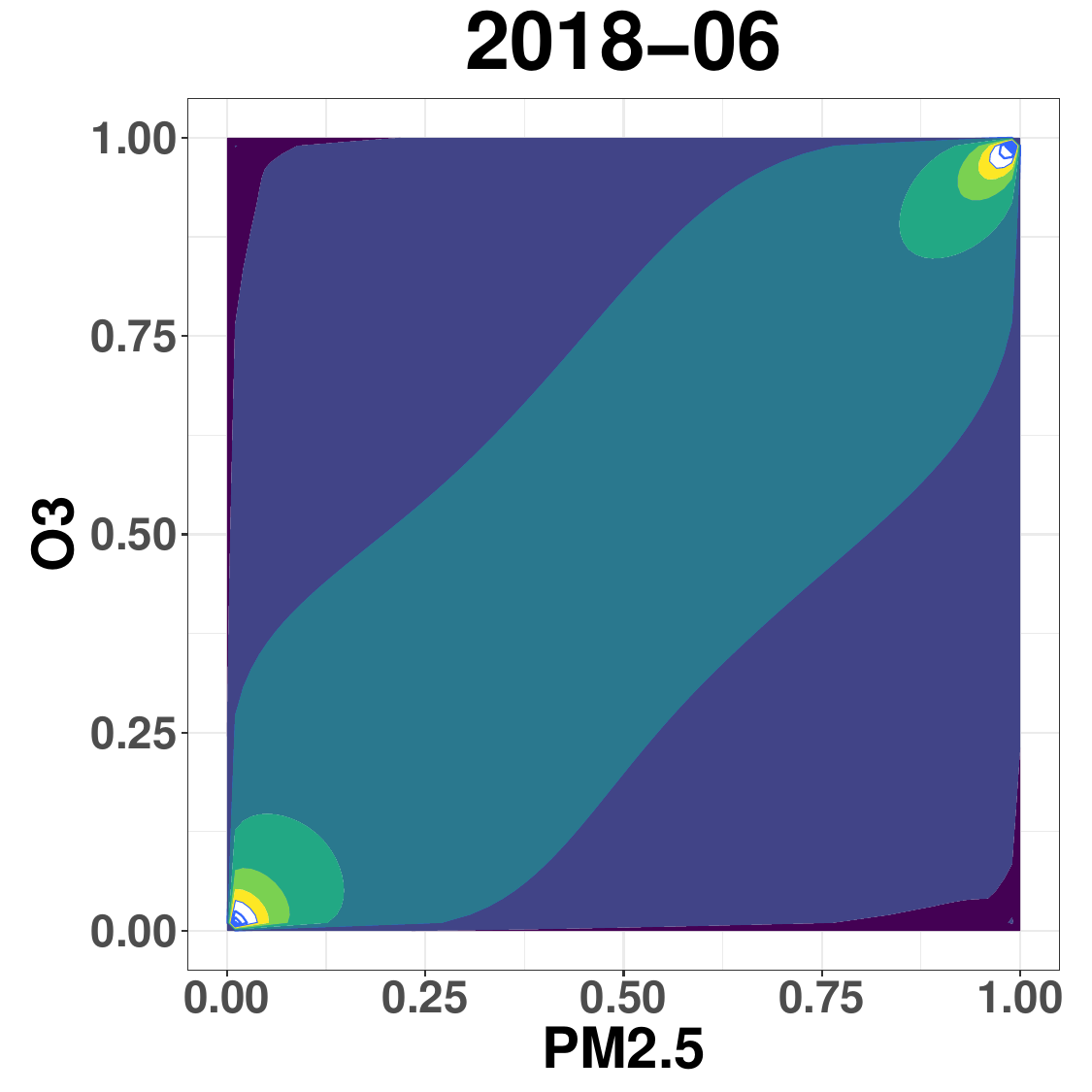}
\includegraphics[scale=0.144]{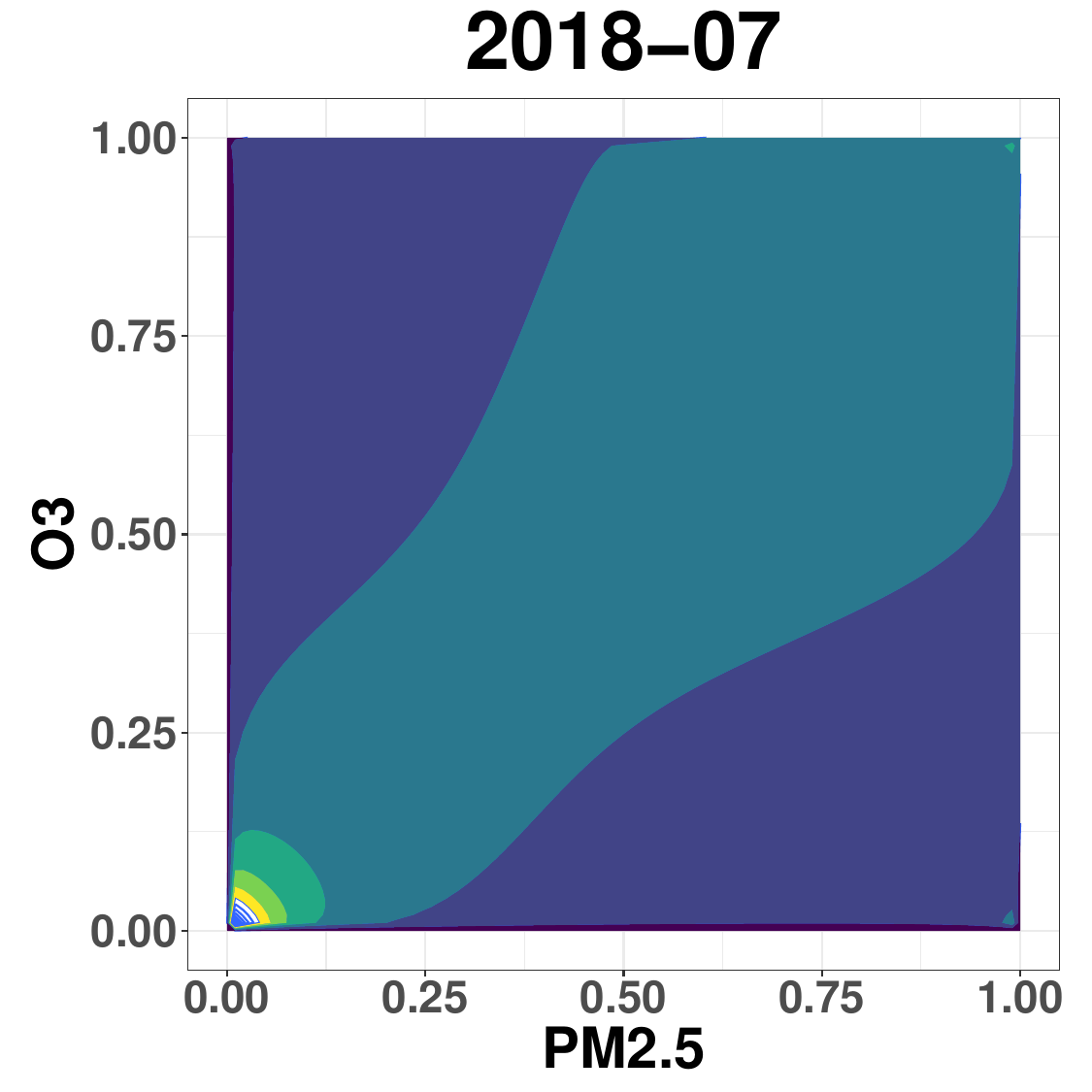}
\includegraphics[scale=0.144]{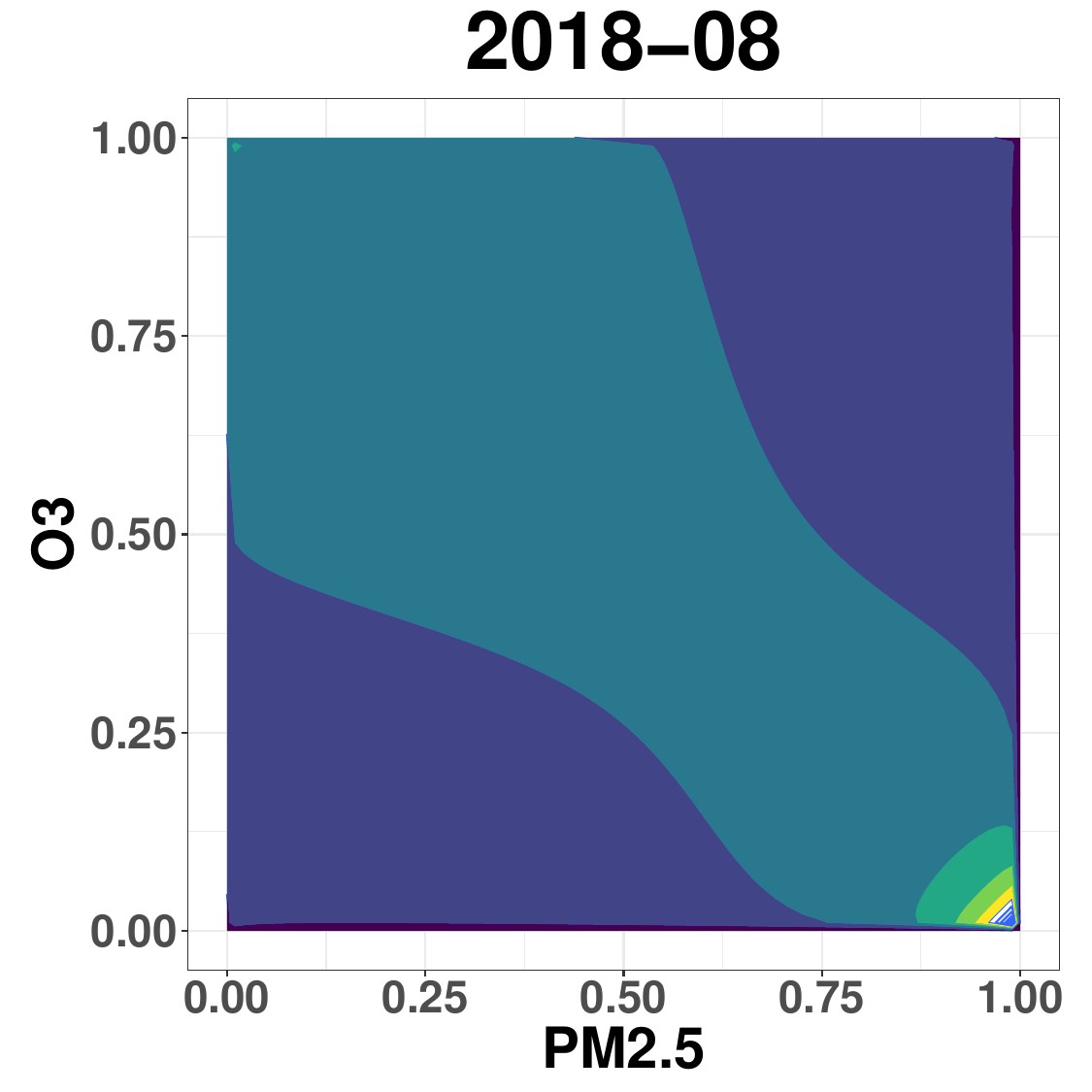}
\includegraphics[scale=0.144]{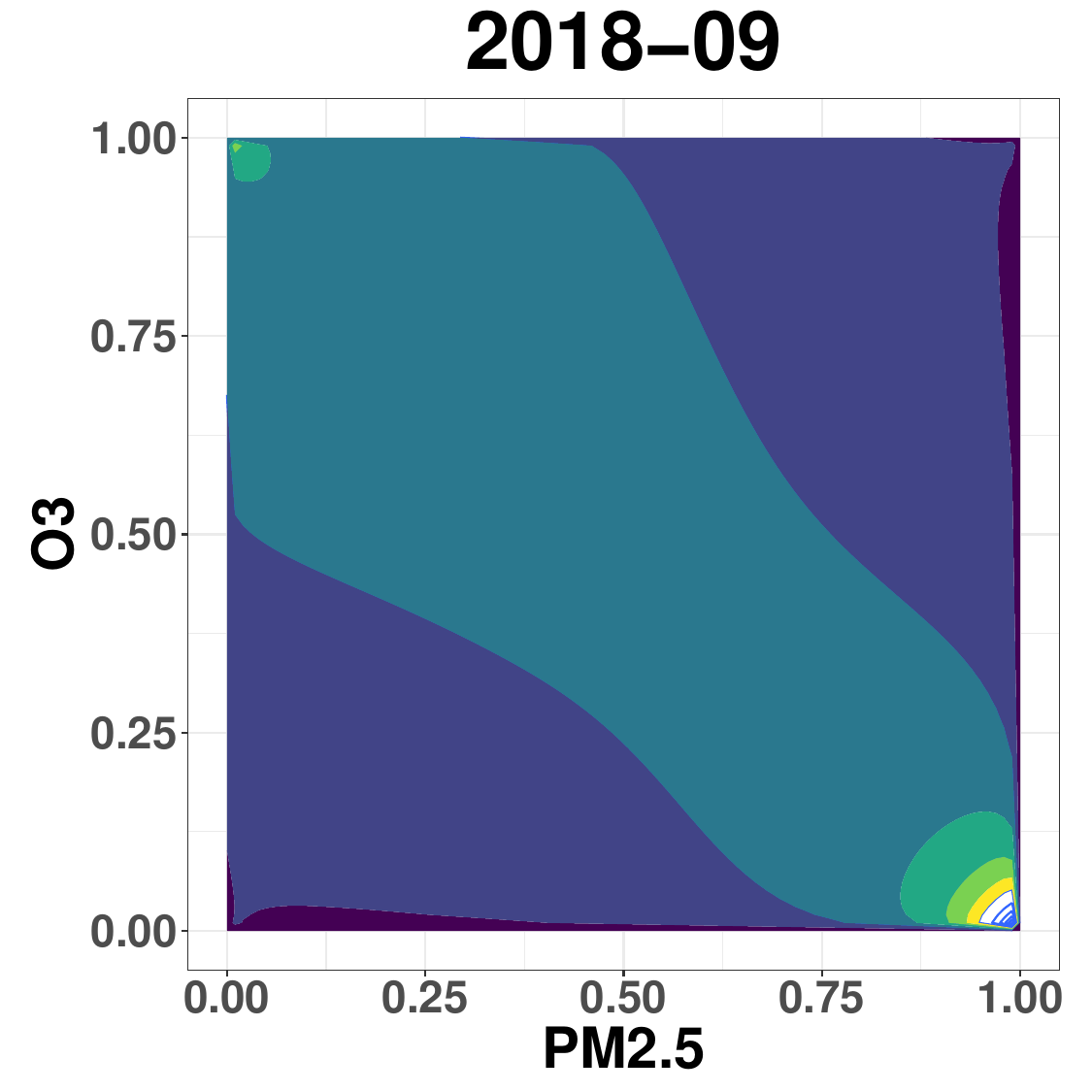}
}
\centerline{
\includegraphics[scale=0.144]{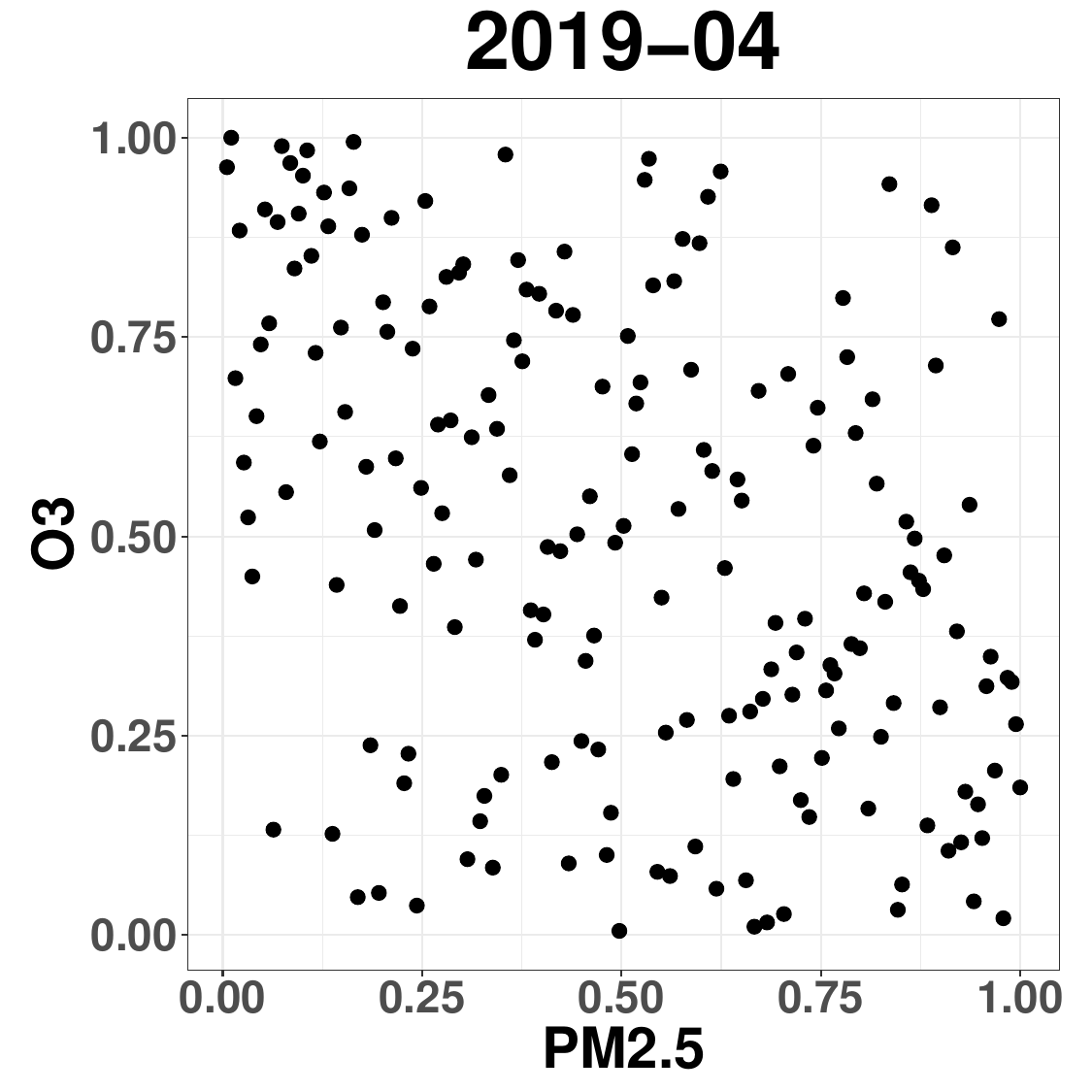}
\includegraphics[scale=0.144]{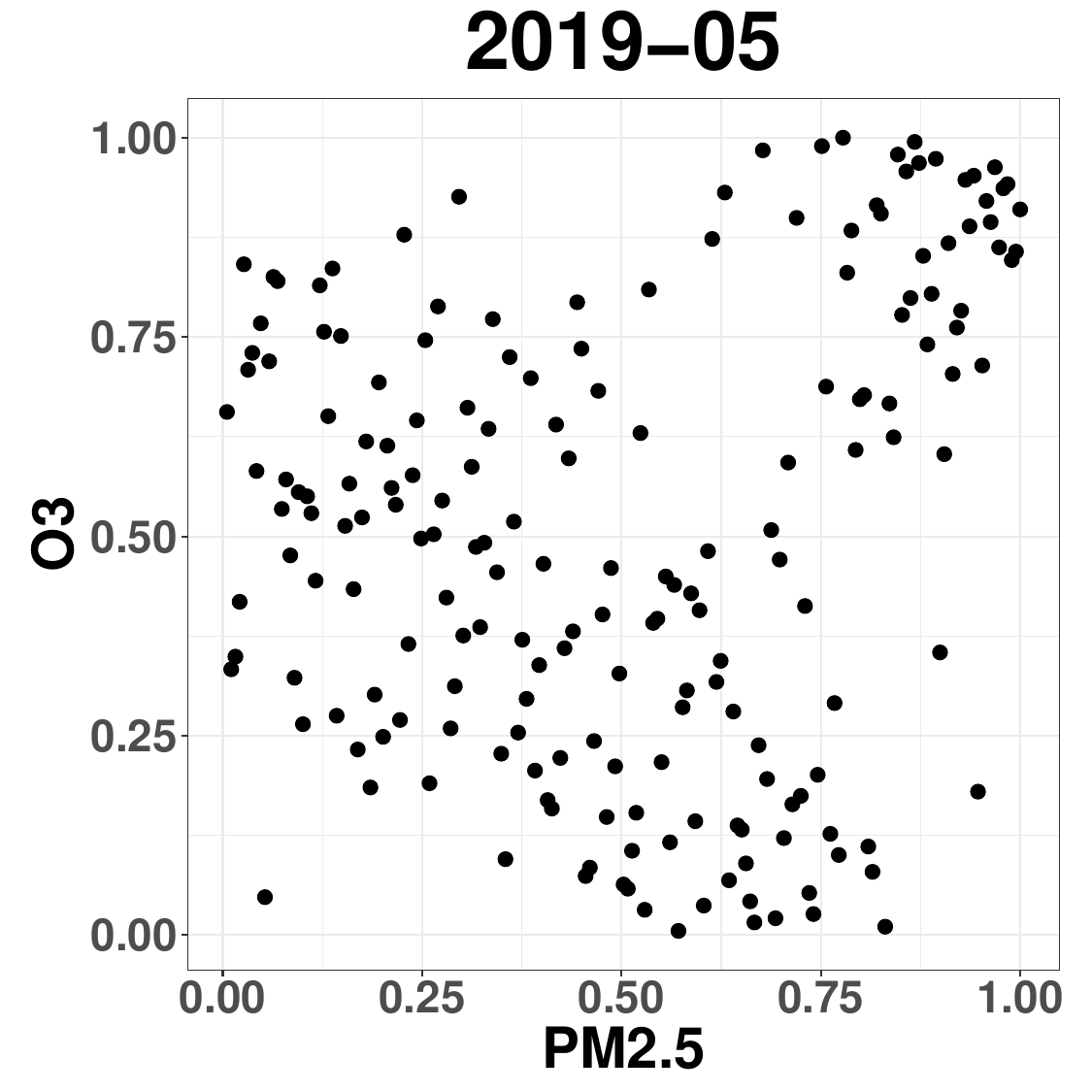}
\includegraphics[scale=0.144]{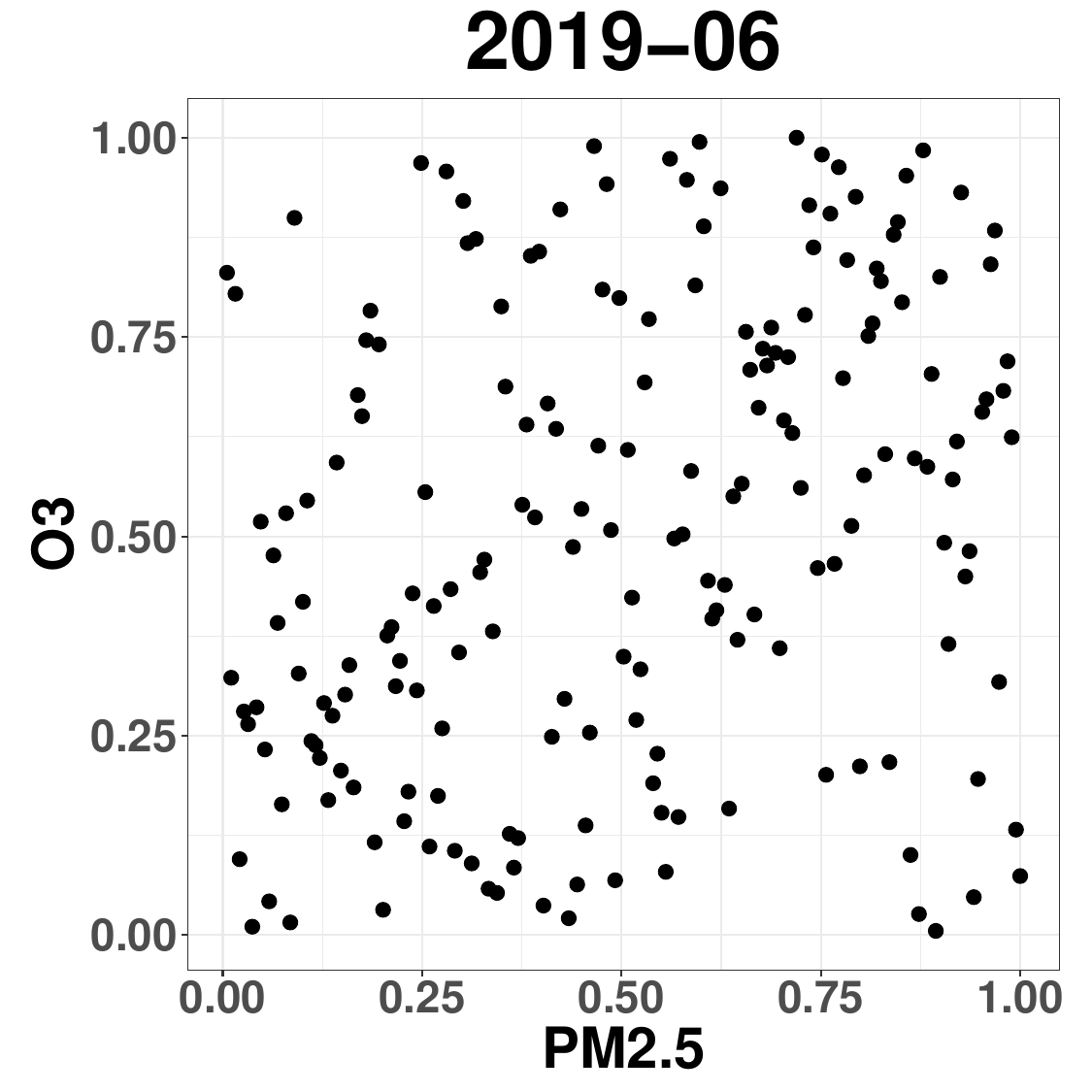}
\includegraphics[scale=0.144]{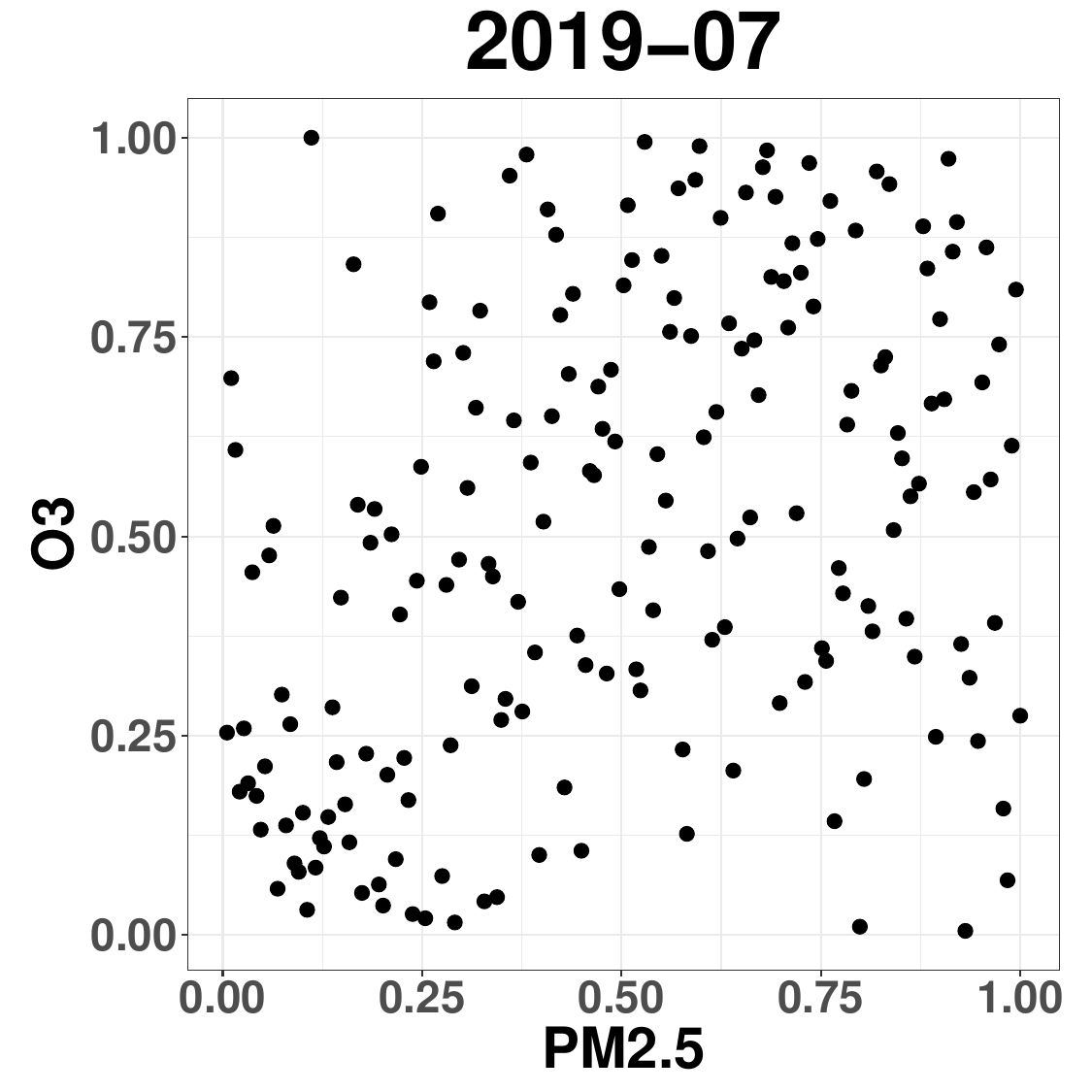}
\includegraphics[scale=0.144]{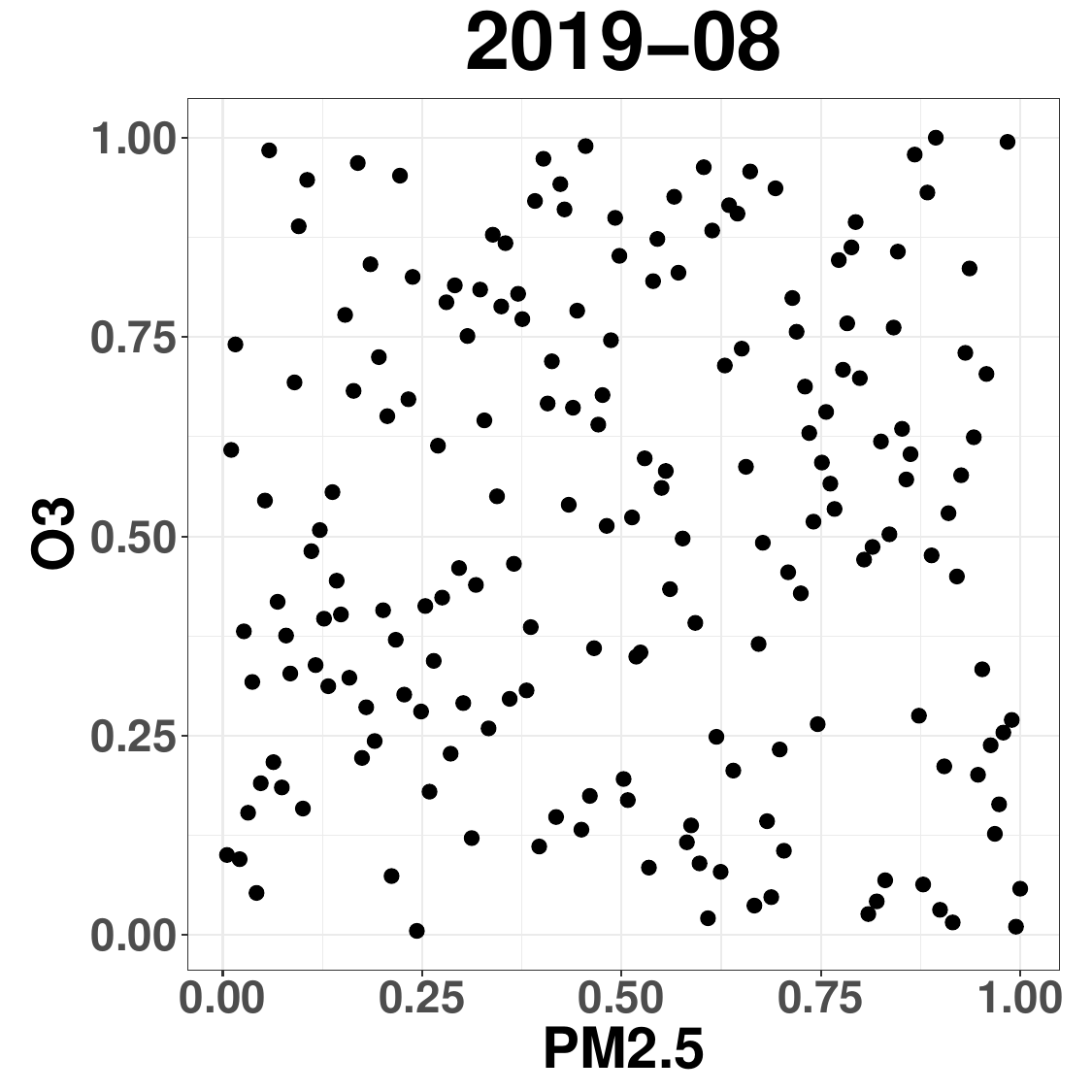}
\includegraphics[scale=0.144]{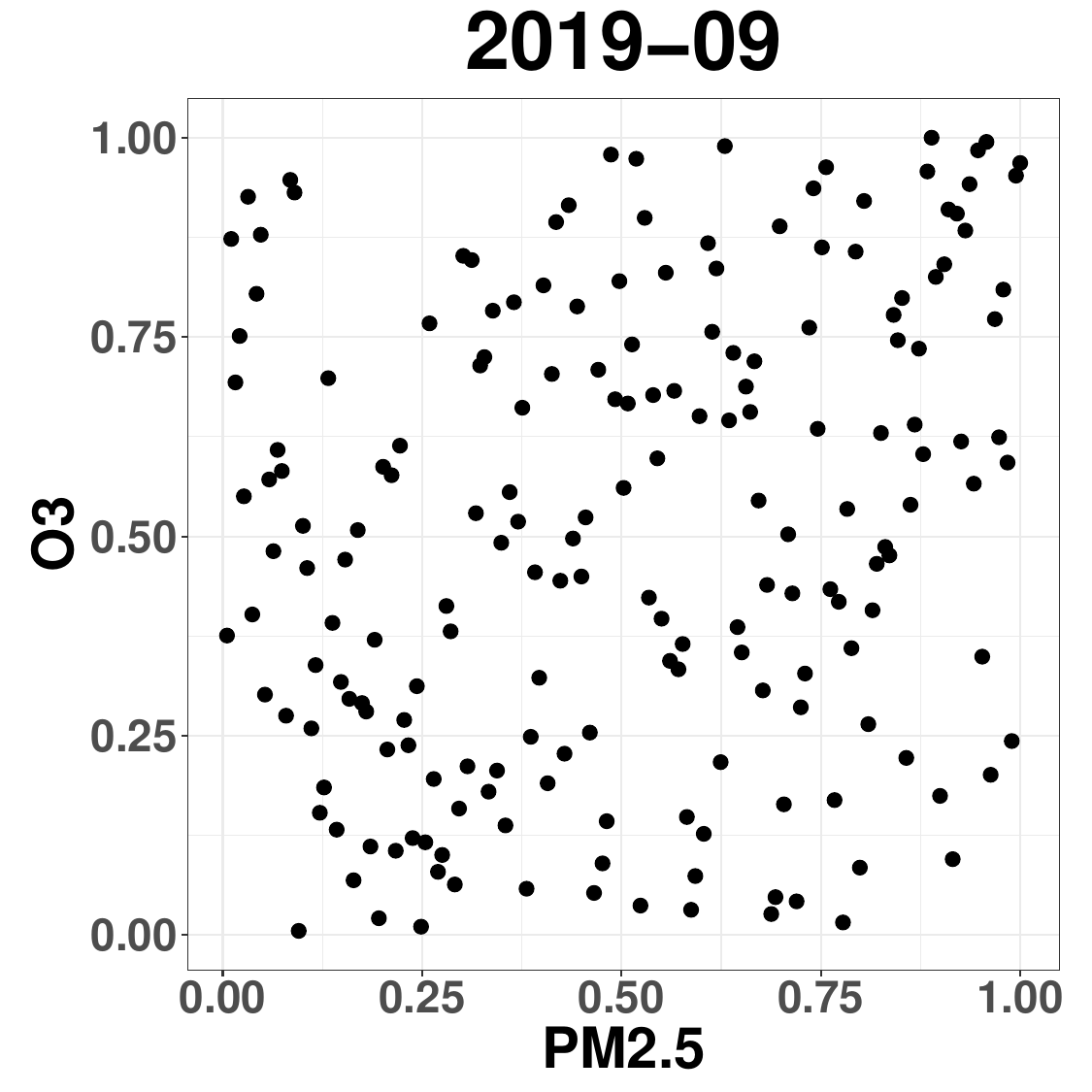}
}
\centerline{
\includegraphics[scale=0.144]{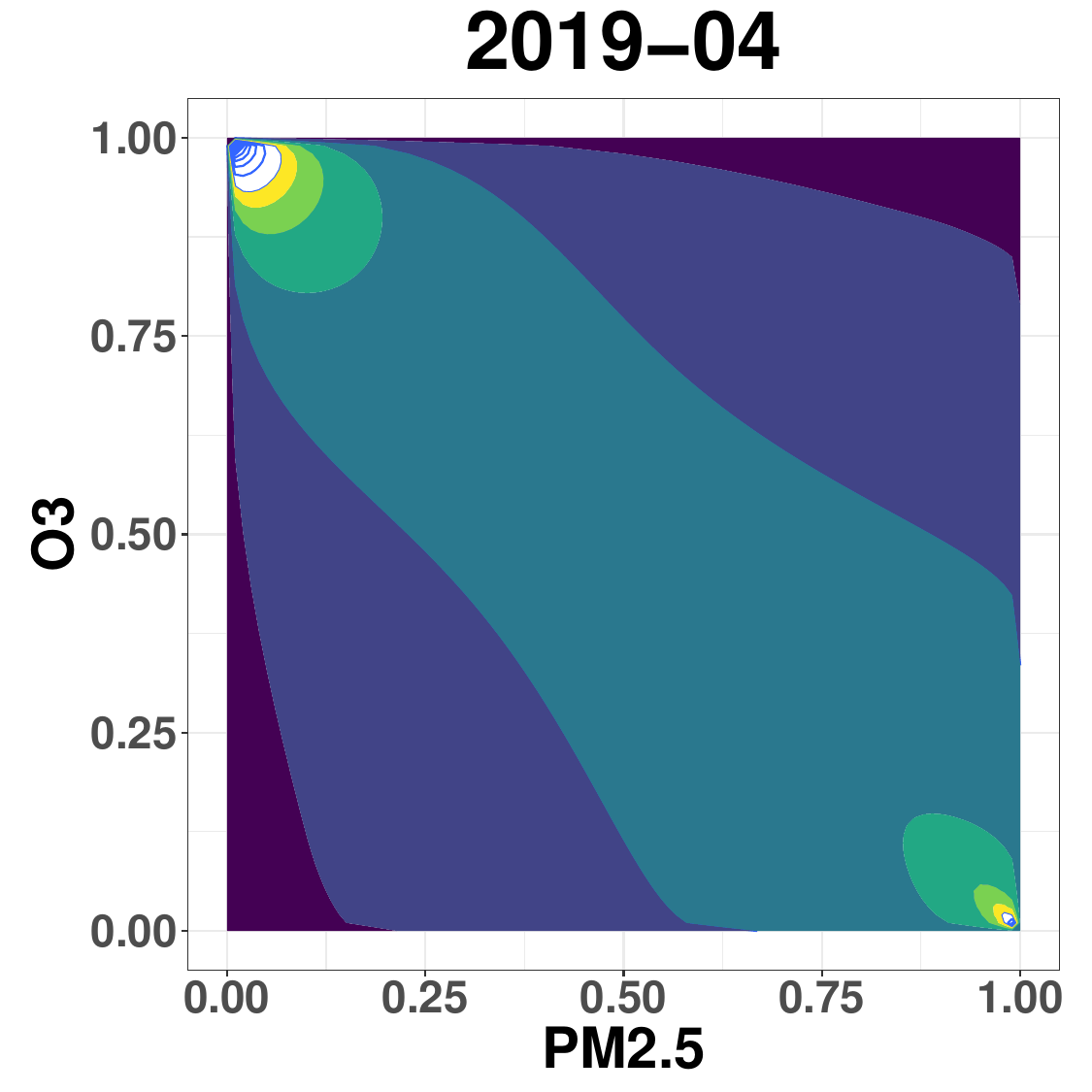}
\includegraphics[scale=0.144]{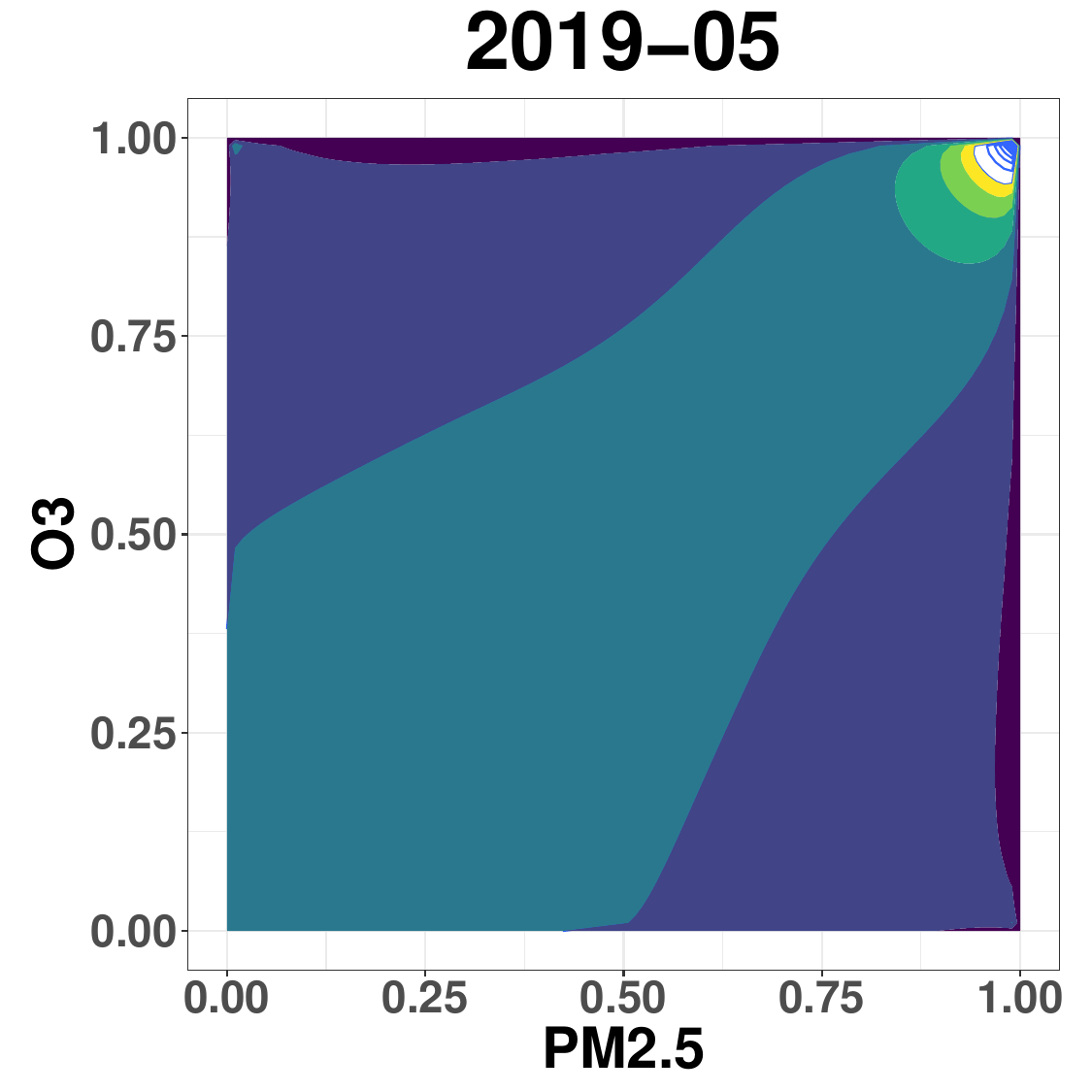}
\includegraphics[scale=0.144]{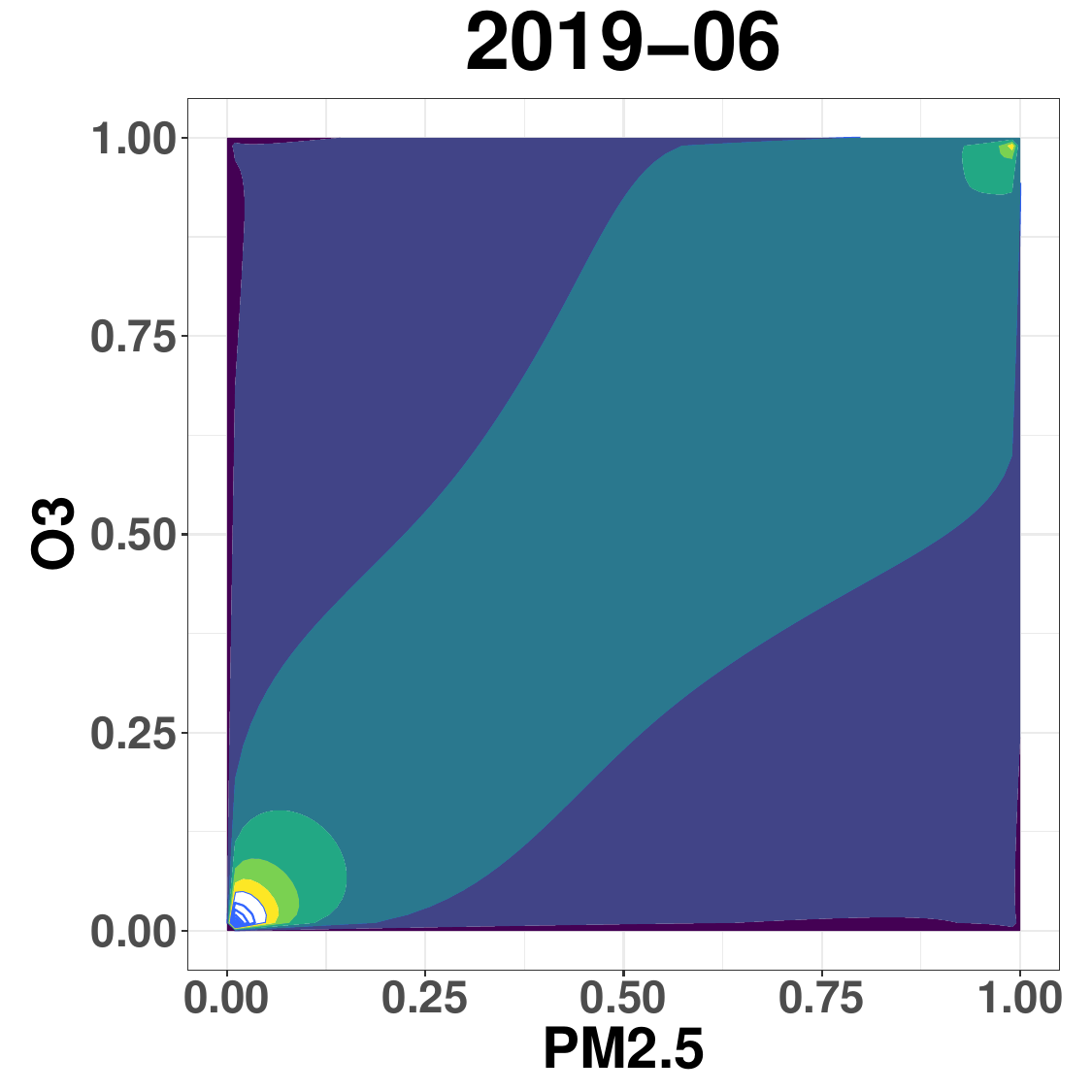}
\includegraphics[scale=0.144]{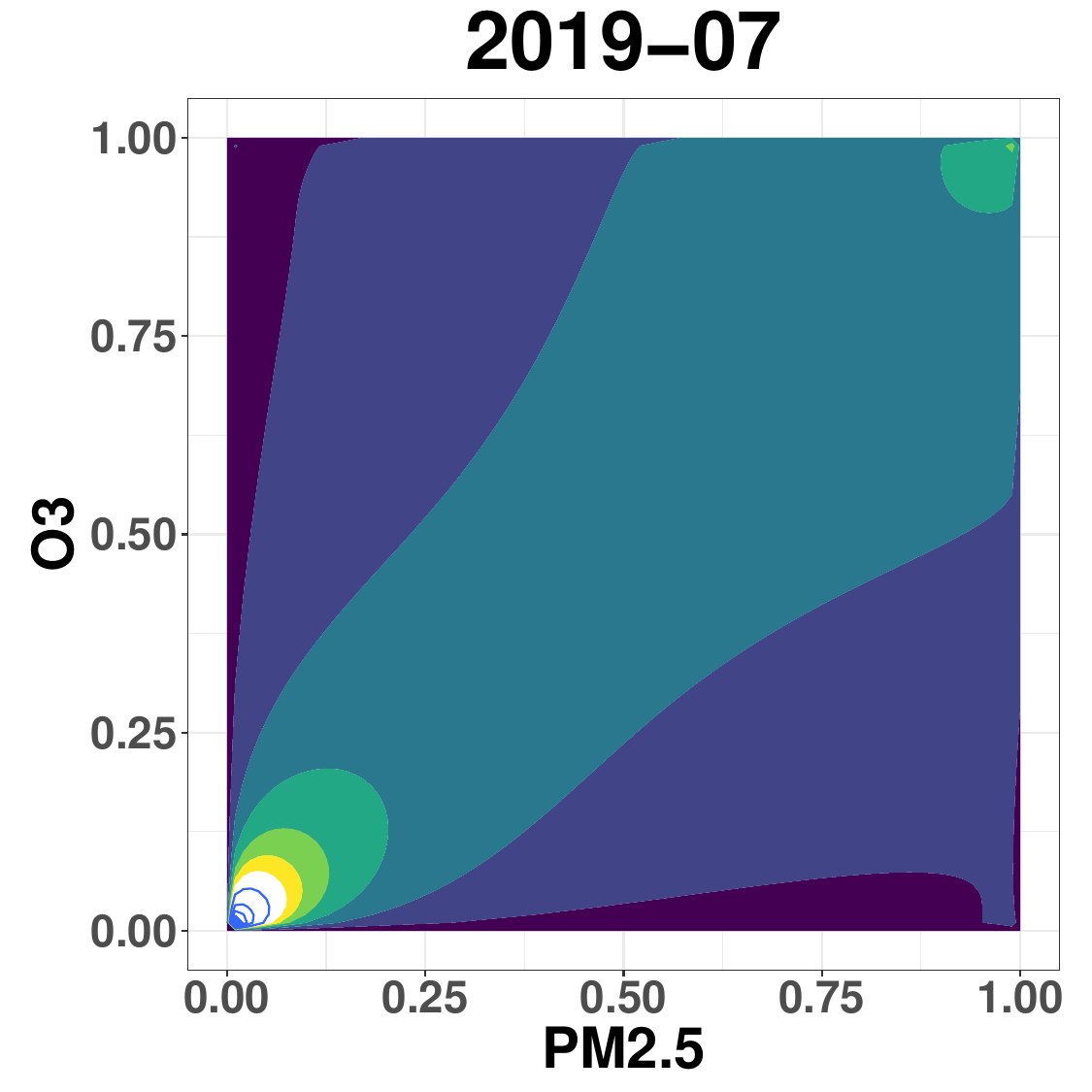}
\includegraphics[scale=0.144]{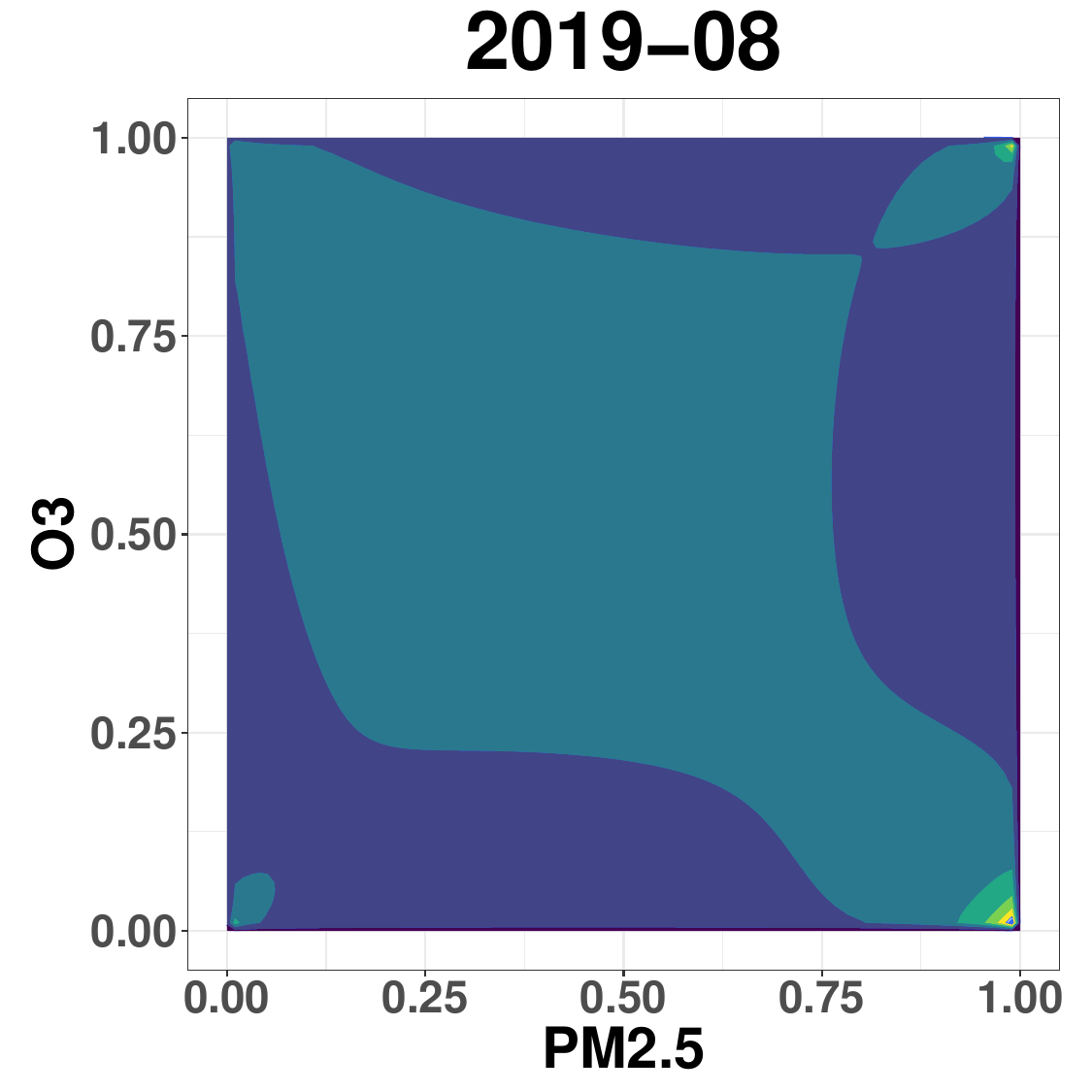}
\includegraphics[scale=0.144]{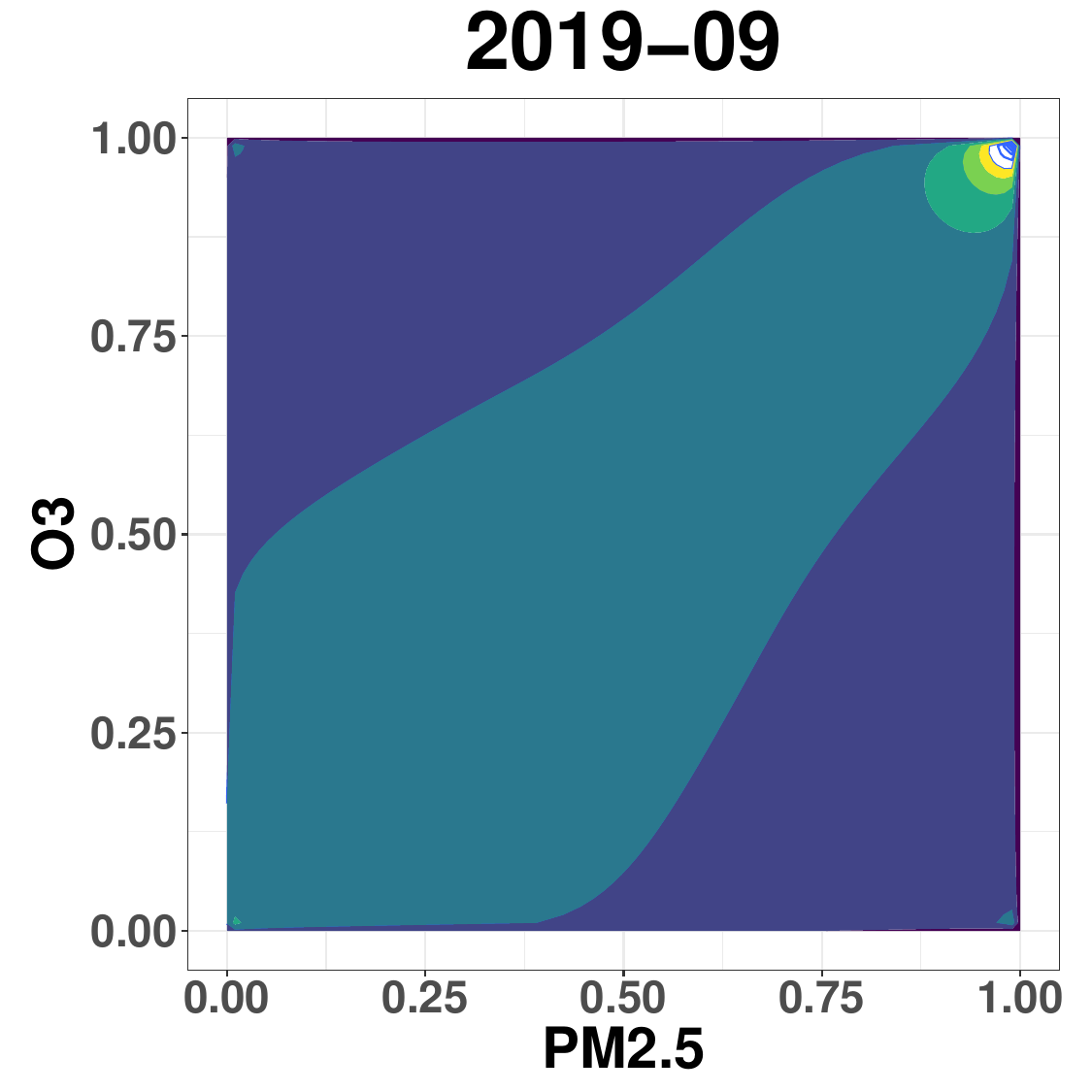}
}
\caption{Observed real data and estimated joint density from April to September 2017 (top), 2018 (middle), and 2019 (bottom). Dependence patterns tend to vary between seasons with some months exhibiting transitional regimes.  }
\label{fig:joint_predict_real}
\end{figure}

\begin{figure}
    \centering
    \includegraphics[scale=0.4]{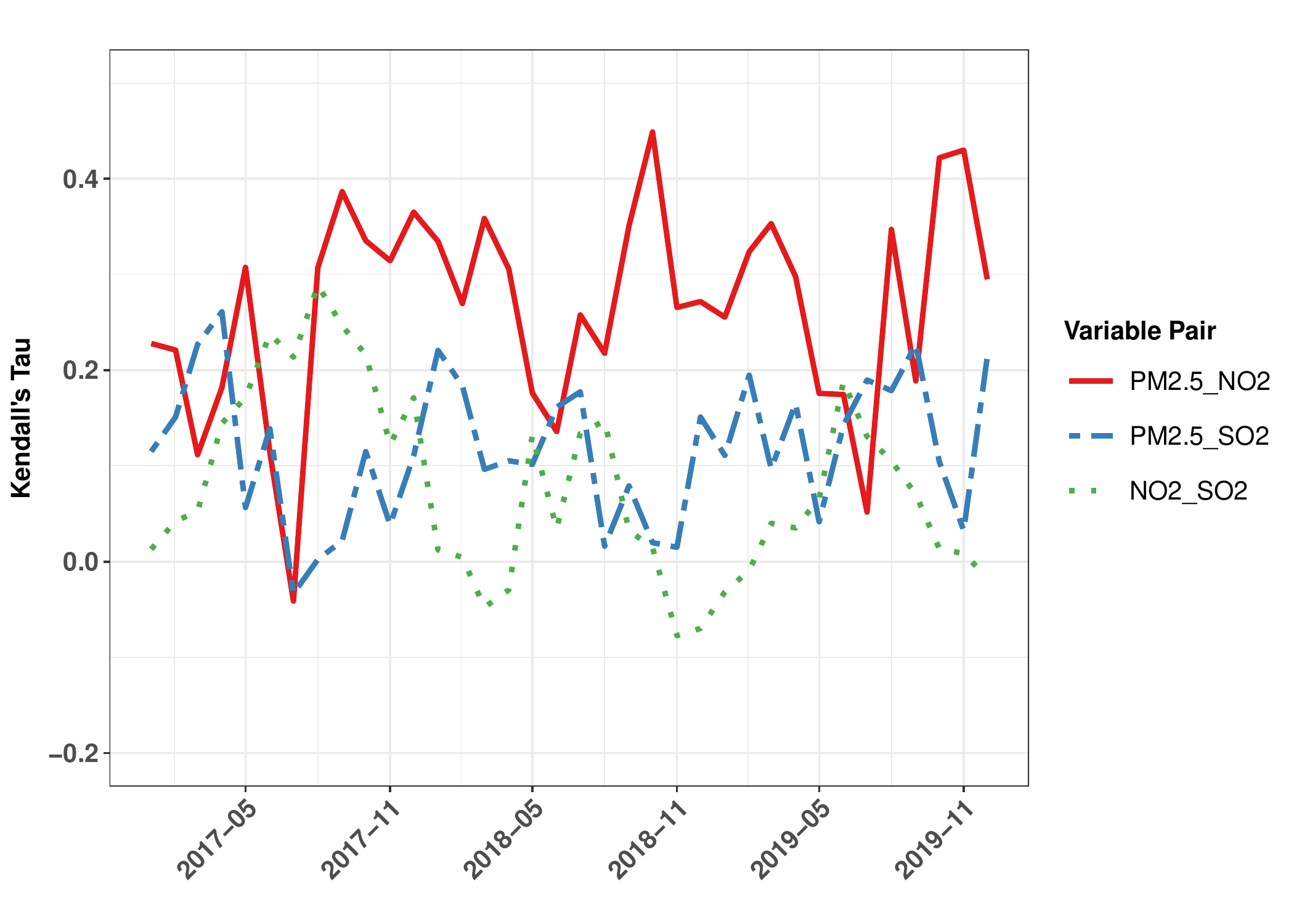}
    \caption{Multivariate real dataset. Pairwise empirical Kendall's tau.}
\label{fig:real_tau}
\end{figure}

\begin{figure}
\centering
\includegraphics[scale=0.4]{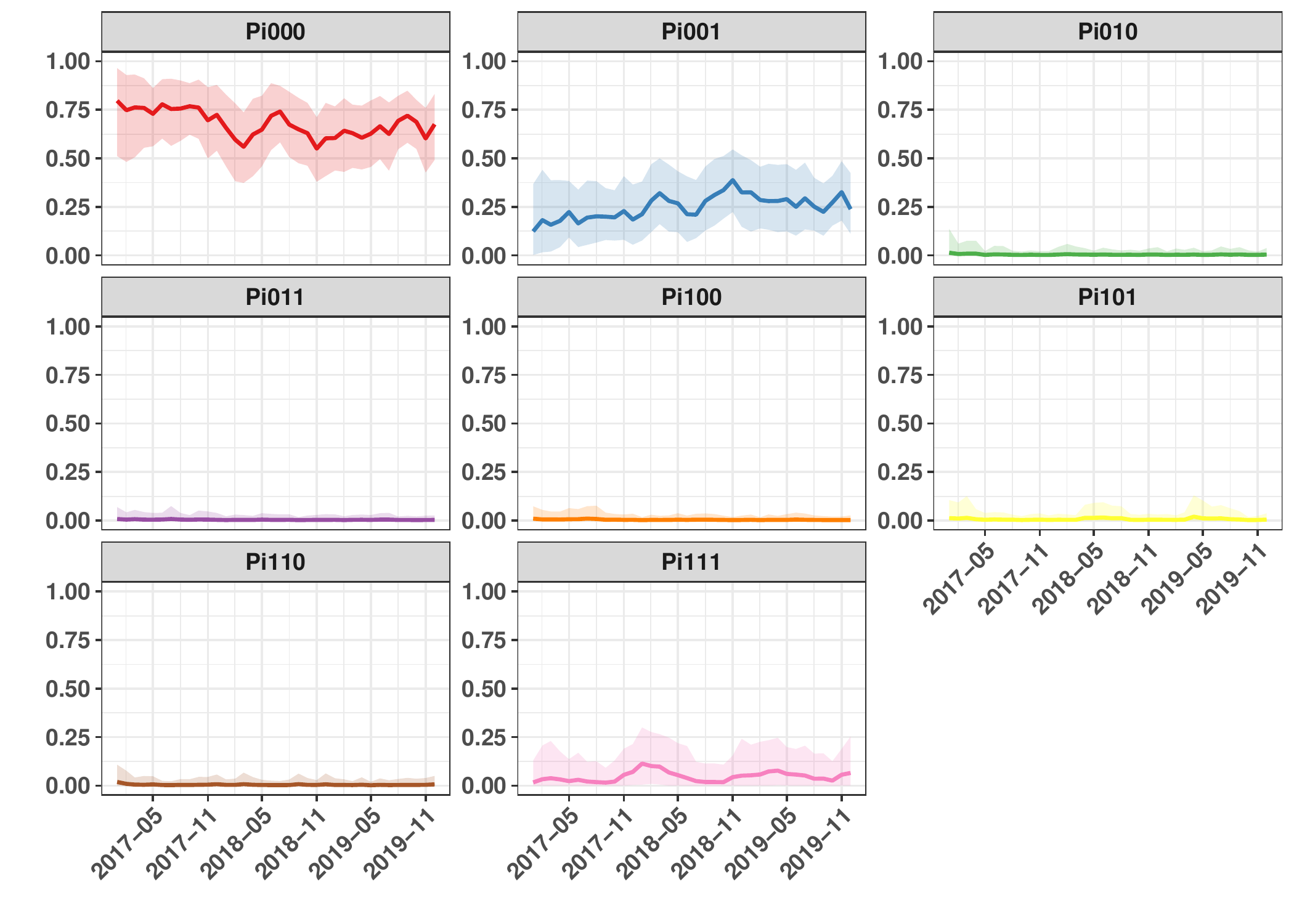}
\includegraphics[scale=0.4]{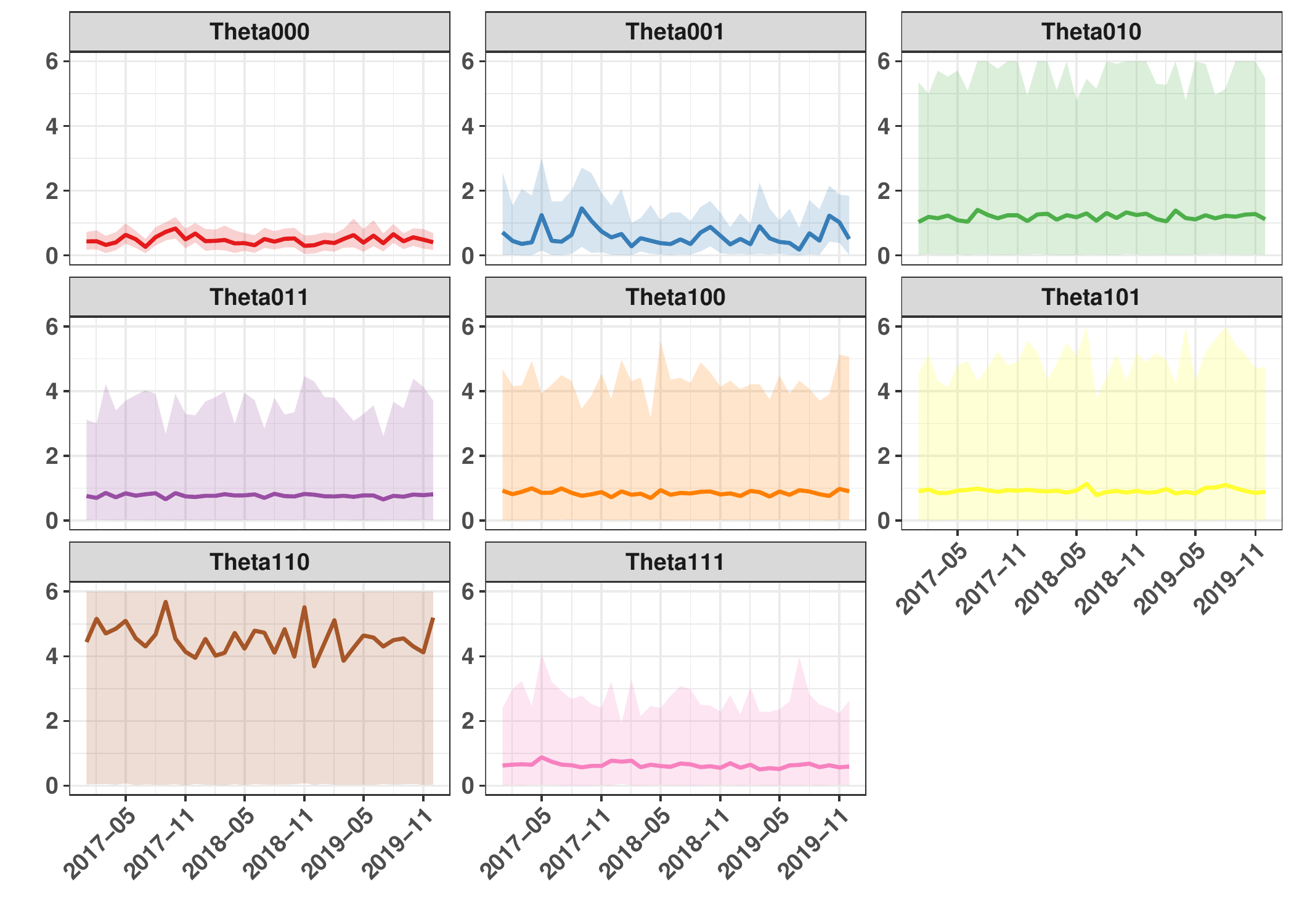}
\caption{Multivariate real dataset. Posterior estimates of $\bpi$ (top) and $\btheta$ (down): posterior mean (solid line) with 95\% credible intervals (shadows).}
\label{fig:pi_dim3_predict_real}
\end{figure}

\begin{figure}
\centerline{
\includegraphics[scale=0.144]{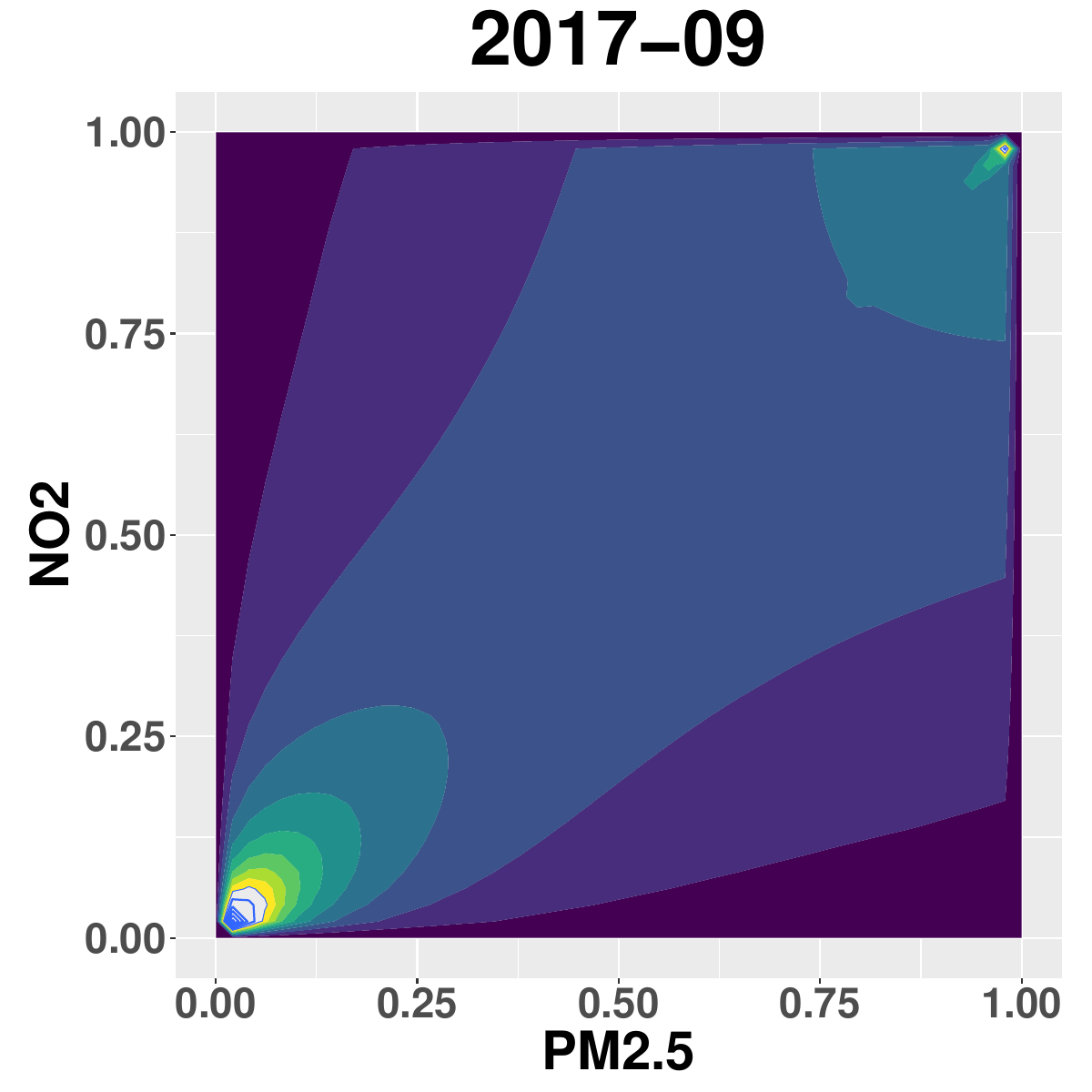}
\includegraphics[scale=0.144]{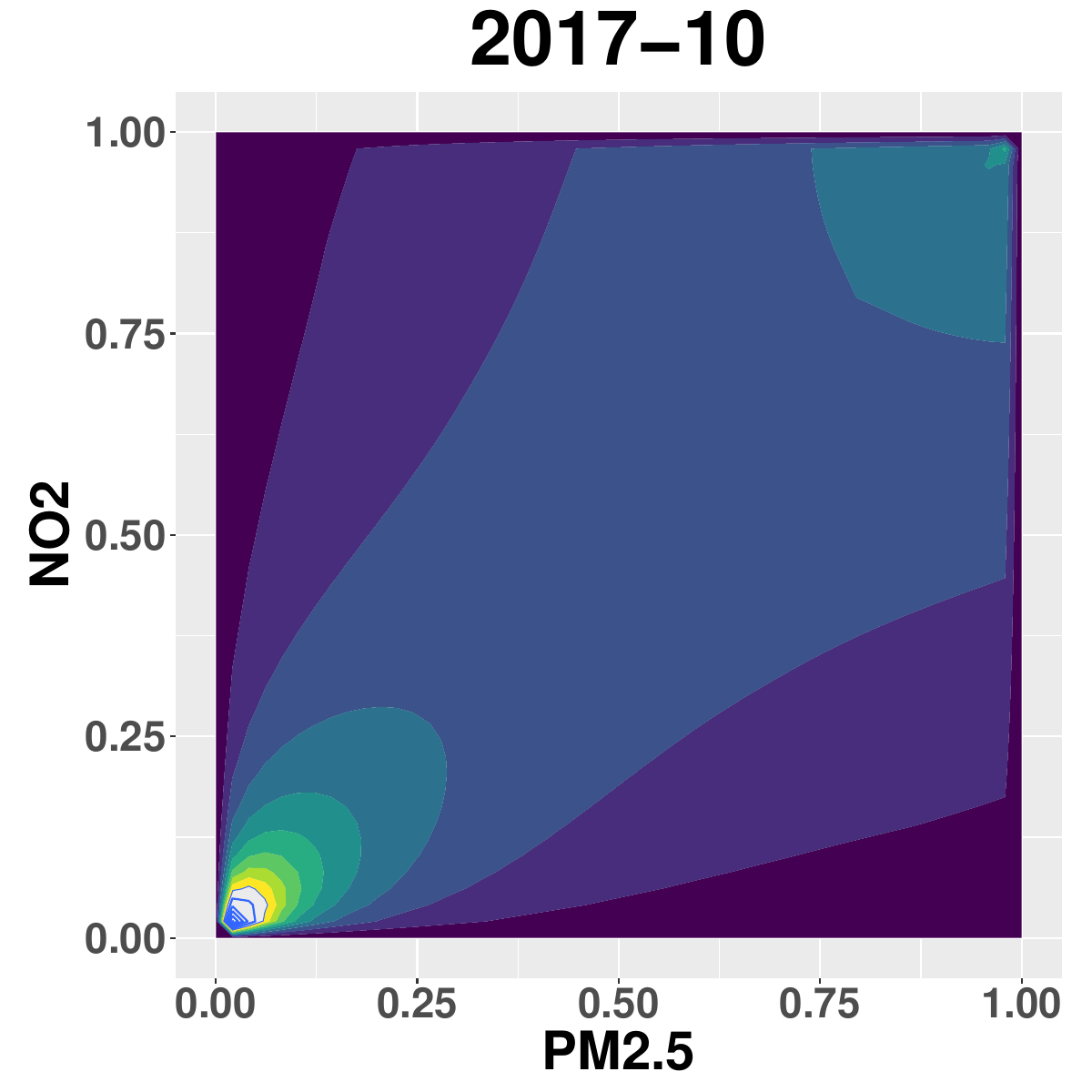}
\includegraphics[scale=0.144]{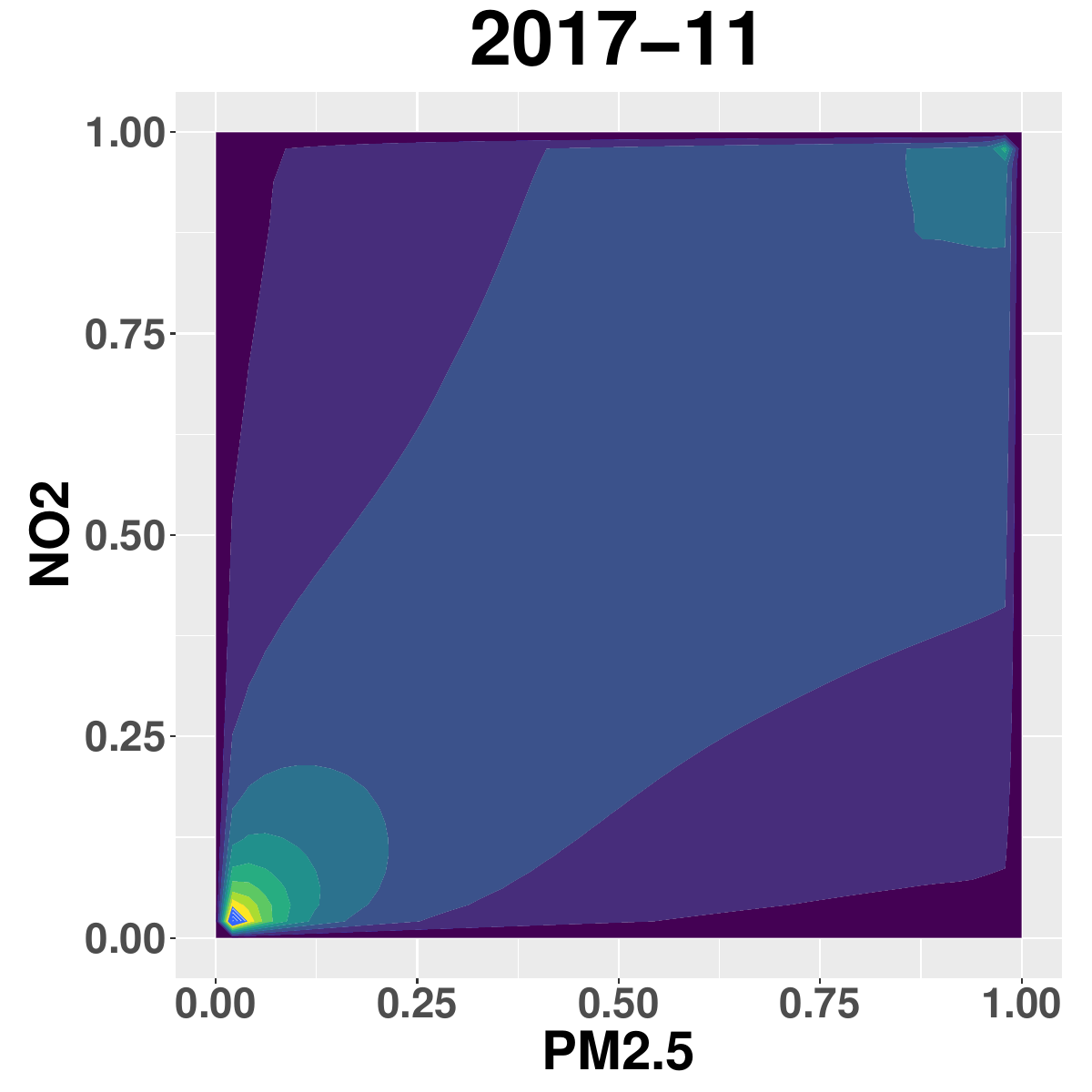}
\includegraphics[scale=0.144]{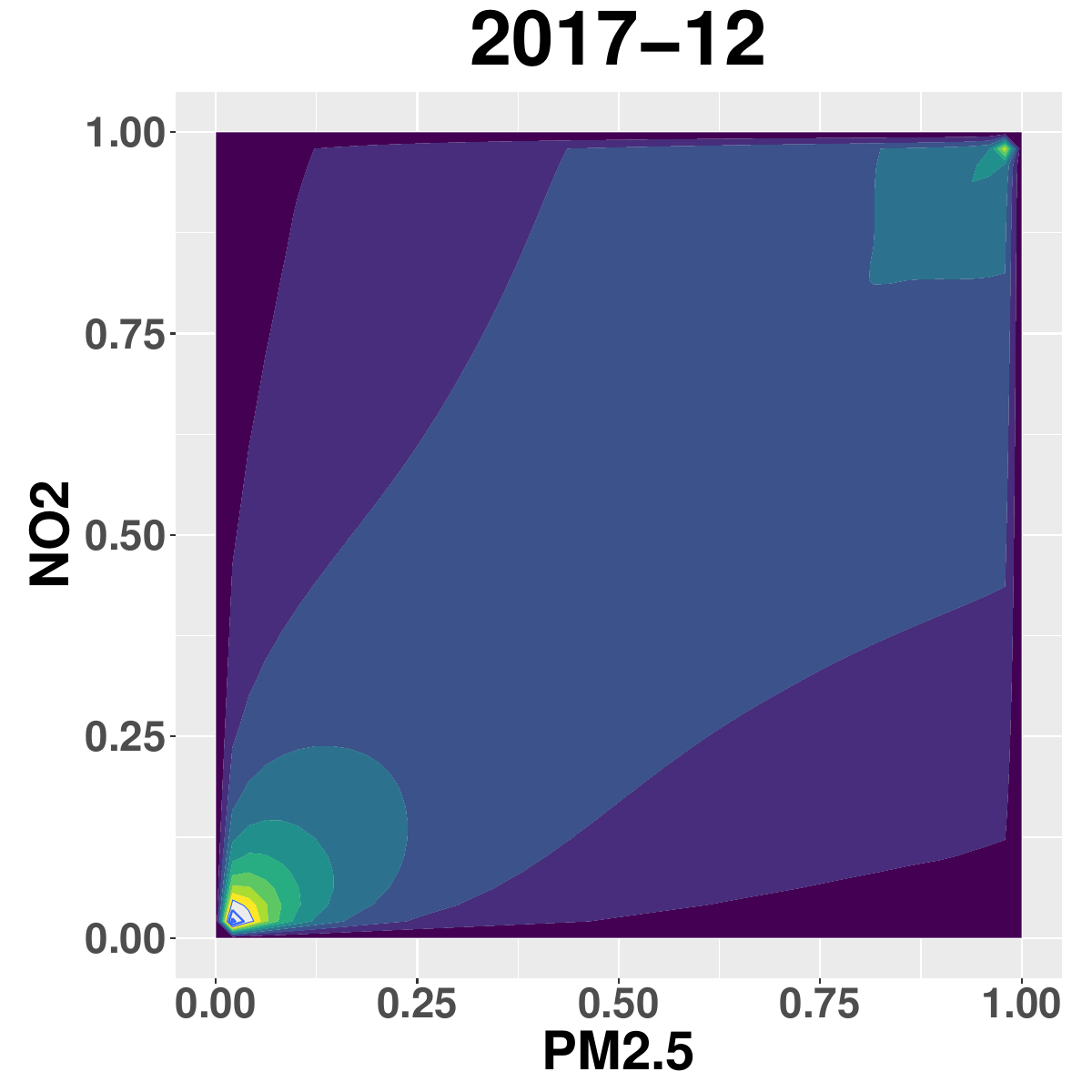}
}
\centerline{
\includegraphics[scale=0.144]{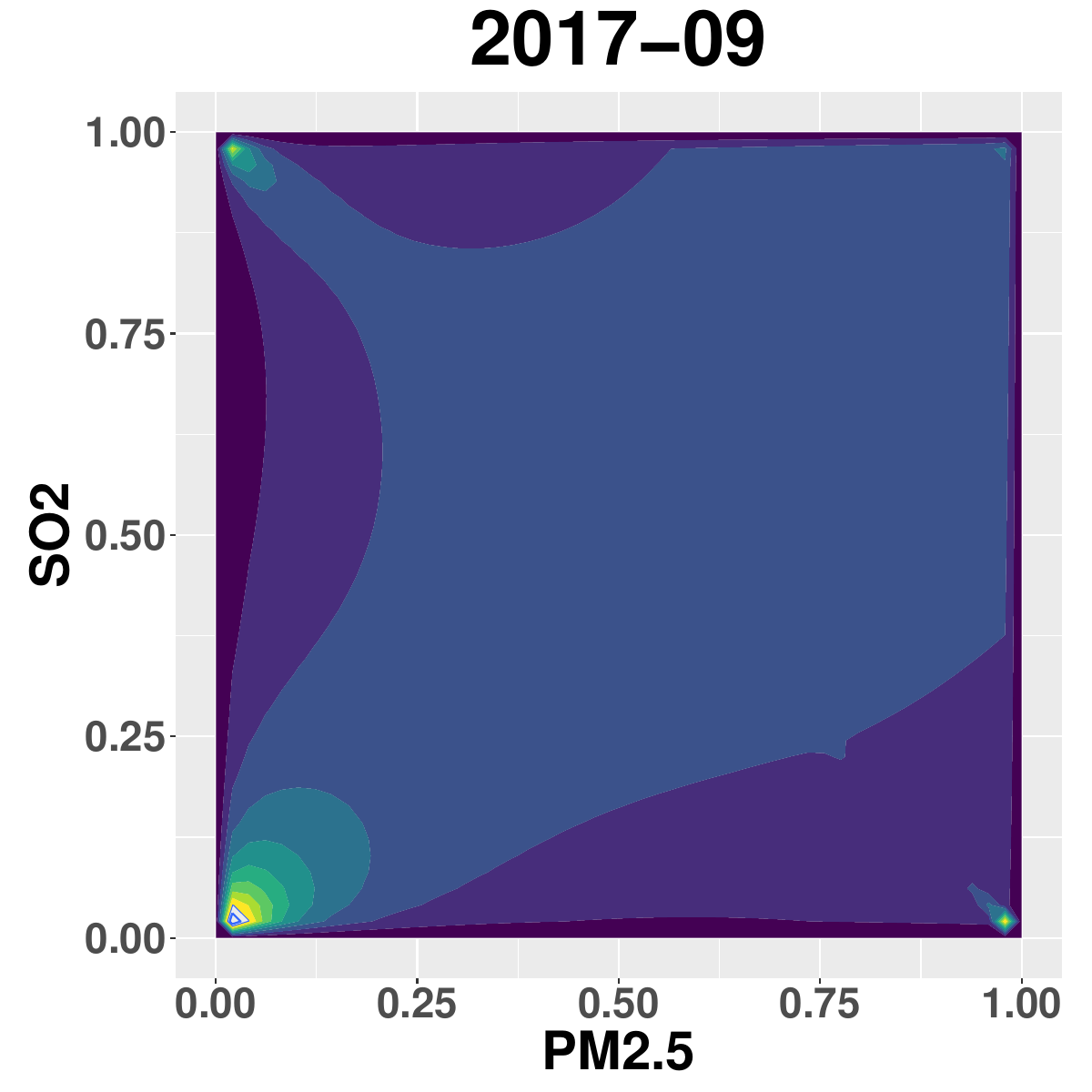}
\includegraphics[scale=0.144]{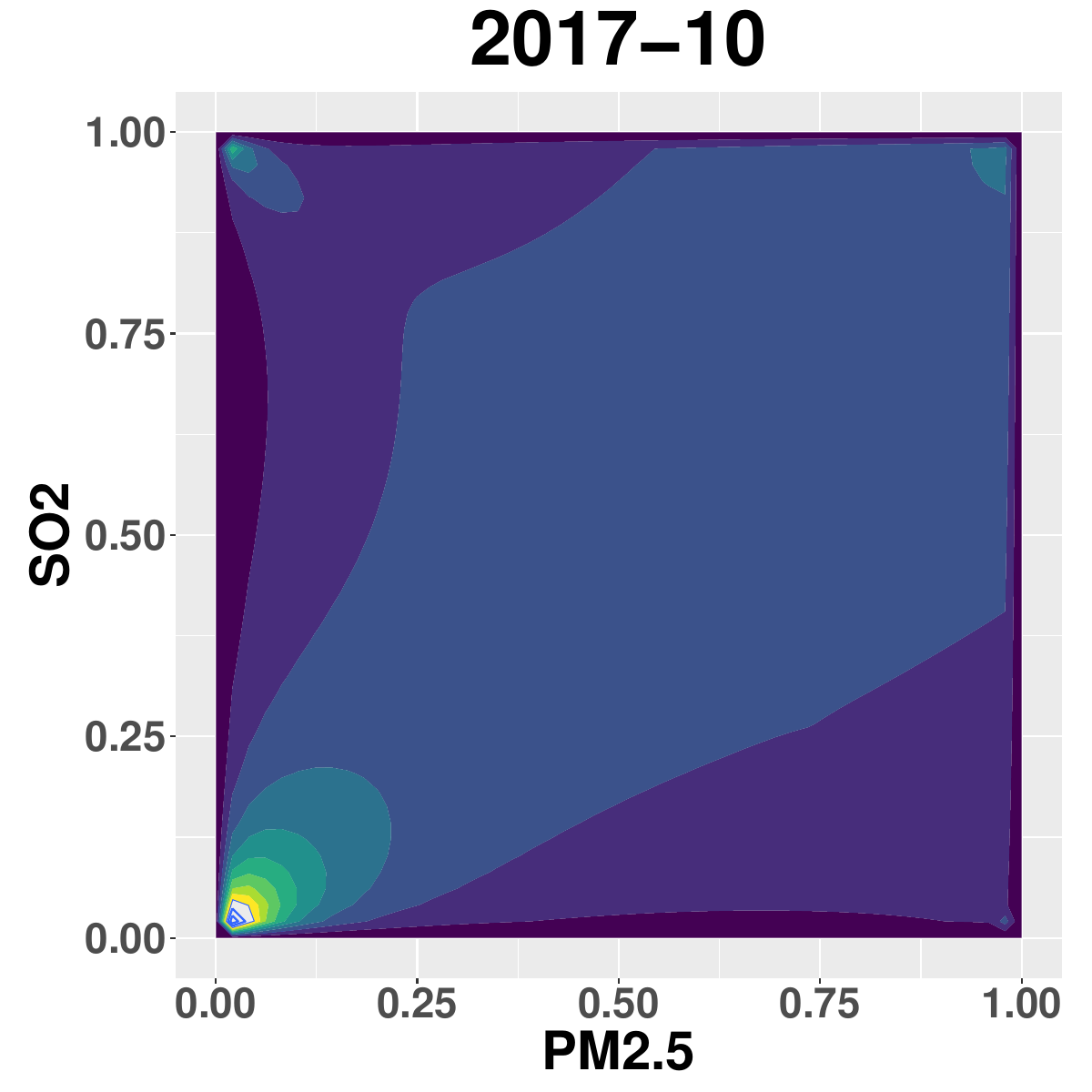}
\includegraphics[scale=0.144]{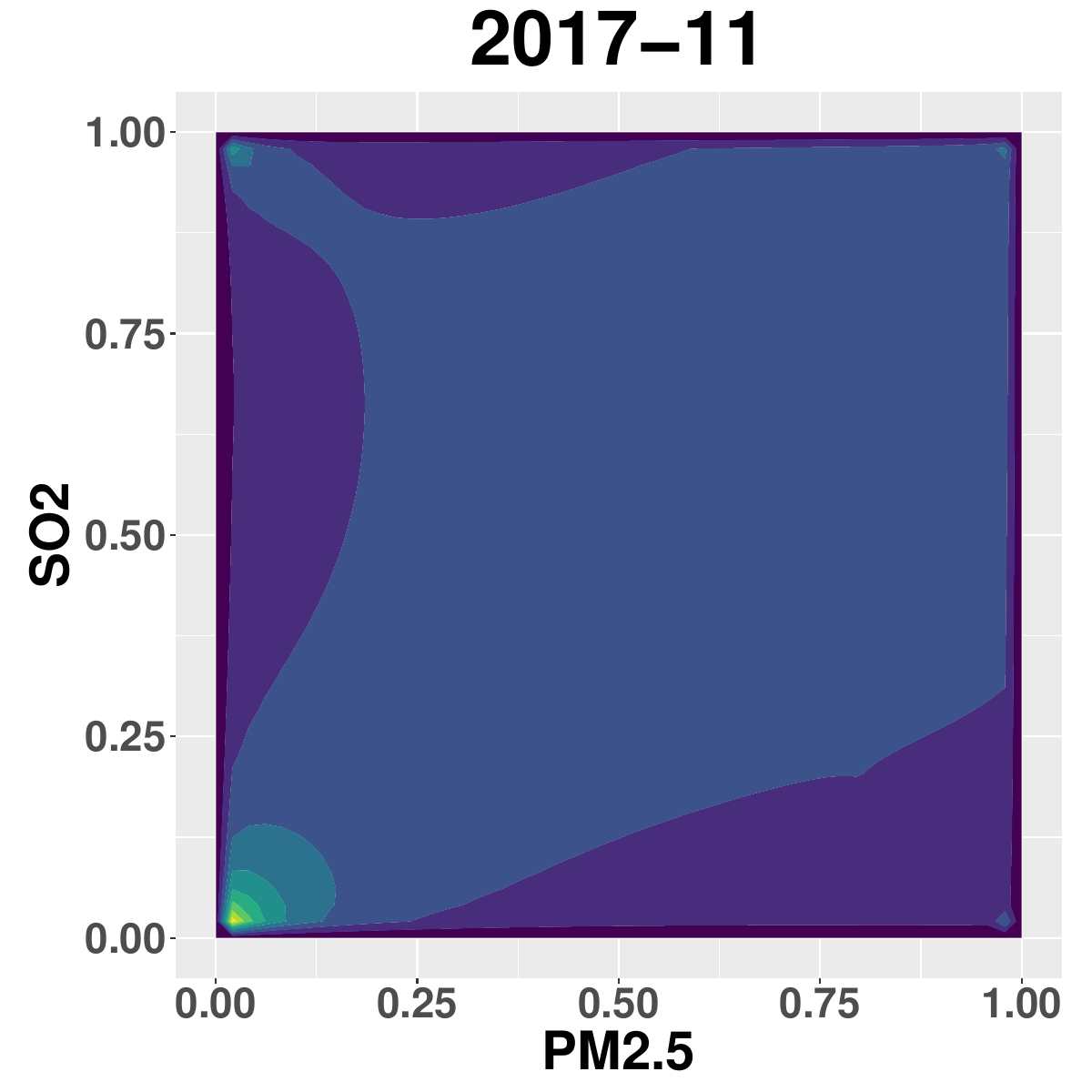}
\includegraphics[scale=0.144]{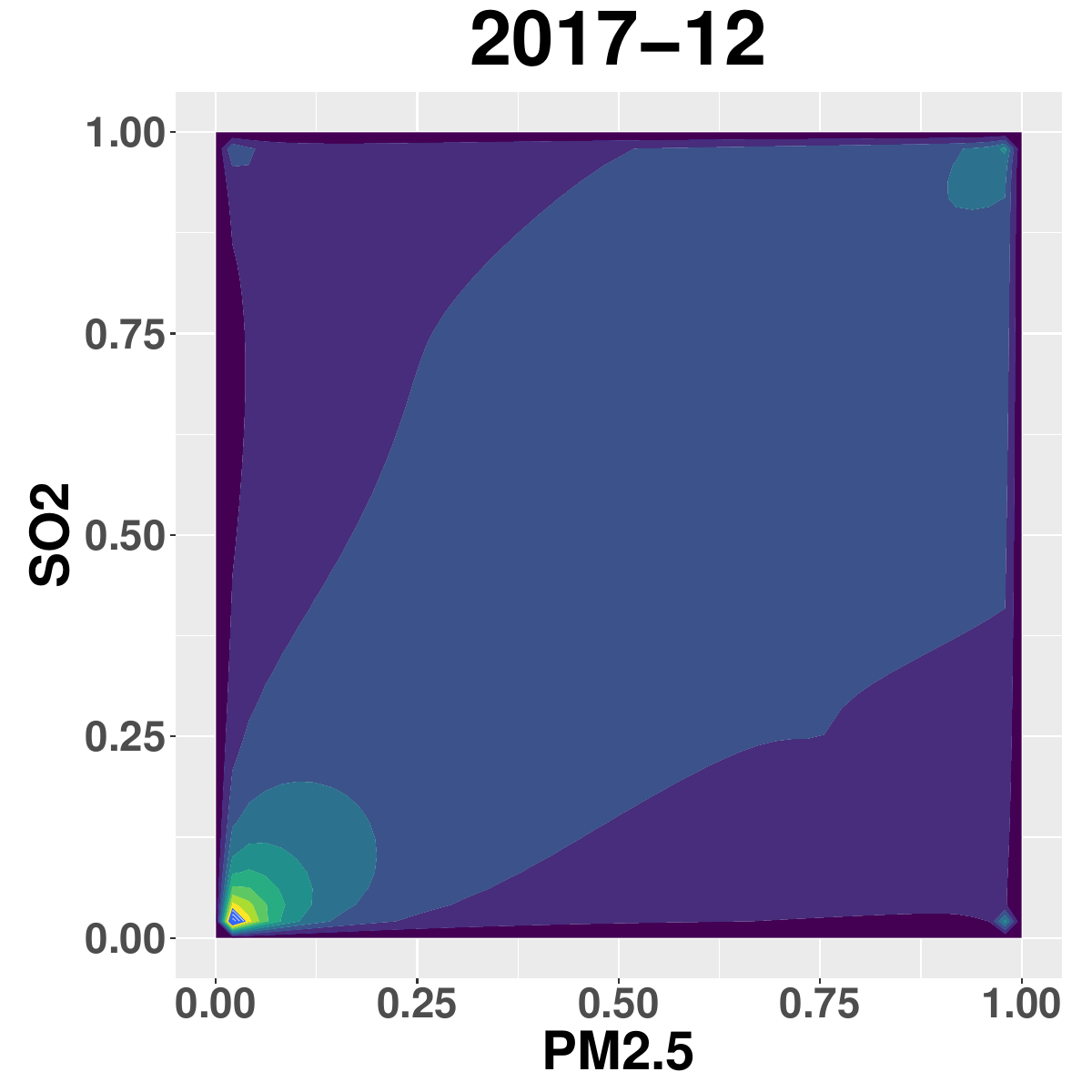}
}
\centerline{
\includegraphics[scale=0.144]{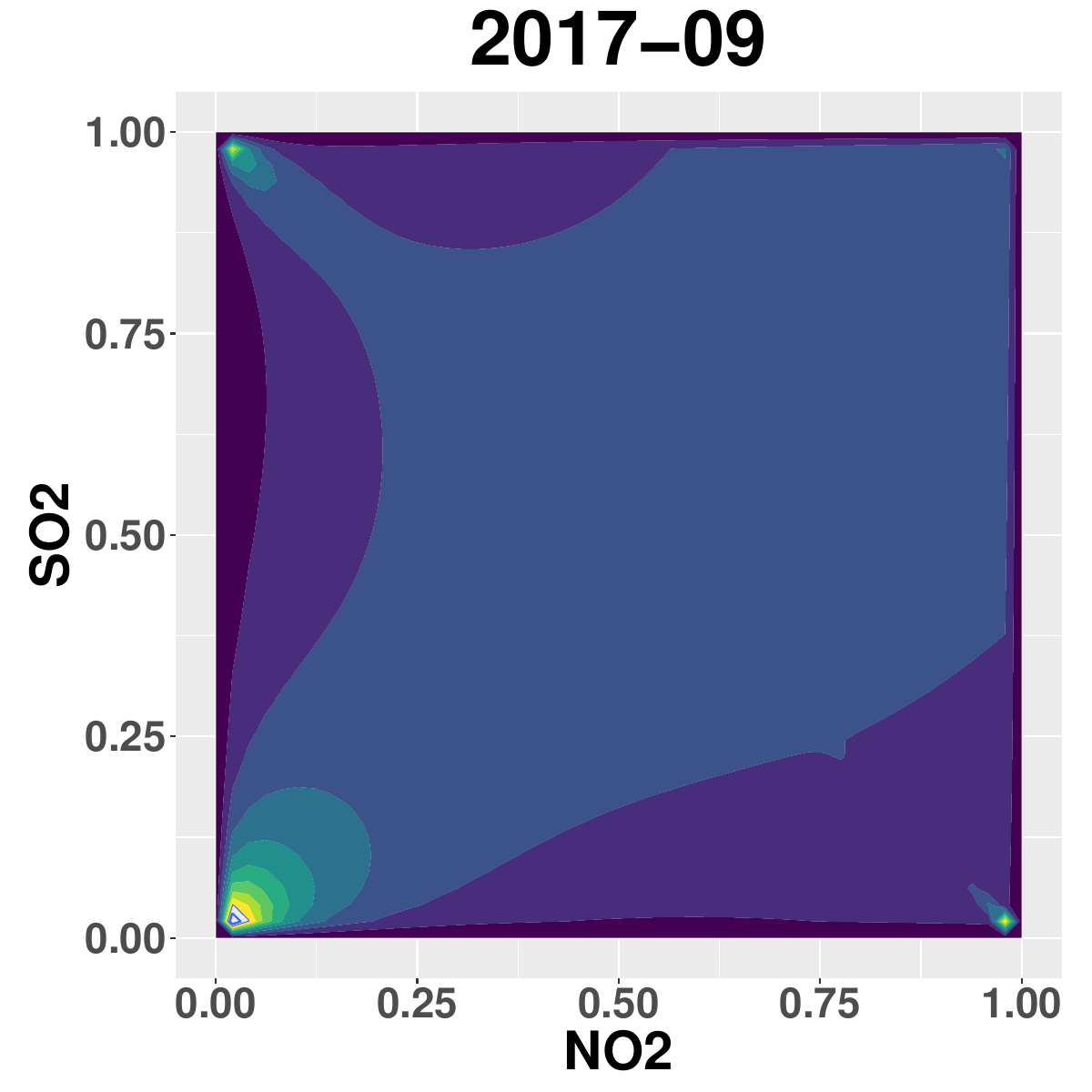}
\includegraphics[scale=0.144]{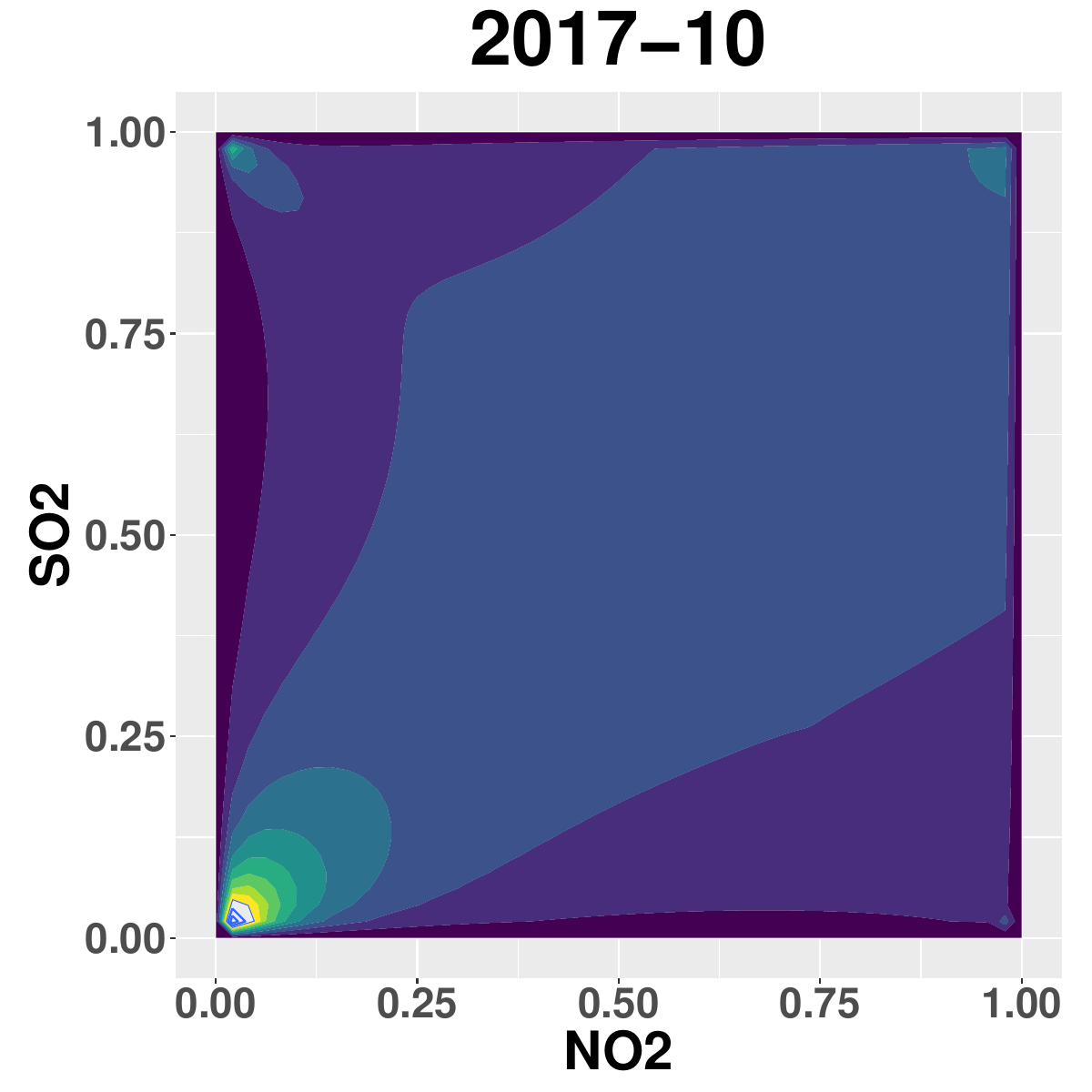}
\includegraphics[scale=0.144]{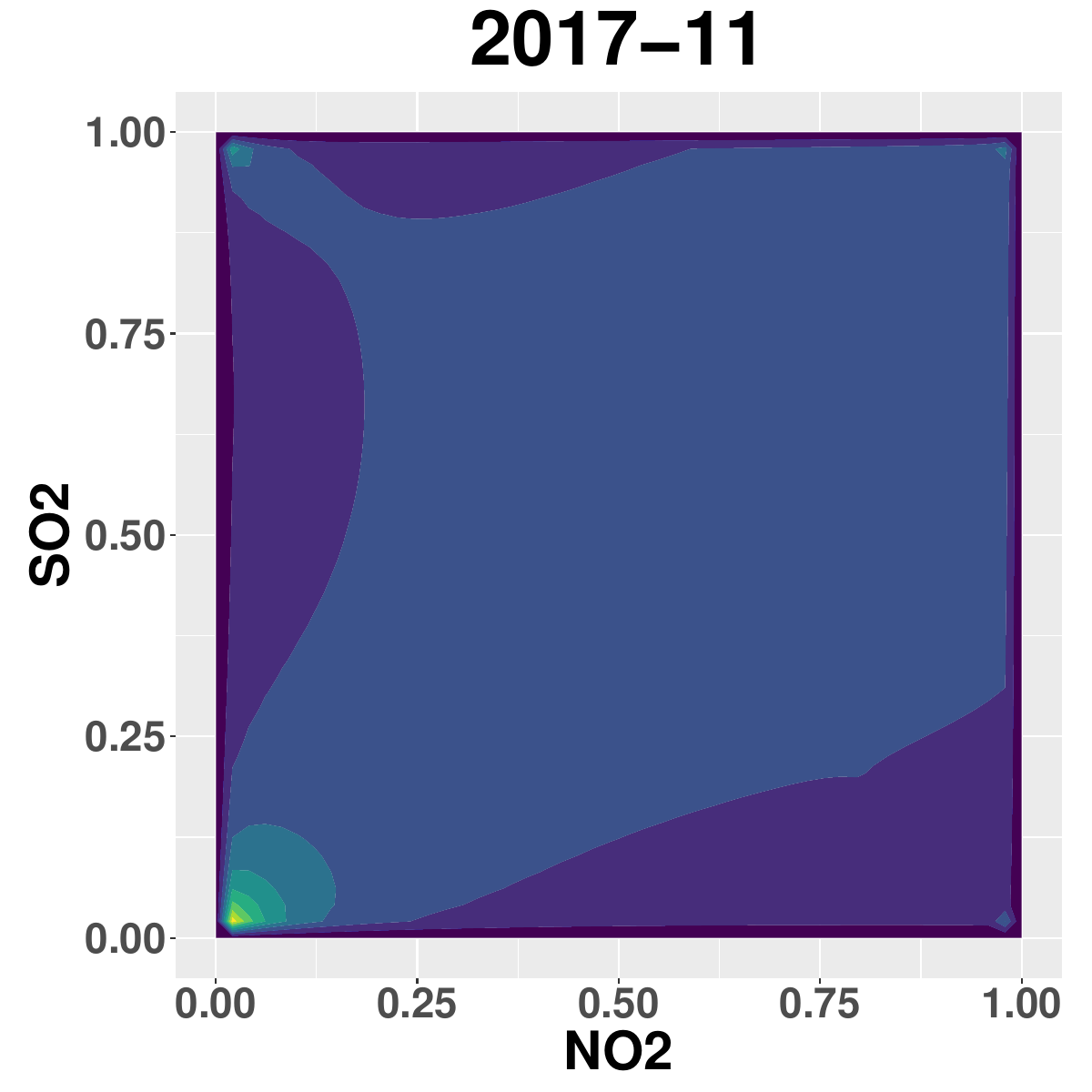}
\includegraphics[scale=0.144]{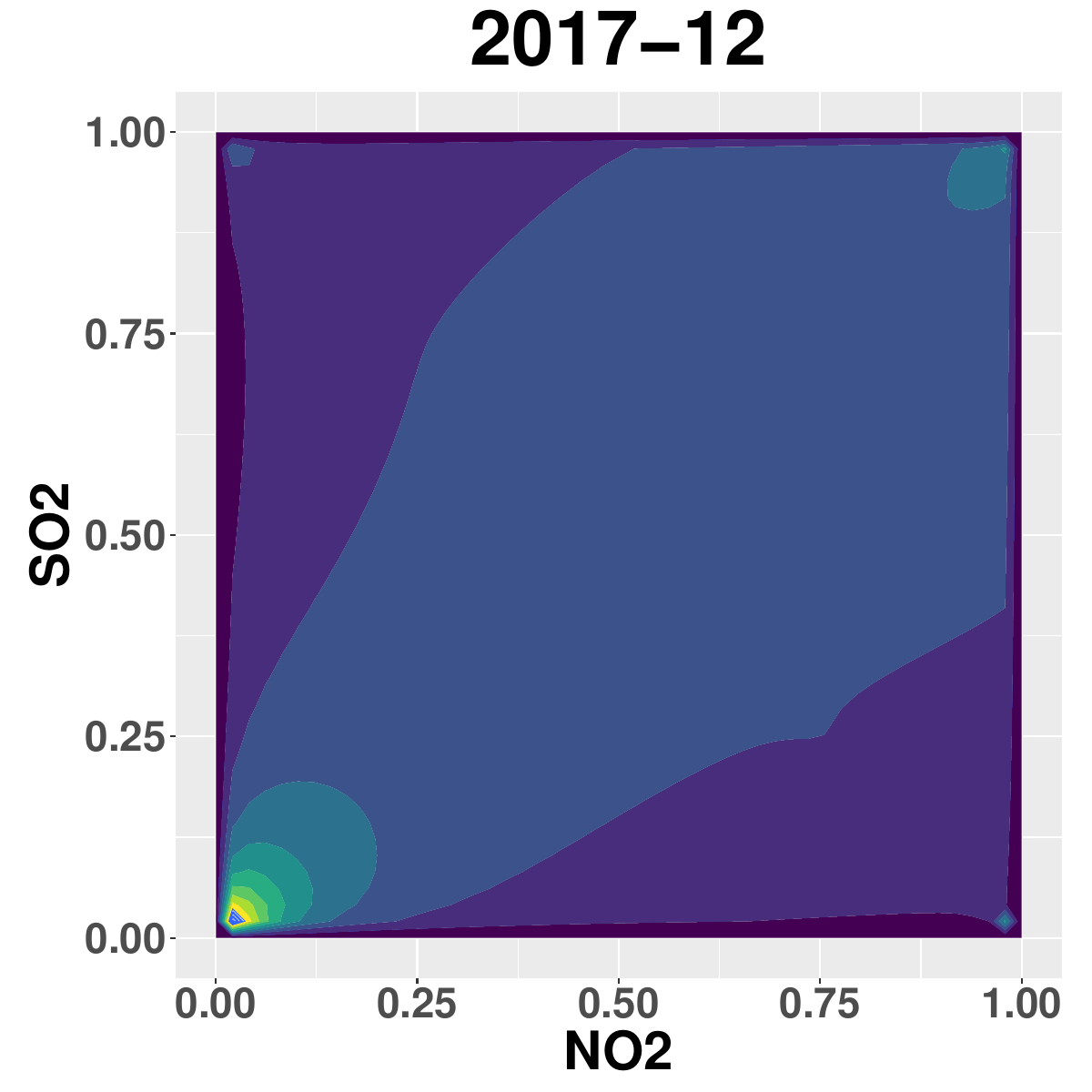}
}
\caption{Multivariate real data. Estimated joint density for months September to December of 2017. $(PM_{2.5},NO_2)$ (top row), $(PM_{2.5},SO_2)$ (middle row) and $(NO_2,SO_2)$ (bottom row).}
\label{fig:joint_predict_real_PM25_NO2_SO2}
\end{figure}

\end{document}